\documentclass[11pt]{article}
\usepackage[margin=1in]{geometry}
\usepackage{amsmath}
\usepackage{amssymb}
\usepackage{amsthm}
\usepackage{amsfonts}
\usepackage{graphicx}
\usepackage{placeins}
\usepackage{enumitem}
\usepackage{booktabs}
\usepackage{multirow}
\usepackage{makecell}
\usepackage{float}
\usepackage{authblk}
\usepackage{hyperref}
\hypersetup{
    colorlinks=true,
    linkcolor=blue,
    citecolor=blue,
    urlcolor=blue
}
\newtheorem{theorem}{Theorem}[section]
\newtheorem{proposition}[theorem]{Proposition}
\newtheorem{lemma}[theorem]{Lemma}
\newtheorem{corollary}[theorem]{Corollary}
\theoremstyle{definition}
\newtheorem{definition}[theorem]{Definition}

\theoremstyle{remark}
\newtheorem{remark}[theorem]{Remark}

\newcommand{\E}{\mathbb{E}}

\title{
Beyond the Skew-Stickiness Ratio: Transport Geometry of Spot-Driven Variance Surface Dynamics
}
\author[1]{Charlie Che\thanks{\texttt{charlie.che@jpmchase.com}}}
\author[2]{Pradeepta Das\thanks{\texttt{pradeepta.das@jpmchase.com}}}
\affil[1]{Quantitative Trading \& Research, JPMorganChase, New York, NY 10017, USA}
\affil[2]{Equity Derivatives Group, JPMorganChase, New York, NY 10017, USA}

\date{}
\begin{document}
\maketitle
\begin{abstract}
We develop a geometric theory of arbitrage-free dynamics for implied
variance surfaces, taking total implied variance $w(k,T)$ as the primitive
state variable and formulating smile dynamics as transport flows
$\partial_u w=v\,\partial_k w$ on the admissible class $\mathcal{A}$ of
static-arbitrage-free surfaces. The classical skew-stickiness ratio is the
zeroth-order coefficient of a Taylor hierarchy that also governs skew and
curvature dynamics, and a local jet transport corollary extends that
hierarchy from the money to every expansion center, making the velocity
profile $v(k)$ estimable from a time series of surfaces. We prove that
explicit Lipschitz and positivity conditions on
$v$ preserve butterfly and calendar arbitrage-freeness locally. Sticky
delta, sticky strike and the skew-stickiness ratio are recovered exactly, as
the constant-velocity transport laws; local and rough volatility are
\emph{not} special cases of the framework, and only their short-maturity
at-the-money velocity is reproduced. We further show that the transport
assumption is not automatic, and give a direct way to test it. Standing at
a fixed strike, the change in its variance as spot moves has two possible
explanations: the smile sliding past that strike at speed $v$, or genuine
extra dynamics with no sliding-smile explanation, which we call the source
$s$. The test is sharpest at the strike where the smile is flattest, since
a sliding smile predicts no change there at all, so any response measured
at that point must be source. We give an estimator that recovers the speed
and the source jointly from the response at each strike, without the
numerical instability of dividing by a near-zero skew, together with an
arbitrage theorem for the source term.

Applying the estimator to five years of SPX implied-volatility data
across seven tenors, we find significant super-skew behavior and reject
self-similar transport at every tenor; a strike-constant velocity is
rejected at the intermediate tenors 2M--12M, with the test inconclusive
at 1M and marginal at 24M. At the two shortest tenors, where the smile
minimum falls within the quoted strike range, a joint test of the
decomposition further shows
that the response is pure transport at the money but carries a genuine
source at the smile minimum. An out-of-sample validation confirms that
the extended model outperforms the scalar SSR on curvature dynamics.
Full estimates and test statistics are in
Section~\ref{sec:spx-empirical}.
\end{abstract}


\section{Introduction}
The dynamics of option implied volatility surfaces under movements of the
underlying remain a central problem in quantitative finance. Sticky strike, sticky delta, local volatility, stochastic volatility, and the skew-stickiness ratio (SSR) of
Bergomi~\cite{Bergomi2004,Bergomi2005,Bergomi2009,BergomiBook} are each stated as an individual rule rather than as an instance of a common structure.
General frameworks for dynamic arbitrage-free surfaces do exist:
Sch\"onbucher~\cite{Schonbucher1999} gives a market model for the full
stochastic implied-volatility surface with no-arbitrage drift conditions;
Carmona and Nadtochiy~\cite{CarmonaNadtochiy2009} characterize admissible
drifts for call-price surfaces and, in~\cite{CarmonaNadtochiy2011}, develop a
tangent-model framework for dynamic calibration; Schweizer and
Wissel~\cite{SchweizerWissel2008} establish absence-of-arbitrage conditions
for implied-volatility term structures; and Durrleman~\cite{Durrleman2010}
derives PDEs for the motion of the surface under spot displacements within a
local-volatility model, in effect obtaining the transport velocity $v(k)$ as a
model-specific output of Dupire's formula. On the SSR itself, Friz and
Gatheral~\cite{FrizGatheral2025} compute the ratio as a
quadratic-covariation projection depending on the forward-variance curve, and
Fukasawa~\cite{Fukasawa2026} gives it an It\^o--Wentzell treatment; both
concern the at-the-money value of the object studied here.

This paper takes the velocity field $v$ in $\partial_u w=v\,\partial_k w$, with $u=\log F$ and $w$ the total implied variance, as the primitive, and
asks three questions about it: when the representation has empirical content,
when it preserves static arbitrage, and what the data say. The contributions
follow that order.

\emph{1.~A testable decomposition, and a test that rejects it.}
Any smooth dynamics can be written as $v\,\partial_k w$ wherever
$\partial_k w\neq 0$, so the transport form alone restricts nothing. The
hypothesis with content is that the response lies in the transport subspace
$\{v\,\partial_k w\}$ for a \emph{bounded} $v$. Writing
$\partial_u w=v\,\partial_k w+s$ (Section~\ref{subsec:transport-and-source}),
pure transport is $s\equiv 0$ and $s$ is the \emph{source} relative to that
restricted class. The decomposition is sharpest where transport is powerless:
at a smile extremum $\partial_k w=0$ a bounded velocity predicts no response
whatever its value, so any response measured there is source. On SPX the source
is indistinguishable from zero at the money at 1M and 3M and decisively
nonzero at each day's tracked smile minimum ($t=-6.5$ and $-7.6$), which
rejects bounded pure transport at those tenors. Because the test is run where
the skew vanishes, it uses no estimated skew as a regressor. A joint estimator
for $(v,s)$ that never divides by the skew
(Proposition~\ref{prop:identification-without-inversion}) supplies the
corresponding profiles.

\emph{2.~Local arbitrage-free transport.}
Which velocity fields keep the flow inside the admissible class
$\mathcal{A}^\circ$ of static-arbitrage-free surfaces? Under explicit
regularity and boundedness conditions on $v$ ((V1)--(V3)) and on the initial
surface ((W)), the flow preserves butterfly and calendar arbitrage-freeness on
a compact strike window for sufficiently small spot perturbations
(Theorem~\ref{thm:local-arb-preserving}), with an admissibility radius given
explicitly in terms of those constants. The conditions are verifiable and the
proof constructive, so the result applies directly to a parametric velocity
field without a post-hoc arbitrage check. Global preservation is not claimed.
A companion statement for the additive source
(Theorem~\ref{thm:admissible-source}) completes the analysis for the full
decomposition.

\emph{3.~The transport organization and its identification.}
Static arbitrage is a property of total variance $w=\sigma^2T$, which is why
we take $w$ rather than implied volatility as the state variable. In this
language the classical regimes are particular velocities: sticky delta $v\equiv 0$, sticky strike $v\equiv 1$, the SSR $v\equiv\beta$. Proposition~\ref{prop:classical-pure-transport} shows that these are
\emph{exactly} the constant-velocity pure transports. That claim is narrower
than it may look, since those three rules are transport laws by construction:
each one already prescribes the response as a multiple of the skew. Local volatility and rough Bergomi are \emph{not}
special cases of the framework. Their velocity is a model-specific output,
nonlocal in $w$, so it falls outside the admissibility theorem
(Remark~\ref{rem:nonlocal-v}); and a generative model of that kind has no
reason to satisfy the constraint pure transport imposes at a smile extremum,
so it generically carries a source rather than lying in the transport
subspace (Remark~\ref{rem:lv-asymptotic}). What the framework reproduces for
them is the short-maturity at-the-money velocity, $v\to2$ and $\mathbb E[v]\to H+\tfrac{3}{2}$, and nothing more.
Differentiating the transport equation gives a hierarchy for the
variance jets $a_n=\partial_k^n w$ governed by the velocity jets
$\beta=v(0)$, $\eta=\partial_k v(0)$, $\zeta=\partial_{kk}v(0)$
(Theorem~\ref{thm:jet-transport}), and it holds at every expansion center
rather than at the money alone, which is what makes the profile $v(k)$ estimable: jointly with $s$, as contribution~1 requires, rather than by
pointwise inversion of the transport equation. The hierarchy is the Leibniz
rule, and we treat it as organizing
machinery rather than as the mathematical novelty: its value is that it turns
the scalar SSR into the first term of a sequence of estimands. Applied to five
years of SPX surfaces it rejects self-similar transport ($\eta=\zeta=0$) at
every tenor.

Section~\ref{sec:transport-representation} sets up the transport
representation and Section~\ref{sec:self-similar-transport} derives the
Bergomi ride formula as its constant-velocity solution.
Section~\ref{sec:general-transport} develops the general velocity field, the
transport--source decomposition, the jet hierarchy and the identification
results. Section~\ref{sec:arb-preservation} proves the admissibility
theorems, Section~\ref{sec:classical-special-cases} places the classical
models, Section~\ref{sec:spx-empirical} reports the SPX evidence
(Experiments 1--4, and the test of pure transport in Experiment 5),
Section~\ref{sec:mc-validation} validates the estimator on synthetic data, and
Section~\ref{sec:extensions} collects extensions.

\FloatBarrier
\section{The Geometry of Arbitrage-Free Variance Surfaces}
We identify the natural mathematical object upon which smile dynamics should
act.
Classical descriptions of smile dynamics are formulated directly in terms of
implied volatility. While this is natural from a market perspective, it is not
the most fundamental mathematical representation. Static arbitrage conditions
are not naturally expressed in implied volatility but rather in total implied
variance. Consequently, we shall regard total variance as the primitive state
variable and formulate smile dynamics as transport of variance surfaces.
\subsection{Total Variance as the Fundamental State Variable}
Let $k=\log(K/F)$ denote log-moneyness and let $\sigma(k,T)$ be the Black implied volatility.
Rather than working directly with implied volatility, we define the total
implied variance
\[
w(k,T)
:=
\sigma^2(k,T)T.
\]
Throughout this paper we regard
\[
w:\mathbb R\times\mathbb R_+
\longrightarrow
\mathbb R_+
\]
as the fundamental state variable. Three reasons: total variance scales
naturally in maturity and admits simpler asymptotics than implied volatility;
essentially every practical arbitrage-free parameterization (SVI, eSSVI and
relatives) is written in terms of it; and, most importantly, the absence of
static arbitrage is a property of the variance surface, so it is the natural
object on which to formulate dynamics.
\subsection{The Admissible Class and Admissible Directions}
Not every smooth $w(k,T)$ corresponds to option prices: admissible surfaces
satisfy positivity of the induced risk-neutral density (butterfly), monotonicity
in maturity (calendar), and asymptotic growth conditions ensuring finite
option prices. We write $\mathcal{A}$ for the collection of surfaces free of
static arbitrage, deferring the explicit analytic form to
Section~\ref{subsec:admissibility-explicit}. The objective throughout is not
merely to describe motions of variance surfaces but to characterize dynamics
remaining within $\mathcal{A}$.

Perturbing a surface, $w_\varepsilon=w+\varepsilon h+o(\varepsilon)$, need not preserve admissibility, since $h$ may immediately violate butterfly or calendar arbitrage. This motivates the following.
\begin{definition}[Admissible direction]
\label{def:admissible-perturbation}
A perturbation $h$ is an \emph{admissible direction} at $w\in\mathcal{A}$ if
there exists $\varepsilon_0>0$ such that
$w+\varepsilon h\in\mathcal{A}$ for every $|\varepsilon|<\varepsilon_0$.
\end{definition}
This is a two-sided condition on finite perturbations, and not merely the
requirement that $h$ satisfy the arbitrage constraints linearized at $w$: it
asks that an entire segment $\{w+\varepsilon h:|\varepsilon|<\varepsilon_0\}$
through $w$ remain admissible. The two requirements coincide when the
constraints hold strictly at $w$, which is the case treated in
Section~\ref{subsec:geometric-viewpoint} and the only case used below. They
differ when a constraint is active: there a direction satisfying the
linearized condition need not generate an admissible segment, and
Theorem~\ref{thm:admissible-source}\,(iii) exhibits that failure explicitly.
\subsection{Explicit Form of the Admissibility Conditions}
\label{subsec:admissibility-explicit}
For the purposes of the dynamic theory we shall need explicit analytic forms
of the static no-arbitrage conditions.

\begin{definition}[Strictly admissible variance surface]
\label{def:admissible-explicit}
A function $w\in C^{2,1}(\mathbb R\times\mathbb R_+)$ is \emph{strictly admissible}
if
\begin{align}
    w(k,T) &> 0 \quad\text{for all }(k,T), \tag{2.1}\label{eq:adm-positivity}\\
    \partial_T w(k,T) &> 0 \quad\text{for all }(k,T)\quad\text{(no calendar arbitrage)},
        \tag{2.2}\label{eq:adm-calendar}\\
    g[w](k,T) &> 0 \quad\text{for all }(k,T)\quad\text{(no butterfly arbitrage)},
        \tag{2.3}\label{eq:adm-butterfly}
\end{align}
where the Gatheral density functional is
\begin{equation}
    g[w](k,T)
    :=
    \left(1-\frac{k\,\partial_k w}{2w}\right)^{\!2}
    -\frac{(\partial_k w)^2}{4}\!\left(\frac{1}{w}+\frac{1}{4}\right)
    +\frac{\partial_{kk} w}{2}.
    \tag{2.4}\label{eq:g-functional}
\end{equation}
The interior of the admissible class is $\mathcal{A}^\circ = \{w\in C^{2,1}:
\eqref{eq:adm-positivity},\eqref{eq:adm-calendar},\eqref{eq:adm-butterfly}\text{ hold strictly}\}$.
\end{definition}

\begin{remark}
The condition $g[w]>0$ is equivalent to positivity of the risk-neutral density
induced by the call-price surface $C(k,T)=\mathrm{BS}(k,w(k,T))$, which
by the Breeden--Litzenberger formula~\cite{Gatheral2006} is equivalent
to $C_{KK}\ge 0$.
See Gatheral~\cite{Gatheral2006} and Roper~\cite{Roper2010} for the
classical derivation.
\end{remark}

\begin{remark}[Wing growth conditions]
\label{rem:wing-conditions}
Definition~\ref{def:admissible-explicit} is a characterization of
interior admissibility conditions on any compact $k$-domain.
A complete characterization of arbitrage-free variance surfaces on
all of $\mathbb{R}$ additionally requires asymptotic growth conditions
on $w(k,T)$ as $|k|\to\infty$ to ensure finite moments of the
risk-neutral distribution; these are characterized by
Lee~\cite{Lee2004}, who shows $\limsup_{|k|\to\infty}w(k,T)/|k|\le 2$
is necessary for finite call prices, with finer tail conditions
governing higher moments.
Throughout this paper the transport dynamics are established on
compact moneyness intervals; we do not impose or analyze wing-growth conditions.
\end{remark}
\subsection{The Ambient Space and the Geometric Vocabulary}
\label{subsec:geometric-viewpoint}
We now fix the function-space setting used in the rest of the paper, and
separate it from the geometric vocabulary used to describe it.

Fix a compact window $Q:=[k_-,k_+]\times[T_1,T_2]$ with $T_1>0$ and let
\[
    \mathcal{X}:=C^{2,1}(Q)
\]
be the ambient Banach space of candidate variance surfaces, with its usual
norm. Write $\mathcal{A}^\circ_Q\subset\mathcal{X}$ for the set of surfaces
satisfying \eqref{eq:adm-positivity}--\eqref{eq:adm-butterfly} strictly on
$Q$. Because these are strict inequalities on the continuous functionals $w$,
$\partial_T w$ and $g[w]$, the set $\mathcal{A}^\circ_Q$ is an \emph{open}
subset of $\mathcal{X}$. Two consequences are all the structure this paper
requires. First, the admissible directions at any $w\in\mathcal{A}^\circ_Q$
in the sense of Definition~\ref{def:admissible-perturbation} are all of
$\mathcal{X}$. Second, an evolution equation $\partial_u w=V(w)$ with
$\mathcal{X}$-valued right-hand side is an ordinary differential equation in
a Banach space, posed on an open set.

\begin{definition}[Transport field]
\label{def:transport-field}
A \emph{transport field} on $\mathcal{A}^\circ_Q$ is a map
$V:\mathcal{A}^\circ_Q\to\mathcal{X}$ assigning to each admissible surface
$w$ the direction $V(w)$ in which it moves under a unit log-forward
displacement.
\end{definition}

We shall call $V$ a \emph{vector field} and the solution of
$\partial_u w=V(w)$ its \emph{flow}. For an $\mathcal{X}$-valued map on an
open subset of $\mathcal{X}$ this is standard terminology, and it is the whole
of what the geometric language in this paper asserts: every result below is
proved by PDE and ODE estimates, and none uses more than
Definition~\ref{def:transport-field}. In particular we do not construct a
manifold structure on the admissible class $\mathcal{A}$ over the noncompact
domain $\mathbb{R}\times\mathbb{R}_+$, and we do not treat a tangent bundle
$T\mathcal{A}$ or a space of sections $\Gamma(T\mathcal{A})$ as established
objects; where such language appears it is interpretation and not hypothesis.
Over the noncompact domain openness is not automatic, since the uniform lower
bounds defining $\mathcal{A}^\circ$ may fail in the tails.
Section~\ref{subsec:rigorous-geometry} describes the weighted topologies under
which the construction is expected to extend, with the
Filipovi\'c--Teichmann geometry of yield-curve
spaces~\cite{FilipovicTeichmann2004} as the precedent; the variance-surface
setting differs there in that admissibility is nonlinear, through the Gatheral
density $g[w]>0$, whereas yield-curve admissibility is a linear cone
condition.

The principal objective of this paper is to characterize transport vector
fields generated by movements of the underlying and to study the flows induced
by such vector fields.
The classical skew-stickiness ratio will emerge as the first non-trivial
transport law within this geometric framework.

\FloatBarrier
\section{Transport Representation of Smile Dynamics}
\label{sec:transport-representation}
We now describe what it means for a static-arbitrage-free total variance
surface to be transported by a movement of the underlying.
The basic distinction is between two effects. The first is a purely mechanical
coordinate effect: when the forward changes, a fixed strike corresponds to a
different log-moneyness. The second is an intrinsic deformation of the surface
itself. Classical notions such as sticky strike, sticky delta and skew
stickiness can all be understood as different prescriptions for this
decomposition.
Throughout this section we write $k=\log(K/F)$,
$w(k,T;F)=T\sigma^2(F e^k,F,T)$,
where the dependence on $F$ records how the total variance surface changes
as the forward level changes. We suppress \(T\) when the maturity is fixed.
\subsection{Fixed and Floating Coordinates}
Let \(F\) denote the current forward. There are two natural ways to observe
the surface: the floating representation holds log-moneyness
\(k=\log(K/F)\) fixed, which is natural for studying the shape of the smile
relative to the forward, and the fixed-strike representation holds \(K\)
fixed. Since \(k=\log K-\log F\),
\[
    \left.\partial_{\log F}\right|_{K}
    =
    \left.\partial_{\log F}\right|_{k}
    -
    \partial_k ,
    \tag{3.1}
\]
an elementary identity that is the source of the floating-versus-fixed
decomposition throughout smile
dynamics~\cite{Bergomi2004,Bergomi2009,BergomiBook}. Writing
$D^{\mathrm{fl}}w:=\partial_{\log F}w|_k$ and
$D^{\mathrm{fx}}w:=\partial_{\log F}w|_K$ for the floating and fixed-strike
variations, \((3.1)\) reads
\[
    D^{\mathrm{fx}} w
    =
    D^{\mathrm{fl}} w-\partial_k w,
    \tag{3.2}
\]
in which \(\partial_k w\) is the \emph{mechanical slide}: the change obtained
solely by reading a different point of the same smile as the forward moves.

A smile dynamics is more than this coordinate change, since it must also
specify how the surface itself moves. Let $u=\log F$ and let
$Dw=\partial_u w$ be an infinitesimal dynamics. The floating representation
identifies $Dw$ with the intrinsic change at fixed log-moneyness and the
fixed-strike representation subtracts the slide, so the same dynamics has the
two descriptions $D^{\mathrm{fl}}w=Dw$ and
$D^{\mathrm{fx}}w=Dw-\partial_k w$, the variance-surface analog of the classical relation between floating-ATM and fixed-strike volatility moves. A
model of smile dynamics is therefore a prescription for $D$: sticky delta is
$D=0$, sticky strike is $D=\partial_k$, and skew stickiness is intermediate.
\subsection{Transport Operators}
We now formalize this observation.
\begin{definition}[Transport operator]
A transport operator is a transport field in the sense of
Definition~\ref{def:transport-field}, written
\[
    \mathcal{D}:\mathcal{A}^\circ_Q\to\mathcal{X},
    \qquad
    w\mapsto \mathcal{D}_w,
\]
where \(\mathcal{D}_w\) is the direction of motion assigned to \(w\).
The induced transport flow is the evolution
\[
    \frac{d}{du}w_u=\mathcal D_{w_u},
    \qquad u=\log F.
    \tag{3.5}
\]
\end{definition}
At this level of generality, \(\mathcal{D}_w\) may contain both intrinsic smile
deformation and coordinate transport. The classical stickiness rules correspond
to simple choices of \(\mathcal D\). At a fixed maturity,
\[
    \begin{array}{lll}
    \text{sticky delta:}  & \mathcal{D}_w=0,             & D^{\mathrm{fx}}w=-\partial_k w, \\[3pt]
    \text{sticky strike:} & \mathcal{D}_w=\partial_k w,  & D^{\mathrm{fx}}w=0.
    \end{array}
    \tag{3.6}
\]
Thus sticky delta means the floating smile is frozen, while sticky strike means
the fixed-strike surface is frozen.

\begin{definition}[Order of a transport operator]
\label{def:order-transport}
A transport operator $\mathcal{D}:\mathcal{A}^\circ_Q\to\mathcal{X}$ is of
\emph{order $N$} if at each $w\in\mathcal{A}^\circ$ it takes the form
\begin{equation}
    \mathcal{D}_w
    =
    \sum_{j=1}^{N} a_j(k,u,T;w)\,\partial_k^j w,
    \quad a_N\not\equiv 0.
    \tag{3.7}\label{eq:order-N-op}
\end{equation}
The first-order case $N=1$ with $a_1=v$ is the transport velocity field developed in the following two sections.
\end{definition}
The simplest non-trivial transport operator is the first-order self-similar transport law. Rather than introducing it axiomatically, we derive it in the next section together with its characteristic solution, ride formula, and connection to the classical skew-stickiness ratio

\FloatBarrier
\section{Self-Similar Transport and the Skew-Stickiness Ratio}
\label{sec:self-similar-transport}
We now show that the classical skew-stickiness ratio is not an independent
ansatz but the first non-trivial instance of a transport law. The self-similar
surface, the Bergomi ride formula~\cite{Bergomi2004,Bergomi2009,BergomiBook},
and the floating-versus-fixed
decomposition~\cite{Bergomi2004,Bergomi2005,Bergomi2009,BergomiBook}
all follow from a single first-order transport equation.
Throughout this section we fix a maturity \(T\) and suppress it from the
notation when no confusion can arise. Let \(u=\log F\) and
\(k=\log(K/F)\). The total variance surface is written as $w(k,u):=T\sigma^2(F e^k,F,T)$.
The floating coordinate is \(k\), while the fixed-strike coordinate is
\(\log K=k+u\).
\subsection{Self-Similar Transport}
\begin{definition}[Self-similar transport]
Let \(\beta\in\mathbb R\). We say that a variance surface satisfies
self-similar transport with coefficient \(\beta\) if its floating-coordinate
evolution is governed by
\[
    \partial_u w(k,u)=\beta\,\partial_k w(k,u).
    \tag{4.1}
\]
The parameter \(\beta\) is called the first-order transport coefficient. In
implied-volatility coordinates it corresponds to the classical
skew-stickiness ratio.
\end{definition}
The cases \(\beta=0\) and \(\beta=1\) recover sticky delta and sticky strike,
respectively. Indeed, \(\beta=0\) freezes the surface in floating
log-moneyness, while \(\beta=1\), combined with the identity
\(\partial_u|_K=\partial_u|_k-\partial_k\), freezes the surface at fixed strike.
\subsection{The Self-Similar Transport Theorem}
\begin{theorem}[Self-similar transport theorem]
\label{thm:self-similar-transport}
Let \(w\in C^1(\mathbb R\times I)\), where \(I\subset\mathbb R\) is an interval
of forward log-levels. The following are equivalent.
\begin{enumerate}[label=(\roman*)]
\item
\(w\) satisfies the transport equation
\[
    \partial_u w=\beta\,\partial_k w.
    \tag{4.2}
\]
\item
There exists a function \(\Phi\) such that
\[
    w(k,u)=\Phi(k+\beta u).
    \tag{4.3}
\]
\item
In strike-forward coordinates, there exists a function \(\Psi\) such that
\[
    w(K,F)=\Psi\!\left(\log K+(\beta-1)\log F\right).
    \tag{4.4}
\]
Equivalently, writing \(m_\beta:=KF^{\beta-1}\), there exists a function
\(\widetilde \Psi\) such that
\[
    w(K,F)=\widetilde\Psi(m_\beta).
    \tag{4.5}
\]
\end{enumerate}
Consequently, if the forward moves from \(F\) to \(F'\), the transported
surface satisfies the exact ride formula
\[
    w_{\mathrm{new}}(K)
    =
    w_{\mathrm{old}}
    \!\left(
        K\left(\frac{F}{F'}\right)^{1-\beta}
    \right).
    \tag{4.6}
\]
The same formula holds for implied volatility:
\[
    \sigma_{\mathrm{new}}(K,T)
    =
    \sigma_{\mathrm{old}}
    \!\left(
        K\left(\frac{F}{F'}\right)^{1-\beta},T
    \right).
    \tag{4.7}
\]
\end{theorem}
\begin{proof}
Assume first that \(w\) satisfies \((4.2)\). Consider the characteristic curves
\((k(u),u)\) along which \(w\) is constant. We require
\[
    \frac{d}{du}w(k(u),u)
    =
    \partial_u w+k'(u)\partial_k w=0.
\]
Using \((4.2)\), this becomes $(\beta+k'(u))\partial_k w=0$,
so the characteristics satisfy $k'(u)=-\beta$. Hence
\(k(u)+\beta u\) is constant along characteristics, and therefore
\(w(k,u)=\Phi(k+\beta u)\) for some function \(\Phi\). This proves
\((i)\Rightarrow(ii)\).
Conversely, if \(w(k,u)=\Phi(k+\beta u)\), then
\[
    \partial_u w=\beta\Phi'(k+\beta u),
    \qquad
    \partial_k w=\Phi'(k+\beta u),
\]
and hence \(\partial_u w=\beta\partial_k w\). Thus \((i)\) and \((ii)\) are
equivalent.
Since \(k=\log K-\log F\) and \(u=\log F\), the invariant coordinate in
\((4.3)\) can be written as
\[
    k+\beta u
    =
    \log K+(\beta-1)\log F.
\]
This gives \((4.4)\). Exponentiating the invariant coordinate gives
\(m_\beta=KF^{\beta-1}\), which yields \((4.5)\). Thus \((ii)\), \((iii)\),
and \((4.5)\) are equivalent.
It remains to derive the ride formula. Let the old forward be \(F\) and the
new forward be \(F'\). The self-similar representation gives
\[
    w_{\mathrm{old}}(K')
    =
    \widetilde\Psi(K'F^{\beta-1}),
    \qquad
    w_{\mathrm{new}}(K)
    =
    \widetilde\Psi(K(F')^{\beta-1}).
\]
Choose \(K'\) so that \(K'F^{\beta-1}=K(F')^{\beta-1}\). Then
\[
    K'=K\left(\frac{F'}{F}\right)^{\beta-1}
      =K\left(\frac{F}{F'}\right)^{1-\beta}.
\]
Substitution gives \((4.6)\). Since \(w=T\sigma^2\) and the maturity is fixed,
the same transport identity holds for \(\sigma\), proving \((4.7)\).
\end{proof}

\begin{remark}[Attribution of the ride formula]
Equations \((4.6)\)--\((4.7)\) recover the forward-smile ride formula
derived by Balland~\cite{Balland2002} and Bergomi~\cite{Bergomi2004,Bergomi2009,BergomiBook}
via different arguments. The present derivation makes the characteristic
structure explicit: \((4.6)\)--\((4.7)\) are the unique solution of the
first-order PDE \((4.2)\), and the invariant moneyness coordinate
$m_\beta := KF^{\beta-1}$ is precisely the characteristic variable.
This derivation reverses the usual logic: the ride formula is not imposed
as an external smile parametrization but emerges as the characteristic
solution of the constant-velocity transport equation.
\end{remark}

\subsection{Fixed-Strike Form and the Classical SSR}
Since \(\partial_u|_K=\partial_u|_k-\partial_k\), the transport equation
\(\partial_u w=\beta\,\partial_k w\) is equivalent to
\[
    \left.\partial_u w\right|_K
    =
    (\beta-1)\,\partial_k w,
    \tag{4.8}
\]
so \(\beta\) is the floating transport coefficient and \(\beta-1\) its
fixed-strike counterpart. The shift by one carries no modeling content: it
comes entirely from the change of coordinates, and measures the smile sliding
past a fixed strike as the forward moves. With the ATM variance skew
$S_w:=\partial_k w(0,\log F,T)$, the two ATM moves differ by exactly
$S_w\,\Delta\log F+o(\Delta\log F)$, which is that slide.

The classical SSR is stated in implied-volatility units, and the translation
is immediate. Writing $\sigma(k,u,T):=\sigma(Fe^k,F,T)$ and
$S_\sigma:=\partial_k\sigma(0,u,T)$, the relations
$\partial_k w=2T\sigma\,\partial_k\sigma$ and
$\partial_u w=2T\sigma\,\partial_u\sigma$ give, wherever $\sigma>0$,
\[
    \partial_u w=\beta\,\partial_k w
    \iff
    \partial_u\sigma=\beta\,\partial_k\sigma,
    \tag{4.9}
\]
so $\beta$ is the same number whether computed from $w$-jets or
$\sigma$-jets. At the money this is
$\Delta\sigma_{\mathrm{ATM}}=\beta S_\sigma\,\Delta\log F+o(\Delta\log F)$,
the classical floating-ATM SSR relation; at fixed strike it is
$\partial_u\sigma|_K=(\beta-1)\partial_k\sigma$, equivalently
\[
    \left.\partial_F\sigma(K,F,T)\right|_K
    =
    (\beta-1)\frac{K}{F}\,\partial_K\sigma(K,F,T),
    \tag{4.10}
\]
the standard fixed-strike riding formula.
\subsection{Covariation and Regression Characterizations}
Let \(F_t\) be the forward process and \(\sigma^*_t=\sigma(F_t,F_t,T)\) the
floating at-the-money implied volatility. Under self-similar transport
$d\sigma^*_t=\beta S_{\sigma,t}\,d\log F_t$ to first order, so in continuous
time
\[
    \beta
    =
    \frac{1}{S_{\sigma,t}}
    \frac{d\langle \sigma^*,\log F\rangle_t}
         {d\langle \log F\rangle_t}
    \qquad(S_{\sigma,t}\neq0),
    \tag{4.11}
\]
and in discrete time $\beta$ is the population least-squares slope in
$\Delta\sigma^*=\beta S_\sigma\,\Delta\log F+\varepsilon$. The geometric
transport coefficient, the fixed-strike riding coefficient, the
quadratic-covariation projection and the regression slope are thus four
representations of the same number.

\begin{figure}[!htbp]
\centering
\includegraphics[width=\textwidth]{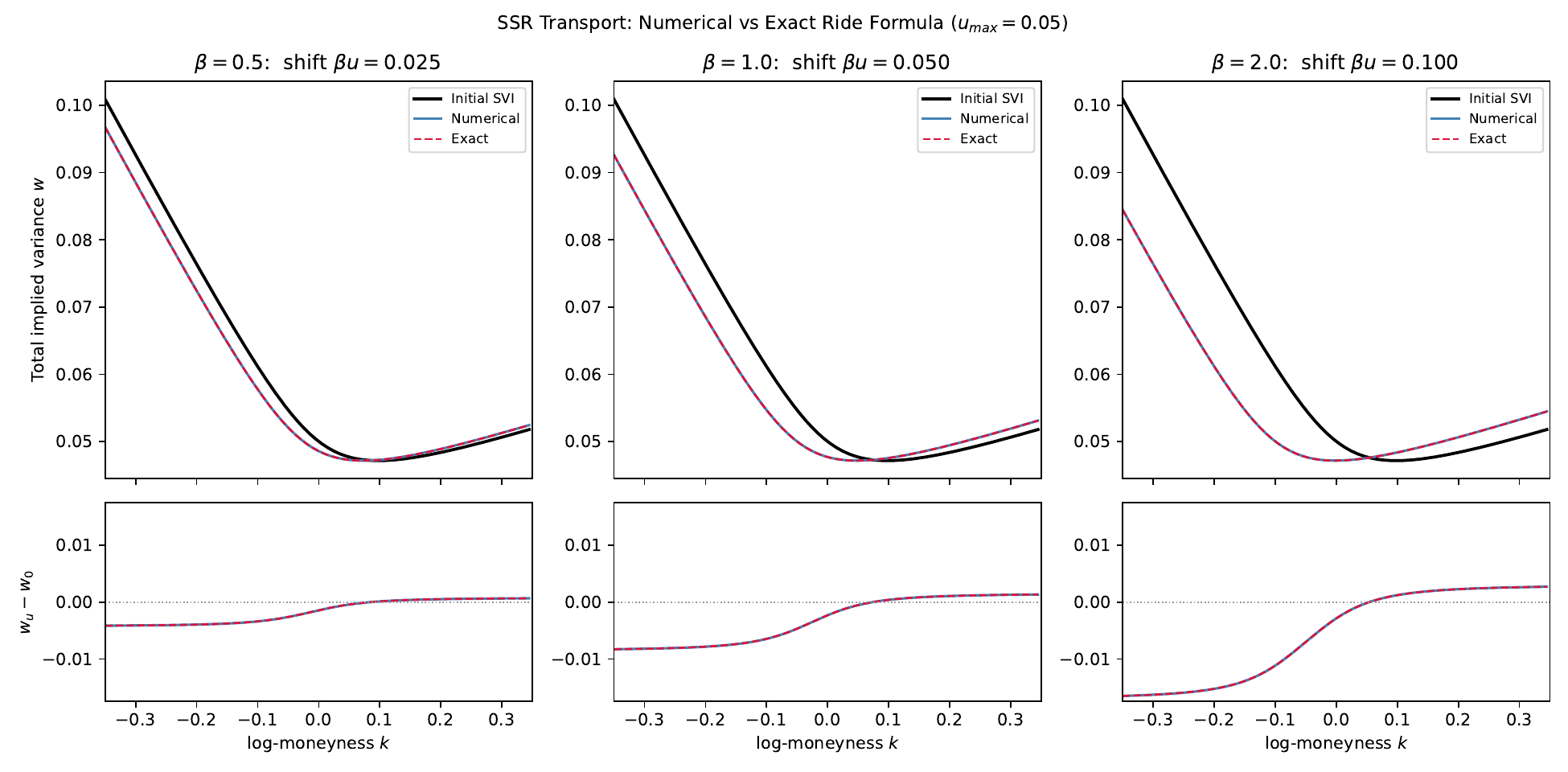}
\caption{Self-similar SSR transport applied to an SVI variance surface, for
a log-forward move $u=0.05$. \emph{Top row}: the initial surface (black) and
the transported surface (blue), with the exact ride formula
$w_0(k+\beta u)$ overlaid (red dashed). \emph{Bottom row}: the induced change
$w_u-w_0$ on a common vertical scale. Left to right: $\beta=0.5$ (between
sticky-delta and sticky-strike), $\beta=1.0$ (sticky-strike), $\beta=2.0$
(the short-maturity local-vol value). The displayed window is
$|k|\le0.35$; the solver grid is $[-1.5,1.5]$, but the transport shift
$\beta u\le0.10$ is not resolvable at that width, which is why the change is
also plotted directly. The difference scales with $\beta$ as the transport equation requires (roughly $-0.004$, $-0.008$ and $-0.016$ on the put wing), and the numerical and exact surfaces coincide to plotting accuracy, the
$L^\infty$ errors being those of Figure~\ref{fig:convergence}.}
\label{fig:transported-smiles}
\end{figure}

\FloatBarrier
\section{General Transport Fields and the Smile Jet Hierarchy}
\label{sec:general-transport}
The self-similar SSR equation $\partial_u w=\beta\,\partial_k w$
($u=\log F$) is too restrictive to describe general smile dynamics. It transports the entire
smile with a constant velocity in log-moneyness, and therefore forces the same
coefficient \(\beta\) to govern the evolution of level, skew, curvature and all
higher derivatives. Empirically, however, the ATM volatility level, the ATM
skew, and the curvature of the smile need not be transported at the same rate.
The natural generalization is to replace the constant transport coefficient
\(\beta\) by a transport velocity field.
\subsection{Transport Velocity Fields}
Let \(w=w(k,u,T)\) denote the total implied variance surface in floating
log-moneyness \(k=\log(K/F)\), forward log-level \(u=\log F\), and maturity
\(T\). A general transport dynamics is an equation of the form
\[
    \partial_u w(k,u,T)
    =
    v(k,u,T;w)\,\partial_k w(k,u,T),
    \tag{5.1}
\]
where \(v\) is a transport velocity field. The dependence of \(v\) on \(w\)
allows the transport speed to depend on the current variance surface itself.
The classical SSR model is recovered as the constant-velocity special case
\[
    v(k,u,T;w)\equiv \beta.
    \tag{5.2}
\]
Thus the SSR is not a separate model. It is the first exactly solvable member
of the general transport theory.
At fixed strike, since
\[
    \left.\partial_u\right|_K
    =
    \left.\partial_u\right|_k-\partial_k,
\]
equation \((5.1)\) is equivalent to
\[
    \left.\partial_u w\right|_K
    =
    \bigl(v(k,u,T;w)-1\bigr)\partial_k w.
    \tag{5.3}
\]
The shift by one remains purely mechanical. It is independent of the choice of
transport field.

\begin{remark}[The fixed-strike response and its hedging use]
\label{subsec:fixed-strike-response}
Write \(\psi(K,u,T):=\partial_u w|_K=(v-1)\partial_k w\) for the
left-hand side of \((5.3)\): the \emph{fixed-strike response} of total
variance to a forward move. The coordinate change is purely mechanical, but
the name earns its place twice below. First, \(\psi\) is defined at every
\((K,T)\), including where \(\partial_k w=0\), whereas
\(v=1+\psi/\partial_k w\) blows up there; this is why
Proposition~\ref{prop:identification-without-inversion} is a regression on
\(\psi\) rather than an inversion for \(v\). Second, \(\psi\) rather than
\(v\) is what enters a spot hedge: if a book is marked by a functional
\(P(u,\widetilde w)\) of the log-forward and the fixed-strike surface
\(\widetilde w(K,u,T):=w(\log K-u,u,T)\), with total-variance vega measure \(\nu\) (the signed measure on strikes for which \(\langle\nu,h\rangle\) is the first-order response of \(P\) to a perturbation \(h\)), then
\(\mathrm{d}\langle P,S\rangle/\mathrm{d}\langle S\rangle
=S^{-1}\partial_u P+S^{-1}\langle\nu,\psi\rangle\), so the instantaneous
minimum-variance spot hedge is the sticky-strike delta plus
\(\langle\nu,\psi\rangle/S\), including at strikes where
\(\partial_k w=0\). This is the empirical minimum-variance delta of Hull and
White~\cite{hullwhite2017} written for a whole surface. The classical
rules read off \((5.3)\) directly: sticky strike is \(\psi\equiv 0\), sticky
delta is \(\psi=-\partial_k w\), and the SSR is \(\psi=(\beta-1)\partial_k w\).
\end{remark}

\subsection{Transport and Source}
\label{subsec:transport-and-source}
Writing the surface's move as \(v\,\partial_k w\) can look like nothing more
than a way of naming \(v\), since wherever the skew is nonzero one can always
divide through and call the quotient the velocity. What gives \((5.1)\) real
content is the demand that the same \(v\) stay bounded as the skew passes
through zero. It is worth setting this out carefully, because the whole
empirical part of the paper turns on it.

Call \(\partial_u w(k,u,T)\) the \emph{response} of the surface to a forward
move. Equation \((5.1)\) says that the response lies in the
\emph{transport subspace}
\[
    \mathcal{T}_w
    :=
    \bigl\{\,v\,\partial_k w \;:\; v\in C(\mathbb{R})\text{ bounded}\,\bigr\}.
\]
Those elements of \(\mathcal{T}_w\) with \(v\in C^2\) and \(v,\partial_k v,
\partial_{kk}v\) bounded lie in the ambient space \(\mathcal{X}\) of
Section~\ref{subsec:geometric-viewpoint}, and \(\mathcal{T}_w\) is a proper
subset of the responses available there; the argument in (ii)
below uses only boundedness of \(v\). The natural object of study
is therefore the decomposition
\[
    \partial_u w
    =
    v(k,u,T)\,\partial_k w(k,u,T)+s(k,u,T),
    \qquad\text{equivalently}\qquad
    \psi=(v-1)\partial_k w+s,
    \tag{5.3c}
\]
in which \emph{pure transport} is the restriction \(s\equiv 0\) and \(s\) is
the \emph{source}. The source is defined relative to the restricted transport
class: it is the part of the response that no velocity \emph{in that class}
can generate, not a canonical component of an arbitrary response. The
restriction is made explicit in
Remark~\ref{rem:representation-restriction}, and no projection criterion
singling out a canonical pair \((v,s)\) is introduced in this paper.
Three consequences are recorded for use below.

\emph{(i) The exact classical rules are the constant-velocity pure
transports.} Sticky delta, sticky strike and the Bergomi SSR are the members
of \((5.3c)\) with \(s\equiv 0\) and \(v\) constant, at the values
\(0,1,\beta\) respectively
(Proposition~\ref{prop:classical-pure-transport}); the short-maturity
local-volatility value \(v\to 2\) is an asymptotic limit rather than a member
of that family (Remark~\ref{rem:lv-asymptotic}).

\emph{(ii) The source is identified where transport is powerless.} At a
point \(k^\ast\) with \(\partial_k w(k^\ast,T)=0\) one has
\(\mathcal{T}_w|_{k^\ast}=\{0\}\): under pure transport with bounded
velocity the response at \(k^\ast\) must vanish, whatever \(v\) is. A
nonzero response there is therefore \(s\) and nothing else, so the
hypothesis \(s\equiv 0\) can be rejected at that point using no auxiliary
assumption beyond boundedness of \(v\). Section~\ref{subsec:exp5} carries
out the test.

\emph{(iii) Pointwise inversion of \((5.1)\) is biased by the source.}
Dividing \((5.3c)\) by \(\partial_k w\) gives
\[
    \frac{\partial_u w}{\partial_k w}=v+\frac{s}{\partial_k w},
    \tag{5.3d}
\]
so the ratio is consistent for \(v\) only under pure transport, and diverges
at \(k^\ast\) whenever \(s(k^\ast)\neq 0\). The pole of the ratio estimator
and the source are two aspects of the same identity, which governs the
reading of Proposition~\ref{prop:identification-v} and of the measured
velocity profiles in Section~\ref{sec:spx-empirical}.

\begin{remark}[Representation versus restriction]
\label{rem:representation-restriction}
Since any response can be written as \((v+s/\partial_k w)\partial_k w\)
wherever \(\partial_k w\neq 0\), the split in \((5.3c)\) is not unique
pointwise: what makes it meaningful is a restriction on the class of
velocities, here boundedness and smoothness across smile extrema together
with constancy over an estimation window. Under that restriction \(s\) is
the part of the response that survives at \(\partial_k w=0\) and is
continuous across it, and \(v\) is the part that scales with the skew. The
empirical content of the transport law is the statement that, within this
class, \(s=0\). Accordingly, where precision is needed we speak of the
source \emph{relative to the restricted transport class} rather than of the
component of the response lying outside transport.
\end{remark}

\subsection{Characteristic Representation}
For a given transport velocity field \(v\), the corresponding characteristic
curves solve
\[
    \frac{d k_\tau}{d\tau}=-v(k_\tau,\tau,T;w),
    \qquad k_\tau\big|_{\tau=u}=k.
    \tag{5.4}
\]
Here \(\tau\) is the \emph{flow parameter}: a running value of the
log-forward level \(u=\log F\), so that \(k_\tau\) is the log-moneyness
reached by the characteristic through \((u,k)\) when the log-forward equals
\(\tau\). Throughout, \(\tau\) denotes the flow parameter and \(s\) the
source of Section~\ref{subsec:transport-and-source}.
Along such curves,
\[
    \frac{d}{d\tau}w(k_\tau,\tau,T)
    =
    \partial_u w(k_\tau,\tau,T)+\dot k_\tau\,\partial_k w(k_\tau,\tau,T)=0.
\]
Thus the variance surface is constant along the characteristic flow.
Figure~\ref{fig:characteristics} draws these curves in the \((u,k)\) plane for
three velocity fields. Note that the characteristics are
determined by \(v\) alone: the residual of
Lemma~\ref{lem:butterfly-transport} changes the value carried along a
characteristic, never the characteristic itself.
Let \(\chi_{u'\to u}\) denote the map that sends a terminal log-moneyness at
time \(u'\) back to its originating log-moneyness at time \(u\). Then the
transported variance surface satisfies the generalized ride formula
\[
    w(k,u',T)=w\bigl(\chi_{u'\to u}(k),u,T\bigr).
    \tag{5.5}
\]
In fixed-strike notation this becomes
\[
    w_{\mathrm{new}}(K,T)
    =
    w_{\mathrm{old}}\!\left(K_{\mathrm{eff}}(K;F,F'),T\right),
    \tag{5.6}
\]
where the effective strike is determined by the characteristic flow.
When \(v\equiv\beta\), the characteristics are straight lines, $k_\tau+\beta\tau=\text{constant}$,
and \((5.5)\) reduces to the classical SSR ride formula
\[
    w_{\mathrm{new}}(K,T)
    =
    w_{\mathrm{old}}
    \!\left(
        K\left(\frac{F}{F'}\right)^{1-\beta},T
    \right).
    \tag{5.7}
\]

\begin{figure}[!htbp]
\centering
\includegraphics[width=\textwidth]{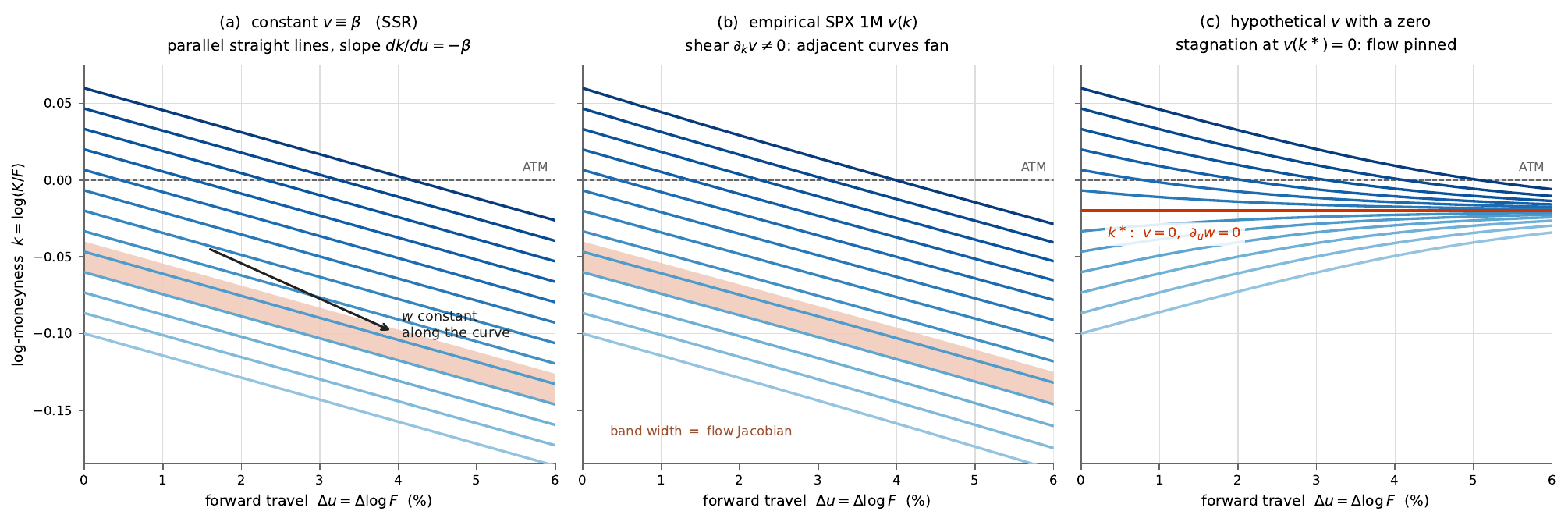}
\caption{Characteristics \((5.4)\) of the transport equation in the \((u,k)\)
plane, \(u=\log F\); shading runs with the starting log-moneyness, and \(w\)
is constant along every curve.
\emph{(a)}~Constant \(v{\equiv}\beta\): straight parallel lines of slope
\(-\beta\), for which \((5.5)\) collapses to the classical ride formula.
\emph{(b)}~The estimated SPX \(1\)M field (quadratic fit of
Table~\ref{tab:strike-dep-ssr}, continued \(C^1\)-linearly outside the quoted
range). Since \(\partial_k v\neq0\) the characteristics fan, and the shaded band (bounded by two adjacent characteristics, so its width is the Jacobian of \(\chi_{u'\to u}\)) widens by a factor \(1.053\) over a
\(6\%\) forward move, against exactly constant in~\emph{(a)}.
\emph{(c)}~A field with an interior zero \(v(k^\ast){=}0\), where the
characteristic is horizontal and neighboring ones are drawn onto it;
reversing the sign of \(v\) makes \(k^\ast\) a repelling separatrix. This
panel is illustrative, not estimated: the fitted SPX speeds lie in
\([1.11,1.56]\), far from the \(v{=}0\) it would require.}
\label{fig:characteristics}
\end{figure}

\subsection{The Smile Jet at the Money}
The advantage of the velocity-field formulation is that it naturally produces
a hierarchy of ATM smile dynamics. Define the ATM variance jet by
\[
    a_n(u,T)
    :=
    \partial_k^n w(0,u,T),
    \qquad n=0,1,2,\ldots .
    \tag{5.8}
\]
Thus \(a_0\) is the ATM total variance, \(a_1\) is the ATM variance skew,
\(a_2\) is the ATM variance curvature, and so on.
Assume that \(v\) is $C^\infty$ jointly in $(k,w)$ near $(0,w(0,u,T))$.
Define the composite
\[
    \tilde v(k;u,T) := v\bigl(k,u,T;\,w(k,u,T)\bigr),
\]
and write its ATM jet (total $k$-derivative) as
\begin{equation}
    v_j(u,T)
    :=
    \partial_k^j\tilde v(0;u,T),
    \qquad j=0,1,2,\ldots .
    \tag{5.9}\label{eq:jet_vel}
\end{equation}
When $v_w\equiv 0$ (i.e.\ $v=v(k,u,T)$), $\tilde v=v$ and $v_j$ reduces to
the ordinary $k$-jet $\partial_k^j v(0,u,T)$.
We define the named transport coefficients
\[
    \beta:=v_0,\qquad \eta:=v_1,\qquad \zeta:=v_2.
    \tag{5.10}\label{eq:named-coeffs}
\]
Thus \(\beta\) is the classical level-transport (SSR) coefficient, \(\eta\) is
the skew-transport coefficient, and \(\zeta\) is the curvature-transport
coefficient.
\subsection{Jet Transport Theorem}
\begin{theorem}[Jet transport hierarchy]
\label{thm:jet-transport}
Let $v=v(k,u,T;w)$ be $C^{N+1}$ jointly in $(k,w)$ in a neighborhood of
$(0,w(0,u,T))$, and let $w$ be $C^{N+1,1}$ jointly in $(k,u)$ near $k=0$.
Suppose $w$ satisfies $\partial_u w = v\,\partial_k w$. With $\tilde v$,
$a_n$, and $v_j$ as defined in $(5.8)$--\eqref{eq:jet_vel}, for each
$0\le n\le N$,
\[
    \dot a_n
    =
    \sum_{j=0}^{n}
    \binom{n}{j}
    v_j\,a_{n-j+1},
    \tag{5.11}
\]
where $\dot a_n = \partial_u a_n$.
\end{theorem}
\begin{proof}
Since $w\in C^{N+1,1}$ jointly in $(k,u)$, applying the Schwarz theorem
iteratively gives
\[
    \dot a_n
    =
    \partial_u\partial_k^n w(0,u,T)
    =
    \partial_k^n\bigl(\partial_u w\bigr)(0,u,T).
\]
Since $v$ is $C^{N+1}$ jointly in $(k,w)$ and $w$ is $C^{N+1}$ in $k$,
the composite $\tilde v(k;u,T)$ is $C^{N+1}$ in $k$ by the chain rule.
Substituting $\partial_u w = \tilde v(k;u,T)\,\partial_k w$ and applying
the Leibniz product rule to the $C^{N+1}$ functions
$k\mapsto\tilde v(k;u,T)$ and $k\mapsto\partial_k w(k,u,T)$ gives
\[
    \partial_k^n\bigl(\tilde v\,\partial_k w\bigr)
    =
    \sum_{j=0}^{n}
    \binom{n}{j}
    \bigl(\partial_k^j\tilde v\bigr)\bigl(\partial_k^{n-j+1}w\bigr).
\]
Evaluating at $k=0$ replaces $\partial_k^j\tilde v(0;u,T)$ with $v_j$
and $\partial_k^{n-j+1}w(0,u,T)$ with $a_{n-j+1}$, yielding $(5.11)$.

The structural content of $(5.11)$ is that $\dot a_n$ depends only on the
$n+1$ composite velocity jets $(v_0,\ldots,v_n)$ and the $n+2$ surface
jets $(a_0,\ldots,a_{n+1})$, not on $w$ or $v$ away from $k=0$.
\end{proof}
\begin{remark}[Identifiability of the velocity jets]
\label{rem:vj-w-dependent}
The $v_j$ in~\eqref{eq:jet_vel} are total $k$-derivatives of the composite
$\tilde v$, including both explicit $k$-dependence and implicit dependence
through $w$. Only $\tilde v$ and its jets are identified from a single trajectory
of $w$; separating the explicit $k$-dependence of $v$ from its
$w$-dependence requires additional structure beyond the transport equation.
\end{remark}
Substituting the named coefficients \eqref{eq:named-coeffs}, the first three
equations become
\[
\begin{aligned}
    \dot a_0 &= \beta a_1,\\
    \dot a_1 &= \beta a_2+\eta a_1,\\
    \dot a_2 &= \beta a_3+2\eta a_2+\zeta a_1.
\end{aligned}
\tag{5.12}\label{eq:jet-three}
\]
These identities show precisely how the classical SSR is extended: the level
move is controlled by \(\beta\), the skew move receives an additional
contribution from \(\eta\), and the curvature move receives additional
contributions from \(\eta\) and \(\zeta\).

\begin{remark}[Model-specific instances of the jet hierarchy]
For specific parametric models, the individual jet equations in
\eqref{eq:jet-three} have been computed in the literature.
The most prominent example is the SABR model~\cite{HaganKumarLesniewskiWoodward2002},
whose closed-form approximations for ATM implied volatility, ATM skew,
and ATM curvature as functions of spot implicitly encode a specific
instance of \eqref{eq:jet-three} for the SABR velocity field; see
also~\cite{Gatheral2006}.
The present theorem establishes \eqref{eq:jet-three} as a model-free
identity: it holds for any transport velocity $v(k)$, with the
SABR and other parametric cases arising as particular choices.
\end{remark}

In the constant-velocity case \(v\equiv\beta\), all higher velocity jets vanish
(\(\eta=\zeta=0\)), and the hierarchy collapses to
\[
    \dot a_n=\beta a_{n+1},
    \qquad n=0,1,2,\ldots .
    \tag{5.13}
\]
Thus the classical SSR is the special case in which the same coefficient
transports every derivative of the smile.
\subsection{Skew and Curvature Stickiness}
The transport coefficients \(\beta, \eta, \zeta\) defined in \eqref{eq:named-coeffs}
carry direct interpretations via the hierarchy \eqref{eq:jet-three}.
The skew-transport coefficient \(\eta\) measures how the transport velocity
changes away from ATM and therefore controls the part of ATM skew dynamics not
explained by constant SSR transport. The curvature-transport coefficient \(\zeta\)
similarly governs the residual curvature dynamics.
Thus skew and curvature dynamics are not added externally to the SSR; they are
the next terms in the Taylor jet of the same transport velocity field.
\begin{corollary}[Fixed-strike jet dynamics]
\label{cor:fixed-strike-jet}
Under the fixed-strike variation $D^{\mathrm{fx}}w = D^{\mathrm{fl}}w - \partial_k w$,
the first three jet equations become
\[
\begin{aligned}
    \dot a_0^{\mathrm{fx}} &= (\beta-1)\,a_1,\\
    \dot a_1^{\mathrm{fx}} &= (\beta-1)\,a_2+\eta\, a_1,\\
    \dot a_2^{\mathrm{fx}} &= (\beta-1)\,a_3+2\eta\, a_2+\zeta\, a_1.
\end{aligned}
\tag{5.14}\label{eq:jet-fixed}
\]
In particular, only \(\beta\) shifts by one; the skew- and curvature-transport
coefficients \(\eta\) and \(\zeta\) are coordinate-intrinsic.
\end{corollary}
\begin{proof}
Fixed-strike variation amounts to replacing \(v\) with \(v-1\) in the transport
equation, hence \(\beta\mapsto\beta-1\) while \(\eta,\zeta\) are unchanged.
Apply \eqref{eq:jet-three}.
\end{proof}
\subsection{Interpretation}
The generalized transport theory replaces the single SSR parameter by a local
transport velocity field. Skew dynamics are governed by its first derivative at ATM, curvature dynamics by
its second derivative, and higher-order smile dynamics by higher derivatives.
This gives a natural extension of the SSR. Rather than postulating separate
parameters for level, skew and curvature, the theory identifies them as the ATM
jet of the velocity field itself: $v(0),\ \partial_k v(0),\ \partial_k^2 v(0),\ldots$
The general case yields a nonlinear
characteristic flow and a corresponding generalized ride formula. The remaining
question is which transport velocity fields preserve static arbitrage. This is
the subject of the next section.
\subsection{Order of the Transport Operator and the It\^{o} Correction}
\label{subsec:order-of-transport}
The transport equations studied above are first order in $k$. We now
address when higher-order operators arise and which ones are admissible.

\begin{proposition}[Deterministic-forward leading order]
\label{prop:deterministic-leading-order}
Let $F_u=e^{u}$ be a $C^1$ deterministic forward and $w(k,T;F)$
be $C^{2,1,1}$. Then
\[
    w(k,T;F_{u+\delta u})-w(k,T;F_u)
    =
    \mathcal{D}_w(k,T)\,\delta u
    +
    o(\delta u),
    \quad
    \mathcal{D}_w:=\partial_u w|_k.
\]
No second-order $k$-derivative operator appears at leading order.
\end{proposition}
\begin{proof}
Immediate Taylor expansion in $u$ along the deterministic path.
\end{proof}

\begin{proposition}[Stochastic-forward It\^{o} correction]
\label{prop:stochastic-ito}
Let $dF_t/F_t=\mu_t\,dt+\sigma^F_t\,dW_t$ and $u_t=\log F_t$, where
$\sigma^F_t$ is the forward's own instantaneous volatility, a property of its
realized dynamics rather than of the quoted surface. Assume
$w=w(k,T;u)$ is $C^{2,2,2}$ jointly. Then
\begin{equation}
    dw(k,T;u_t)
    =
    \partial_u w\!\left(\mu_t-\tfrac{1}{2}(\sigma^F_t)^2\right)dt
    +\partial_u w\,\sigma^F_t\,dW_t
    +\tfrac{1}{2}\partial_{uu}w\,(\sigma^F_t)^2\,dt.
    \tag{5.15}\label{eq:dw-ito}
\end{equation}
If $\partial_u w=v\,\partial_k w$ along the path, where $v=v(k,u,T;w)$
may depend on $w$ pointwise, then writing $p:=\partial_k w$ and letting
$\partial_k v = \partial_k^{\mathrm{exp}}v + v_w\,p$ denote the total
$k$-derivative of $v$ (explicit plus implicit through $w$),
\begin{equation}
    \partial_{uu}w
    =
    v_u\,p
    +v\,v_w\,p^2
    +v\,(\partial_k v)\,p
    +v^2\,\partial_{kk}w,
    \tag{5.16}\label{eq:double-u}
\end{equation}
where $v_u:=\partial v/\partial u\big|_{k,T,w}$ and $v_w:=\partial v/\partial w$.
When $v$ has no explicit $u$- or $w$-dependence ($v_u\equiv v_w\equiv 0$),
\eqref{eq:double-u} reduces to $v\,\partial_k v\,p+v^2\,\partial_{kk}w$.
\end{proposition}
\begin{proof}
Equation~\eqref{eq:dw-ito} is It\^{o}'s formula applied to $f(u_t)=w(k,T;u_t)$.
For~\eqref{eq:double-u}: differentiate $\partial_u w=v\,p$ in $u$.
Since $v=v(k,u,T;w(k,T;u))$,
\[
    \partial_u v = v_u + v_w\,\partial_u w = v_u + v_w\,v\,p.
\]
Since $\partial_u p = \partial_u\partial_k w = \partial_k(v\,p)
=(\partial_k v)\,p+v\,\partial_{kk}w$, assembling gives
\[
    \partial_{uu}w
    = (\partial_u v)\,p + v\,\partial_u p
    = (v_u + v_w v p)\,p + v\bigl[(\partial_k v)\,p+v\,\partial_{kk}w\bigr],
\]
which is~\eqref{eq:double-u}.
\end{proof}

\begin{proposition}[Backward-parabolic transport excluded]
\label{prop:backward-parabolic-excluded}
Let $L=\alpha(k)\,\partial_{kk}+\gamma(k)\,\partial_k$ with
$\alpha,\gamma\in C^\infty$ and $\alpha(k_0)<0$ for some $k_0$.
Then $L$ does not generate a locally
well-posed $C^{2,1}$ flow; in particular it cannot define an
admissibility-preserving transport flow on $\mathcal{A}^\circ$.
\end{proposition}
\begin{proof}
Where $\alpha<-\lambda_0<0$ the principal symbol
$-\alpha(k)\xi^2>\lambda_0\xi^2$ is
positive at high frequencies. A standard frozen-coefficient
geometric-optics construction produces $w_0^{(m)}\in C^{2,1}$ with
$\|w_0^{(m)}-w_0^*\|_{C^{2,1}}\to 0$ whose solutions satisfy
$\|S_u w_0^{(m)}\|_{C^0}\gtrsim e^{\lambda_0 m^2 u}$ for any $u>0$,
indexed by the oscillation mode $m$. This
contradicts continuity of the solution operator in
$C^{2,1}$~\cite{Friedman1964}. Hence the initial-value problem for $L$
is ill-posed in $C^{2,1}$, so $L$ cannot generate a continuous flow
on $\mathcal{A}^\circ$.
\end{proof}

\begin{remark}[Heuristic obstruction for $N\ge3$]
\label{rem:higher-order-obstruction}
The following motivates restricting attention to $N\le2$ but is not a proof,
and nothing later depends on it. The constraints defining $\mathcal{A}^\circ$
involve only $w,\partial_k w,\partial_{kk}w$ through $g[w]$, so
$\partial_k^j w(k_0)$ for $j\ge3$ is freely specifiable. The chain rule gives
$\partial_u g[w_u](k_0)|_{u=0}
=\tfrac12 c_N(k_0)\,\partial_k^{N+2}w_0(k_0)+O(1)$ with the $O(1)$ depending
only on $\partial_k^{j\le N+1}w_0(k_0)$, so choosing
$c_N(k_0)\partial_k^{N+2}w_0(k_0)$ large and negative while keeping
$g[w_0]>0$ would force $g[w_u](k_0)<0$ for arbitrarily small $u>0$. Making
this a proof requires controlling the $O(u^2)$ remainder uniformly; the
statement remains a conjecture (Section~\ref{sec:extensions}).
\end{remark}

\begin{remark}[The admissible transport class]
\label{rem:admissible-class}
We work throughout within the convex-parabolic class
\begin{equation}
    \mathcal{D}_w
    =
    v(k,T;w)\,\partial_k w
    +\tfrac{1}{2}D^2(k,T;w)\,\partial_{kk}w,\quad D^2\ge0,
    \tag{5.17}\label{eq:transport-with-diffusion}
\end{equation}
which the three preceding results delimit: a deterministic forward gives the
pure first-order operator at leading order
(Proposition~\ref{prop:deterministic-leading-order}); a stochastic forward
appends an It\^{o} correction with $D^2=(\sigma^F_t)^2v^2\ge0$ when
$v_w\equiv0$ (Proposition~\ref{prop:stochastic-ito}); and $D^2<0$ is excluded
by backward-parabolic ill-posedness
(Proposition~\ref{prop:backward-parabolic-excluded}), with $N\ge3$
heuristically obstructed above. When $v$ depends on $w$ the It\^{o}
correction also generates $\tfrac12(\sigma^F_t)^2 v\,v_w(\partial_k w)^2$,
outside \eqref{eq:transport-with-diffusion}; whether that quasilinear class
preserves $\mathcal{A}^\circ$ is open, and would be approached through the
standard theory for such operators~\cite{LadyzhenskayaSolonnikovUraltseva1968}.

Two features of the resulting framework matter later. The
factorization $v:=\mathcal{D}_w/\partial_k w$ is algebraic wherever
$\partial_k w\neq0$ and so imposes nothing; it is the requirement of a
bounded, smooth $v$ across smile extrema that restricts
(Section~\ref{subsec:transport-and-source}). And the hierarchy of
Theorem~\ref{thm:jet-transport} is \emph{local but not closed}: row $n$ needs
$a_{n+1}$, so $(a_0,\ldots,a_n)$ closes only once $a_{n+1}$ is supplied or
treated as an observed input, which is what
Section~\ref{sec:sequential_estimation} does.
\end{remark}

\subsection{Local Jet Transport}
\label{sec:local-jet-transport}
The Jet Transport Theorem~\ref{thm:jet-transport} expands the transport
velocity $v$ around the at-the-money point $k=0$. The theorem itself is
however a pointwise application of the Leibniz rule and holds equally at any
expansion center $k_0$. Setting $k_0\neq0$ yields local transport equations
governing the dynamics of the variance surface away from the money, and the
family $\{v_0(k_0)\}_{k_0}$ traces out a velocity profile across strikes.

One caveat governs the whole subsection and is better registered before the
results than after them. The hierarchy, and with it the pointwise formula of
Proposition~\ref{prop:identification-v}, is derived \emph{from} the transport
equation $(5.1)$ and so presupposes pure transport; under a source the same
ratio measures $v+s/\partial_k w$ by $(5.3d)$. Since
Section~\ref{subsec:exp5} rejects $s\equiv0$ at the smile minimum at the two
shortest tenors, the inverted profiles are not estimates of $v$ there and we
do not treat them as such. Proposition~\ref{prop:identification-without-inversion},
which estimates $(v,s)$ jointly, is the primary identification route; the
inversion is reported because its failure mode is informative.

\begin{definition}[Local variance jet]
\label{def:local-jet}
For $w\in C^{N+1}$ in $k$ near $k_0\in\mathbb R$, the
\emph{local variance jet at $k_0$} is
\begin{equation}
    b_n^{(k_0)}(u,T) := \partial_k^n w(k_0,u,T), \quad n=0,1,\ldots,N,
    \tag{5.18}\label{eq:local-jet-def}
\end{equation}
and the \emph{local velocity jet} is
$v_j(k_0,u,T) := \partial_k^j\tilde v(k_0;u,T)$,
the $j$-th total $k$-derivative of the composite
$\tilde v(k;u,T):=v(k,u,T;w(k,u,T))$ evaluated at $k_0$.
Setting $k_0=0$ recovers the ATM objects: $b_n^{(0)}=a_n$ and
$v_j(0,u,T)=v_j(u,T)$ as in \eqref{eq:jet_vel}.
\end{definition}

\begin{corollary}[Local jet transport hierarchy]
\label{thm:local-jet-transport}
Under the same regularity near $k_0$ as Theorem~\ref{thm:jet-transport}
requires near $k=0$,
\begin{equation}
    \dot b_n^{(k_0)}
    = \sum_{j=0}^{n} \binom{n}{j}\,v_j(k_0,u,T)\,b_{n-j+1}^{(k_0)}(u,T),
    \quad 0\le n\le N,
    \tag{5.19}\label{eq:local-jet-hierarchy}
\end{equation}
where $v_j(k_0,u,T):=\partial_k^j\tilde v(k_0;u,T)$ denotes the $j$-th
total $k$-derivative of the composite $\tilde v(k;u,T):=v(k,u,T;w(k,u,T))$
evaluated at $k_0$.
\end{corollary}
\begin{proof}
The proof of Theorem~\ref{thm:jet-transport} uses only the local structure
of $w$ and $\tilde v$ near the expansion point. Replacing $k=0$ by $k=k_0$ and
$a_n, v_j$ by $b_n^{(k_0)}, v_j(k_0)$ throughout gives
\eqref{eq:local-jet-hierarchy} verbatim.
\end{proof}

\begin{proposition}[Nonparametric identification of the velocity field]
\label{prop:identification-v}
Let $\Omega:=\{k_0:\partial_k w(k_0,u,T)\neq 0\}$. The composite transport
velocity $\tilde v(k_0;u,T):=v(k_0,u,T;w(k_0,u,T))$ is identified
pointwise on $\Omega$ by
\begin{equation}
    \tilde v(k_0;u,T) = \frac{\partial_u w(k_0,u,T)}{\partial_k w(k_0,u,T)}
    = \frac{\dot b_0^{(k_0)}}{b_1^{(k_0)}}.
    \tag{5.20}\label{eq:identification-v}
\end{equation}
Moreover, given a discrete time series $\{w(\cdot,u_i,T)\}$ with step
$\Delta u = u_{i+1}-u_i$, the first-difference estimator
\[
    \hat{v}(k_0,T) := \frac{\Delta b_0^{(k_0)}}{b_1^{(k_0)}\,\Delta u}
\]
satisfies $\hat{v}(k_0,T) = \tilde v(k_0;u,T) + O(\Delta u)$ uniformly on compact
subsets of $\Omega$ as $\Delta u\to 0$.
\end{proposition}
\begin{proof}
Setting $n=0$ in \eqref{eq:local-jet-hierarchy} gives $\dot b_0^{(k_0)} =
v_0(k_0)\,b_1^{(k_0)}$. Since $v_0(k_0)=\tilde v(k_0;u,T)$ by definition
and $b_1^{(k_0)}=\partial_k w(k_0,u,T)\neq 0$ on $\Omega$,
solving yields \eqref{eq:identification-v}. For the estimator, Taylor's
theorem with remainder gives $\Delta b_0^{(k_0)} = \dot b_0^{(k_0)}\Delta u
+ \tfrac{1}{2}\ddot b_0^{(k_0)}(\Delta u)^2 + O((\Delta u)^3)$; dividing by
$b_1^{(k_0)}\Delta u$ and using boundedness of $\ddot b_0^{(k_0)}/b_1^{(k_0)}$
on compact subsets of $\Omega$ establishes the $O(\Delta u)$ bound.
\end{proof}

The ratio \eqref{eq:identification-v} \emph{inverts} the transport law: under
the decomposition $(5.3c)$ the same ratio measures $v+s/\partial_k w$ by
$(5.3d)$, so the inversion is consistent for $v$ only under pure transport.
The following proposition identifies $(v,s)$ jointly and divides by nothing.

\begin{proposition}[Identification without inversion]
\label{prop:identification-without-inversion}
Fix $(k_0,T)$ and suppose $(5.3c)$ holds with $(v,s)$ constant in $u$ over
the estimation window. Then $(v,s)$ is identified by the joint variation of
the response and the skew, and is consistently estimated by the
least-squares fit of
\begin{equation}
    \Delta w(k_0,t)
    =
    c+s\,\Delta u_t+v\bigl(b_1^{(k_0)}(t)\,\Delta u_t\bigr)+\theta\,\Delta t
    +\varepsilon_t,
    \tag{5.20a}\label{eq:identification-vs}
\end{equation}
where $b_1^{(k_0)}(t)=\partial_k w(k_0,t)$ is measured at $t$, and $\theta$
is a plain time-trend term in calendar time $t$, included so that any
variance change over the sample that is unrelated to the forward move
($\Delta u_t$) or the skew ($b_1^{(k_0)}(t)\Delta u_t$) is absorbed by
$\theta\Delta t$ rather than attributed to $v$ or $s$; $(5.3c)$ does not
imply this term, and $\theta$ itself is not estimated for its own sake.
The fit involves no division by $b_1$ and is therefore valid at and across
$\partial\Omega$. Both regressors are evaluated at the fixed log-moneyness
node $k_0$, not at a fixed dollar strike: a fixed-$K$ panel drifts through
moneyness regimes as the forward moves, which over the multi-year windows
used in Section~\ref{sec:spx-empirical} would reintroduce exactly the
non-stationarity that motivates working in $k$ rather than $K$ in the first
place (Remark~\ref{subsec:fixed-strike-response} is a pointwise identity,
not a claim that a fixed-strike time series is stationary).
Proposition~\ref{prop:identification-v}
is the special case $s\equiv 0$, in which \eqref{eq:identification-vs} and the
ratio coincide. The proposition identifies $(v,s)$ \emph{within} the restricted class (velocities constant in $u$ over the estimation window, identified through the time variation of the skew), and does not assert a
canonical decomposition of an arbitrary response
(Remark~\ref{rem:representation-restriction}).
\end{proposition}
\begin{proof}
Substituting $(5.3c)$ into the floating increment and discretizing gives
\[
    \Delta w(k_0,t)
    =
    \bigl(s+v\,b_1^{(k_0)}(t)\bigr)\Delta u_t+\theta\,\Delta t+O(|\Delta u|^2),
\]
which is \eqref{eq:identification-vs}. The regressors $\Delta u_t$ and
$b_1^{(k_0)}(t)\Delta u_t$ are linearly independent whenever
$b_1^{(k_0)}(\cdot)$ is non-degenerate over the window, and $b_1^{(k_0)}(t)$
is predetermined at $t$; least squares is then consistent for $(s,v)$ under
the usual moment conditions, with Newey--West standard errors accommodating
overlapping-shock dependence.
\end{proof}

\begin{remark}[Where the inversion fails]
\label{rem:inversion-fails}
The strikes this excludes are not a negligible corner of the SPX surface. An
interior minimum of $w(\cdot,T)$ sits inside the quoted strike range on
$98.5\%$ of trading days at 1M (mean log-moneyness $+5.2\%$) and $92.7\%$ at
3M (mean $+8.8\%$), while by 6M and 1Y it has moved outside the range on
essentially every day. By
$(5.3d)$ the ratio diverges and changes sign across that point whenever
$s\neq0$, so the location of the divergence moves outward with maturity, a signature examined at 1M and 3M in
Remark~\ref{subsec:reading-velocity}, and unavailable at 6M and 1Y.
\end{remark}

\begin{remark}[What the joint fit identifies, and what rejecting $s=0$ means]
\label{rem:rejection-meaning}
Two points about \eqref{eq:identification-vs}, both used in
Section~\ref{subsec:exp5}. First, identification comes from the \emph{time}
variation of the skew at the node: the equation is one relation in two
unknowns per date, and it is the day-to-day movement of $b_1^{(k_0)}$ that
separates them, so where that movement is small relative to the level of
$b_1$ the two are weakly separated and the standard errors report it. The
cost of not dividing by $b_1$ is therefore paid in variance, not bias; note that $b_1$ is measured from the same quotes as the response, so correlated
measurement error can transfer loading between the coefficients.
Second, the null $s=0$ asserts that a velocity \emph{constant over the
window} reproduces the response. Where $b_1(t)\neq0$ throughout, the
day-by-day velocity $v_t=v+s/b_1(t)$ reproduces it with no source, so a
rejection does not exclude state-dependent transport as a representation; it
excludes it as a \emph{bounded} one, since $v_t$ blows up and changes sign as
$b_1(t)\to0$. The test asks whether transport survives with a velocity that
stays finite across the smile minimum, and a nonzero response at a date on
which the node \emph{is} the minimum answers no. Rejections at nodes near
rather than at the minimum are correspondingly weaker.
\end{remark}

\begin{remark}[The jet pattern is not forced by the theory]
\label{rem:hierarchy-not-forced}
No ordering among $v(k_0),\partial_k v(k_0),\partial_{kk}v(k_0)$ is implied by
(V1)--(V3) and (W): $v(k)=\tfrac12+\varepsilon\sin k$ with
$\varepsilon\in(0,\tfrac14)$ satisfies every admissibility condition and has
jets of every relative size across $k_0$. Whatever pattern the estimates
display is therefore a property of equity markets, not of the theory. Note
also that coefficients estimated under the restriction $s\equiv0$ absorb the
source into the velocity, so a nonzero source produces apparent strike
dependence of $v$ even when $v$ is constant; what the rejection of
self-similar transport is evidence \emph{for} is settled by
\eqref{eq:identification-vs}, not by the jets alone.
\end{remark}

\subsection{Sequential Estimation and Empirical Identification}
\label{sec:sequential_estimation}
The Jet Transport Theorem establishes a hierarchy of local transport equations
governing the evolution of the variance jet. We now derive a sequential
estimator for the transport coefficients $\beta,\eta,\zeta$ that is consistent
with the hierarchy.

Recall from $(5.8)$ that $a_n := \partial_k^n w(0,u,T)$ are the
ATM total-variance jets; the surface has the Taylor expansion
\begin{equation}
w(k,u,T)
=
\sum_{n=0}^{\infty}
\frac{a_n(u,T)}{n!}\,k^n.
\tag{5.21}\label{eq:local_taylor_empirical}
\end{equation}
The Jet Transport Theorem (eq.~\eqref{eq:jet-three}) gives the exact coupled
dynamics
\begin{equation}
\dot a_n
=
\sum_{j=0}^{n}\binom{n}{j}v_j\,a_{n-j+1},
\qquad
\dot a_n := \partial_u a_n,
\tag{5.22}\label{eq:jet_regression}
\end{equation}
which is the lower-triangular linear system in
matrix form.  Because the matrix is lower triangular with $a_1\neq0$ on the
diagonal, the coefficients $(\beta,\eta,\zeta)$ are recovered exactly in the
noiseless case by \emph{forward substitution}: solve row $n=0$ for $\beta$,
then use $\hat\beta$ to eliminate the known term from row $n=1$ and solve for
$\eta$, and so on.

The resulting sequential estimator is as follows.
The first step estimates $\beta=v_0$ directly from the $n=0$ row, which
contains no cross-terms:
\begin{equation}
\Delta a_0
=
\beta\, a_1\Delta u
+
\varepsilon_0.
\tag{5.23}\label{eq:ssr_regression}
\end{equation}
The second step subtracts the known $\hat\beta$ contribution from the $n=1$
row and regresses the residual on $a_1$:
\begin{equation}
\Delta a_1 - \hat\beta\, a_2\,\Delta u
=
\eta\, a_1\,\Delta u
+
\varepsilon_1,
\tag{5.24}\label{eq:skew_regression}
\end{equation}
isolating $\eta=v_1$.  The third step subtracts both known contributions from
the $n=2$ row:
\begin{equation}
\Delta a_2 - \hat\beta\, a_3\,\Delta u - 2\hat\eta\, a_2\,\Delta u
=
\zeta\, a_1\,\Delta u
+
\varepsilon_2,
\tag{5.25}\label{eq:curv_regression}
\end{equation}
isolating $\zeta=v_2$.  The factors $1$, $1$, $2$ in the subtracted terms are
the binomial coefficients $\binom{n}{j}$ from \eqref{eq:jet_regression}; they
are not free parameters.  Higher-order coefficients $v_n$, $n\ge 3$, follow
the same pattern.

\begin{proposition}[Forward substitution recovers the jet hierarchy exactly]
\label{prop:forward-sub}
Suppose $a_1\neq 0$ and let $\varepsilon_n=0$ in
\eqref{eq:ssr_regression}--\eqref{eq:curv_regression}. Then the
forward-substitution procedure uniquely recovers $(\beta,\eta,\zeta)$.
Equivalently, equations \eqref{eq:ssr_regression}--\eqref{eq:curv_regression}
are the forward-substitution solution of the lower-triangular linear system
\begin{equation}
    \begin{pmatrix}a_1 & 0 & 0\\ a_2 & a_1 & 0\\ a_3 & 2a_2 & a_1\end{pmatrix}
    \begin{pmatrix}\beta\\\eta\\\zeta\end{pmatrix}\Delta u
    =
    \begin{pmatrix}\Delta a_0\\\Delta a_1\\\Delta a_2\end{pmatrix},
    \tag{5.26}\label{eq:lower-tri-system}
\end{equation}
whose determinant is $a_1^3\neq 0$.
\end{proposition}
\begin{proof}
The matrix entries $\binom{n}{j}a_{n-j+1}$ are exactly the Leibniz
coefficients from the jet hierarchy \eqref{eq:jet_regression}: the $(n,j)$
entry is $\binom{n}{j}a_{n-j+1}$ for $j\le n$ and zero otherwise. The
diagonal is $a_1$ throughout, so $\det = a_1^3\neq 0$. Forward substitution
on the lower-triangular system solves row $n=0$ for $\beta$, then row
$n=1$ for $\eta$ after substituting $\hat\beta$, then row $n=2$ for $\zeta$
after substituting $\hat\beta$ and $\hat\eta$, giving exactly
\eqref{eq:ssr_regression}--\eqref{eq:curv_regression}.
\end{proof}

Under noise, forward substitution remains the preferred implementation
because each step isolates one coefficient from the rows below it, so misspecification in equation $n$ (an omitted higher-order jet term, say) cannot contaminate coefficients identified earlier: $\hat\beta$ is
unaffected by problems in the $\eta$ or $\zeta$ equations, and $\hat\eta$ by
problems in the $\zeta$ equation.

The predicted surface change follows by substituting the estimated hierarchy
into \eqref{eq:local_taylor_empirical}:
\begin{equation}
\widehat{\Delta w}(k)
=
\Delta u\sum_{n=0}^{N}\frac{k^n}{n!}\,\widehat{\dot a}_n,
\qquad
\widehat{\dot a}_n
:=
\sum_{j=0}^{n}\binom{n}{j}\hat v_j\,a_{n-j+1},
\tag{5.27}\label{eq:jet_prediction}
\end{equation}
where $\hat v_0=\hat\beta$, $\hat v_1=\hat\eta$, $\hat v_2=\hat\zeta$.
Explicitly for the first three orders:
\[
\widehat{\Delta w}(k)
=
\Delta u\Bigl[
\hat\beta\,a_1
+(\hat\beta\,a_2+\hat\eta\,a_1)\,k
+\tfrac{1}{2}(\hat\beta\,a_3+2\hat\eta\,a_2+\hat\zeta\,a_1)\,k^2
+\cdots
\Bigr].
\]

\begin{remark}
The empirical implementation of this estimator requires ATM total-variance jets
$a_0,a_1,a_2,a_3$ extracted from market data by fitting \eqref{eq:local_taylor_empirical}
with factorial denominators intact (see eq.~\eqref{eq:jet-poly-fit}).
\end{remark}

\begin{remark}[The sequential estimator under a source]
Equations \eqref{eq:ssr_regression}--\eqref{eq:curv_regression} are the
hierarchy under pure transport. Under $(5.3c)$ the row-$n$ equation acquires
the term $\partial_k^n s(0)\,\Delta u$, so the population slope of
\eqref{eq:ssr_regression} is
\[
    \operatorname{plim}\hat\beta
    =
    \beta+\frac{\E\bigl[s(0)\,a_1\,(\Delta u)^2\bigr]}{\E\bigl[a_1^2(\Delta u)^2\bigr]},
\]
and similarly for the higher rows. The classical SSR regression therefore
measures $\beta$ exactly when the source vanishes at the money;
Experiment~5 finds that it does, within sampling error, at 1M--3M, and that
a small positive source appears at 6M--1Y, where the reported $\hat\beta$
should be read as $\beta+s(0)/a_1$ up to the weighting above. Away from the
money the same identity is $(5.3d)$.
\end{remark}

\FloatBarrier
\section{Arbitrage-Free Transport}
\label{sec:arb-preservation}
The transport hierarchy of Section~\ref{sec:general-transport} produces a broad
family of candidate dynamics. This section identifies conditions on the
velocity field under which the induced flow remains in~$\mathcal{A}^\circ$.

\subsection{Transport of the Calendar Density}

\begin{lemma}[Calendar-density transport]
\label{lem:calendar-transport}
Let $w\in C^{2,1}$ solve $\partial_u w=v\,\partial_k w$ with $v\in C^{1}$ in
$(k,T,w)$. Set $q:=\partial_T w$. Then
\begin{equation}
    \partial_u q
    =
    v\,\partial_k q
    +(\partial_T v)\,\partial_k w
    +v_w[q]\,\partial_k w,
    \tag{6.1}\label{eq:q-transport}
\end{equation}
where $v_w[q]$ denotes the Fr\'echet derivative of $v$ in its $w$-argument,
applied to the perturbation $q=\partial_T w$.
\end{lemma}
\begin{proof}
Differentiate $\partial_u w=v(k,u,T;w)\,\partial_k w$ with respect to $T$:
\begin{align*}
    \partial_u\partial_T w
    &=(\partial_T v)\,\partial_k w + v_w[\partial_T w]\,\partial_k w
     + v\,\partial_k\partial_T w.
\end{align*}
Setting $q=\partial_T w$ and using $\partial_u\partial_T=\partial_T\partial_u$
(justified by $C^{2,1}$ regularity) yields~\eqref{eq:q-transport}.
\end{proof}

\begin{proposition}[Local calendar preservation]
\label{prop:calendar-preserved}
In the setting of Lemma~\ref{lem:calendar-transport}, if
$q(k,0,T)\ge q_*>0$ uniformly in $k$ and $v$, $\partial_T v$, $v_w$,
$\partial_k w$ are bounded, then $q>0$ for all sufficiently small $|u|$.
\end{proposition}
\begin{proof}
Let $q_*:=\inf_k q(k,0,T)>0$.
Fix $k$ and let $\tau\mapsto k_\tau$ solve $\dot k_\tau=-v(k_\tau,\tau,T;w_\tau)$,
$k_0=k$, with $\tau$ the flow parameter of $(5.4)$, here anchored at $\tau=0$
rather than at $\tau=u$.
Since $v$ depends on $w$ pointwise, $v_w[q]=(\partial v/\partial w)\,q$,
so along this characteristic equation~\eqref{eq:q-transport} reduces to
the scalar linear ODE
\[
    \frac{d}{d\tau}q(k_\tau,\tau,T) \;=\; A(\tau) \;+\; B(\tau)\,q(k_\tau,\tau,T),
\]
where
$A(\tau):=(\partial_T v)(k_\tau,\tau,T;w_\tau)\,\partial_k w(k_\tau,\tau,T)$ and
$B(\tau):=(\partial v/\partial w)(k_\tau,\tau,T;w_\tau)\,\partial_k w(k_\tau,\tau,T)$.
The boundedness hypotheses give uniform bounds
\[
    |A(\tau)|\le \bar A:=\|\partial_T v\|_\infty\|\partial_k w\|_\infty,
    \qquad
    |B(\tau)|\le \bar B:=\|v_w\|_\infty\|\partial_k w\|_\infty.
\]
Variation of constants yields
\[
    q(k_u,u,T)
    =
    e^{\int_0^u B(\tau)\,d\tau}
    \!\left[\,q(k,0,T)+\int_0^u\!A(\tau)\,e^{-\int_0^\tau B(r)\,dr}\,d\tau\right].
\]
Using $|\!\int_0^\tau B\,dr|\le\bar B|\tau|$ and $|A|\le\bar A$, the bracket
is bounded below by $q_*-(\bar A/\bar B)(e^{\bar B|u|}-1)$
(interpreting $(\bar A/\bar B)(e^{\bar B|u|}-1)$ as $\bar A|u|$ when
$\bar B=0$), so
\[
    q(k_u,u,T) \;\ge\; e^{-\bar B|u|}\,q_*
    \;-\;
    \frac{\bar A}{\bar B}\bigl(1-e^{-\bar B|u|}\bigr).
\]
The right-hand side is strictly positive for all $|u|<\delta_1$,
where $\delta_1$ is the calendar-density admissibility radius,
\[
    \delta_1
    :=
    \frac{1}{2\bar B}\,\ln\!\left(1+\frac{\bar B\,q_*}{\bar A}\right)
    > 0
    \qquad(\delta_1:=q_*/(2\bar A)\text{ when }\bar B=0).
\]
Since $k$ was arbitrary and all bounds are uniform in $k$, this gives
$q(\,\cdot\,,u,T)>0$ pointwise for all $|u|<\delta_1$.
\end{proof}

\subsection{Transport of the Butterfly Density}

\begin{lemma}[Butterfly-density transport]
\label{lem:butterfly-transport}
Under $\partial_u w=v\,\partial_k w$ with $v\in C^2$ in $k$, the Gatheral
density functional $g[w]$ from~\eqref{eq:g-functional} satisfies
\begin{equation}
    \partial_u g
    =
    v\,\partial_k g
    +\mathcal R[v;w],
    \tag{6.2}\label{eq:g-transport}
\end{equation}
where, with $p:=\partial_k w$ and $c:=\partial_{kk}w$ the curvature at a
general strike $k$ (so $c|_{k=0}=a_2$), and the explicit-$k$ partial
\begin{equation}
    \partial_k^{\mathrm{exp}} g
    :=
    \left.\frac{\partial g(k,w,p,c)}{\partial k}\right|_{w,p,c\,\mathrm{fixed}}
    =
    -\frac{p}{w}\!\left(1-\frac{kp}{2w}\right),
    \tag{6.3}\label{eq:dkexpg}
\end{equation}
the residual is
\begin{equation}
    \mathcal R[v;w]
    =
    -v\,\partial_k^{\mathrm{exp}} g
    +(\partial_k v)\bigl[(\partial_p g)\,p+c\bigr]
    +\tfrac{1}{2}(\partial_{kk}v)\,p,
    \tag{6.4}\label{eq:g-residual}
\end{equation}
with $\partial_p g = -\frac{k}{w}(1-\frac{kp}{2w}) - \frac{p}{2}(\frac{1}{w}+\frac{1}{4})$
and $\partial_c g = \tfrac{1}{2}$.

\emph{Special case.} For constant $v\equiv\beta$, so $\partial_k v=\partial_{kk}v=0$,
\begin{equation}
    \mathcal R[\beta;w]
    =
    \frac{\beta\,p}{w}\!\left(1-\frac{kp}{2w}\right),
    \tag{6.5}\label{eq:R-constant-v}
\end{equation}
which is nonzero in general: even the SSR transport generates a nonzero
residual due to the explicit $k$-dependence of the Gatheral functional.
\end{lemma}
\begin{proof}
The functional $g$ from~\eqref{eq:g-functional} depends on $k$ both
explicitly (through the factor $k$ in $(1-kp/2w)^2$) and implicitly through
$w(k,T)$, $p=\partial_k w$, $c=\partial_{kk}w$. Write $g=g(k,w,p,c)$.

Set $r:=\partial_{kkk}w$. Applying $\partial_k$ and $\partial_{kk}$ to
$\partial_u w=vp$ gives
$\partial_u p = vc+(\partial_k v)p$ and
$\partial_u c = vr+2(\partial_k v)c+(\partial_{kk}v)p$.
Differentiating $g=g(k,w,p,c)$ along the flow, noting that
$u$-differentiation does not act on the explicit $k$ in $g$:
\begin{align*}
    \partial_u g
    &=
    (\partial_w g)(vp)
    +(\partial_p g)\bigl[vc+(\partial_k v)p\bigr]
    +(\partial_c g)\bigl[vr+2(\partial_k v)c+(\partial_{kk}v)p\bigr].
\end{align*}

The \emph{total} $k$-derivative of $g$ is
\[
    \partial_k g
    =
    \partial_k^{\mathrm{exp}} g
    +(\partial_w g)p+(\partial_p g)c+(\partial_c g)r.
\]
Grouping all terms multiplied by $v$ in $\partial_u g$ yields
$v[\partial_k g - \partial_k^{\mathrm{exp}} g]$. The remaining terms give
$\mathcal R[v;w]$ as in~\eqref{eq:g-residual}.
Setting $\partial_k v=\partial_{kk}v=0$ in~\eqref{eq:g-residual} reduces
$\mathcal R[\beta;w]$ to $-\beta\,\partial_k^{\mathrm{exp}} g
=\beta p/w(1-kp/(2w))$, confirming~\eqref{eq:R-constant-v}.
\end{proof}

\begin{proposition}[Local butterfly preservation]
\label{prop:butterfly-preserved}
Assume $g[w_0]\ge g_*>0$ pointwise on the window $Q$ of
Section~\ref{subsec:geometric-viewpoint}, and that
along the characteristic flow generated by $v$ the quantities
$w$, $1/w$, $p=\partial_k w$, $c=\partial_{kk}w$, $kp/w$,
$\partial_k v$, $\partial_{kk}v$
remain uniformly bounded. Then $\mathcal R[v;w_u]$ from~\eqref{eq:g-residual}
is uniformly bounded, say by $\|\mathcal R\|_\infty$, and
\begin{equation}
    g[w_u](k,T)
    \;\ge\;
    g_*-\|\mathcal R\|_\infty\,|u|,
    \tag{6.6}\label{eq:butterfly-gronwall}
\end{equation}
so that $g[w_u]\ge g_*/2>0$ for $|u|\le g_*/(2\|\mathcal R\|_\infty)$.
\end{proposition}
\begin{proof}
Along characteristics $\dot k_\tau=-v(k_\tau,\tau,T;w_\tau)$, equation~\eqref{eq:g-transport}
reduces to $\frac{d}{d\tau}g[w_\tau](k_\tau,T)=\mathcal R[v;w_\tau](k_\tau,T)$.
Each term in the corrected residual~\eqref{eq:g-residual} is bounded under
the stated hypotheses: the term $-v\,\partial_k^{\mathrm{exp}}g
=vp/w(1-kp/(2w))$ is bounded since $v$, $p/w$, $kp/w$ are all bounded;
the remaining terms involve products of $\partial_k v$ or $\partial_{kk}v$
with $c=\partial_{kk}w$, all bounded by assumption. Hence
$|\mathcal R|\le\|\mathcal R\|_\infty$
along the flow, and integrating gives~\eqref{eq:butterfly-gronwall}.
\end{proof}

\begin{remark}[Constant-velocity case]
For $v\equiv\beta$, the residual is $\mathcal R[\beta;w]
=\beta p/w(1-kp/(2w))$ from~\eqref{eq:R-constant-v}. Its uniform bound
$\|\mathcal R[\beta;\cdot]\|_\infty$ depends on the surface through
$\sup|p/w|$ and $\sup|kp/w|$. For SVI surfaces satisfying the
no-static-arbitrage conditions of Gatheral--Jacquier~\cite{GatheralJacquier2014},
these quantities are bounded explicitly in terms of the SVI parameters,
giving a computable lower bound for the local admissibility radius
$\delta_{\mathrm{SSR}}=g_*/(2\|\mathcal R[\beta;\cdot]\|_\infty)$.
\end{remark}

\subsection{Main Arbitrage-Free Transport Theorem}

\begin{theorem}[Local arbitrage-free transport on a compact window]
\label{thm:local-arb-preserving}
Let $Q=[k_-,k_+]\times[T_1,T_2]$ with $T_1>0$ and $\mathcal{X}=C^{2,1}(Q)$ be
as in Section~\ref{subsec:geometric-viewpoint}, and fix an inner window
$Q_\rho:=[k_-+\rho,k_+-\rho]\times[T_1,T_2]$ with $\rho>0$.
Let $w_0\in\mathcal{A}^\circ_Q$ with $\partial_T w_0\ge q_*>0$ and
$g[w_0]\ge g_*>0$ on $Q$. Conditions
\emph{(V1)--(V3)} are regularity and boundedness conditions on the velocity
field $v$ on $Q$; condition \emph{(W)} is a boundedness condition on the
initial surface $w_0$. Let $v=v(k,u,T,w)$ depend on $w$ \emph{pointwise}, as in
Lemma~\ref{lem:calendar-transport}, and satisfy
\begin{itemize}
\item[\emph{(V1)}] $v\in C^2$ jointly in $(k,w)$ and $C^1$ in $(u,T)$;
  the partial derivative $v_w:=\partial v/\partial w$ is bounded;
\item[\emph{(V2)}] $v$, $\partial_k v$, $\partial_{kk}v$ bounded on $Q$,
  where $\partial_k v=\partial_k^{\mathrm{exp}}v+v_w\,\partial_k w$ denotes
  the total $k$-derivative at the evaluation point $w=w(k,T;u)$;
\item[\emph{(V3)}] $\partial_T v$ bounded on $Q$;
\end{itemize}
and let $w_0$ satisfy
\begin{itemize}
\item[\emph{(W)}] $w_0\ge\underline w>0$ pointwise and
  $\partial_k w_0$, $\partial_{kk}w_0$, $k\,\partial_k w_0/w_0$ bounded on
  $Q$ (ensures the residual $\mathcal R[v;w]$
  in~\eqref{eq:g-residual}, including the $(\partial_k v)\,c$ term,
  is uniformly bounded along the flow).
\end{itemize}
Then, with
\[
    \delta_\star
    :=
    \min\bigl\{\delta_1,\;\delta_2,\;\delta_3,\;\rho/\|v\|_\infty\bigr\}>0,
\]
there is a unique $w\in C^1\bigl((-\delta_\star,\delta_\star);
C^{2,1}(Q_\rho)\bigr)$ solving
$\partial_u w_u=v(\cdot,u,\cdot;w_u)\,\partial_k w_u$, $w|_{u=0}=w_0$, with
$w_u\in\mathcal{A}^\circ_{Q_\rho}$ for every $|u|<\delta_\star$. Here
$\delta_1$ and $\delta_3$ are the calendar- and butterfly-admissibility radii
of Propositions~\ref{prop:calendar-preserved}
and~\ref{prop:butterfly-preserved}, $\delta_2$ is the Jacobian-positivity
radius established in the proof, and the term $\rho/\|v\|_\infty$ confines the
characteristics through $Q_\rho$ to $Q$.
\end{theorem}
\begin{proof}
\emph{Existence.}
Since $v$ depends on $w$ pointwise, $w_0(k_0,T)$ is a fixed parameter
along each characteristic, and the ODE
\begin{equation}
    \dot k = -v(k,u,T;\,w_0(k_0,T)),\qquad k(0)=k_0,
    \tag{6.7}\label{eq:char-ode}
\end{equation}
has a unique local solution $k(u;k_0,T)$ by Cauchy--Lipschitz,
since $v$ is $C^1$ in $k$ by (V1).

\emph{Regularity of the characteristic map.}
Differentiating \eqref{eq:char-ode} in $k_0$ gives the Jacobi equation
\[
    \frac{d}{du}\partial_{k_0}k
    = -\partial_k^{\mathrm{exp}}v\cdot\partial_{k_0}k
      - v_w\,\partial_{k_0}w_0,
\]
a linear ODE with bounded coefficients by (V1), (V2), and (W).
A further differentiation in $k_0$ produces a similar ODE for
$\partial_{k_0}^2 k$, with coefficients bounded by (V2) and (W).
Regularity in $T$ follows from (V3) and (W). Hence the flow map
$\Phi_u\colon k_0\mapsto k(u;k_0,T)$ is $C^2$ in $k_0$ and $C^1$ in $T$.

\emph{The characteristic map is a diffeomorphism onto its image.}
Let $J:=\partial_{k_0}k(u;k_0,T)$. The Jacobi equation
\[
    \dot J = -\partial_k^{\mathrm{exp}}v\cdot J - v_w\,\partial_{k_0}w_0
\]
is a linear ODE with coefficients bounded uniformly in $k_0$ by (V2)
and (W): $|\partial_k^{\mathrm{exp}}v|\le C_A$ and
$|v_w\,\partial_{k_0}w_0|\le C_B$.
Since $J(0)=1$, variation of constants gives
$J(u)\ge e^{-C_A|u|}-C_B|u|e^{C_A|u|}$,
which is strictly positive for $|u|<\delta_2$, where $\delta_2>0$ is the
Jacobian-positivity radius, depending only on $C_A$ and $C_B$ (not on $k_0$).
Hence $\partial_{k_0}\Phi_u>0$ uniformly in $k_0$ for $|u|<\delta_2$,
so $\Phi_u$ is strictly increasing and therefore injective.
Since $v$ is bounded, $|\Phi_u(k_0)-k_0|\le\|v\|_\infty|u|$, so for
$|u|<\rho/\|v\|_\infty$ the image $\Phi_u(Q_\rho)$ is contained in $Q$: no
characteristic issuing from the inner window leaves $Q$ before time
$\delta_\star$, and the hypotheses (V1)--(V3), (W) are in force everywhere
along it. A strictly increasing $C^2$ map with positive derivative is a
$C^2$ diffeomorphism onto its image; write $\chi:=\Phi_u^{-1}$, defined on
$\Phi_u(Q_\rho)$.

\emph{Solution construction.} Set
\[
    w_u(k,T) := w_0(\chi(k),T),
    \qquad (k,T)\in Q_\rho .
\]
By the chain rule $w_u\in C^{2,1}(Q_\rho)$. Differentiating
$w_u(k,T)=w_0(\chi(k),T)$ in $u$ and using the characteristic relation
\[
    \dot\chi = -v(k,u,T;w_u)\,\partial_k\chi
\]
gives $\partial_u w_u = v(\cdot,u,\cdot\,;w_u)\,\partial_k w_u$.

\emph{Uniqueness.}
Any $C^{2,1}$ solution $\tilde w$ satisfies $d\tilde w/du=0$ along
its characteristics, so $\tilde w$ is constant along them. These
characteristics satisfy \eqref{eq:char-ode} with $w_0(k_0,T)$ as
initial data, and Cauchy--Lipschitz uniqueness forces
$\tilde w = w_u$.

\emph{Admissibility.}
Propositions~\ref{prop:calendar-preserved} and
\ref{prop:butterfly-preserved} yield $\delta_1>0$ and the
butterfly-positivity radius $\delta_3>0$
with $\partial_T w_u\ge q_*/2>0$ and $g[w_u]\ge g_*/2>0$
for $|u|<\min(\delta_1,\delta_3)$. Positivity holds since
$w_u=w_0\circ\chi>0$. All four constraints therefore hold on $Q_\rho$ for
$|u|<\delta_\star=\min\{\delta_1,\delta_2,\delta_3,\rho/\|v\|_\infty\}$,
which completes the proof.
\end{proof}

\begin{remark}[Extension to $k\in\mathbb R$]
\label{rem:extension-to-R}
The inner-window restriction can be removed when the hypotheses hold with constants independent of the window, that is, when $q_*$, $g_*$ and the
bounds in (V2), (V3) and (W) are uniform in $k$ on
$\mathbb R\times[T_1,T_2]$. The characteristic ODE~\eqref{eq:char-ode} is then
solved for every $k_0\in\mathbb R$; the Jacobian estimate
$J(u)\ge e^{-C_A|u|}-C_B|u|e^{C_A|u|}$ holds uniformly in $k_0$, so $\Phi_u$
is strictly increasing on $\mathbb R$; and since
$|\Phi_u(k_0)-k_0|\le\|v\|_\infty|u|$ one has $\Phi_u(k_0)\to\pm\infty$ as
$k_0\to\pm\infty$, making $\Phi_u$ a strictly increasing surjection of
$\mathbb R$ and hence a global $C^2$ diffeomorphism. The conclusion becomes
$w_u\in\mathcal{A}^\circ$ on $\mathbb R\times[T_1,T_2]$ for
$|u|<\min(\delta_1,\delta_2,\delta_3)$, with no $\rho$-dependent term.
The uniformity is a genuine restriction rather than a technicality: it
requires the skew and curvature bounds of (W) to hold out to arbitrarily
large $|k|$, which constrains the wings of $w_0$. Consistently with
Remark~\ref{rem:wing-conditions}, we do not impose or analyze wing-growth
conditions, and the compact-window form of
Theorem~\ref{thm:local-arb-preserving} is the one used in the rest of the
paper.
\end{remark}

\begin{remark}[Nonlocal velocity is not covered]
\label{rem:nonlocal-v}
The pointwise dependence of $v$ on $w$ is essential to the proof:
it allows $w_0(k_0,T)$ to serve as a fixed parameter in the
characteristic ODE~\eqref{eq:char-ode}, so that each characteristic
is determined by a scalar ODE with no feedback from the evolving
surface. If $v$ depends nonlocally on $w$ (as in local volatility, where $v^{\rm LV}$ is determined by Dupire's formula applied to the full smile), then $w_0(k_0,T)$ is no longer a valid substitute for
$w(k(u),T;u)$ along the characteristic, and the argument breaks down.
Extending the theorem to nonlocal functional dependence is left for
future work. In particular, this theorem makes no statement about
local-volatility dynamics; see Remark~\ref{rem:lv-asymptotic}.
\end{remark}

\begin{corollary}[SSR flow: local admissibility]
\label{cor:ssr-preserves}
Let $\beta\in\mathbb R$ be constant and let $w_0\in\mathcal{A}^\circ_Q$
satisfy (W) with constants $(\underline w, M_p, M_{kp})$ on the window $Q$ of
Section~\ref{subsec:geometric-viewpoint}. Define
$w_u(k,T):=w_0(k+\beta u,T)$, which is well defined on the inner window
$Q_\rho$ for $|\beta u|<\rho$. Then $w_u$ solves
$\partial_u w_u=\beta\,\partial_k w_u$, the calendar density satisfies
$\partial_T w_u(k,T)=\partial_T w_0(k+\beta u,T)>0$, and there
exists $\delta_{\mathrm{SSR}}>0$, depending only on $g_*$, $\beta$,
and $(\underline w, M_p, M_{kp})$, such that
$w_u\in\mathcal{A}^\circ_{Q_\rho}$ for all
$|u|<\min\{\delta_{\mathrm{SSR}},\rho/|\beta|\}$. If the hypotheses hold
uniformly in $k$ on $\mathbb R\times[T_1,T_2]$, then as in
Remark~\ref{rem:extension-to-R} the restriction to $Q_\rho$ is unnecessary
and $w_u\in\mathcal{A}^\circ$ for all $|u|<\delta_{\mathrm{SSR}}$.
\end{corollary}
\begin{proof}
Calendar preservation: the calendar-density equation with $v\equiv\beta$
has forcing term $(\partial_T\beta)\partial_k w=0$, so
$\partial_T w_u(k,T)=\partial_T w_0(k+\beta u,T)>0$ for all $u$.
Butterfly preservation: since $w_u(k,T)=w_0(k+\beta u,T)$, the residual
from~\eqref{eq:R-constant-v} evaluated along the flow is
\[
    \mathcal R[\beta;w_u](k,T)
    =
    \frac{\beta\,p_0(k+\beta u,T)}{w_0(k+\beta u,T)}
    \!\left(1-\frac{k\,p_0(k+\beta u,T)}{2\,w_0(k+\beta u,T)}\right),
\]
where $p_0:=\partial_k w_0$. Writing $k=(k+\beta u)-\beta u$ and using
$|p_0|\le M_p$, $w_0\ge\underline w$, and
$|(k+\beta u)p_0/w_0|\le M_{kp}$, one obtains the uniform bound
\[
    \|\mathcal R[\beta;w_u]\|_\infty
    \;\le\;
    \frac{|\beta|M_p}{\underline w}
    \!\left(1+\frac{M_{kp}}{2}+\frac{|\beta||u|M_p}{2\underline w}\right)
    =: \bar R(|u|),
\]
which is finite and increasing in $|u|$. The Gronwall estimate of
Proposition~\ref{prop:butterfly-preserved} then gives
$g[w_u]\ge g_*/2>0$ for all $|u|<\delta_{\mathrm{SSR}}$, where
$\delta_{\mathrm{SSR}}>0$ is the largest $\delta$ satisfying
$2\bar R(\delta)\,\delta < g_*$.
\end{proof}

\begin{remark}[Global SSR admissibility is not automatic]
\label{rem:global-ssr}
The SSR ride formula does \emph{not} preserve $g$ by pure transport:
$g[w_u](k,T)\neq g[w_0](k+\beta u,T)$ in general, since $g$ depends on $k$
explicitly through $(1-kp/(2w))^2$, and the two sides evaluate that factor at
$k$ and at $k+\beta u$ respectively. Global admissibility for all
$u\in\mathbb R$ therefore needs $g[w_0(\cdot+\beta u,T)](k,T)>0$ for every
$(k,T,u)$, a non-trivial constraint on $w_0$ to be checked case by case, analytically for SVI surfaces satisfying
Gatheral--Jacquier~\cite{GatheralJacquier2014}.
Figure~\ref{fig:gatheral-density} shows the transported density staying
positive for all tested $\beta$ while not being a translate of the initial
one, confirming that $g$ is not translation-equivariant in $k$.
\end{remark}

\subsection{Admissible Sources}
Theorem~\ref{thm:local-arb-preserving} governs the flow generated by $v$: a
pushforward along characteristics carries the butterfly and calendar
structure with it, which is why the argument works. A source acts
\emph{additively} on $w$ and is not covered by that result. The following
completes the section for the full decomposition $(5.3c)$, in the form in
which it is used in practice: a finite straight-line bump $w+us$ of a
quoted surface. We state it on a compact strike window and a maturity strip
bounded away from zero, where all the constants are finite, and record
separately what is needed at the short end.

\begin{theorem}[Admissible source, local form]
\label{thm:admissible-source}
Work on a compact window $[k_-,k_+]\times[T_1,T_2]$ with $T_1>0$. Assume
$w\ge m>0$, $g[w]\ge\delta>0$ and $\partial_T\widetilde w\ge\delta_c>0$
there, where $\widetilde w$ is the fixed-strike surface, and assume $s$ has
continuous bounded strike derivatives through order two and a continuous
bounded maturity derivative in fixed-strike coordinates, denoted
$\partial_T\widetilde s$.
\begin{itemize}
    \item[(i)] Choose $\bar u>0$ with $\bar u\|s\|_\infty\le m/2$ and set
    $C_g:=\sup_{|u|\le\bar u}\bigl\|\tfrac{\mathrm d}{\mathrm du}g[w+us]\bigr\|_\infty$
    and $C_c:=\|\partial_T\widetilde s\|_\infty$. Then with
    \[
        u_0:=\min\Bigl\{\bar u,\;\frac{\delta}{2C_g},\;\frac{\delta_c}{2C_c}\Bigr\}
    \]
    (zero denominators read as $+\infty$), the surface $w+us$ satisfies all
    three constraints of Definition~\ref{def:admissible-explicit} on the
    window for every $|u|\le u_0$.
    \item[(ii)] To extend a fixed nonzero bump of both signs down to zero
    maturity, positivity requires $|s|\le w/u_0$, in particular $s=O(T)$
    wherever $w=O(T)$. A linear taper defined in fixed-strike coordinates,
    $\widetilde s^{\,\mathrm{ext}}(K,T):=\min(1,T/T_1)\,\widetilde s(K,T\vee T_1)$,
    has maturity derivative bounded almost everywhere by
    $\max\{\|\widetilde s(\cdot,T_1)\|_\infty/T_1,\sup_{T>T_1}\|\partial_T\widetilde s\|_\infty\}$,
    which yields a sufficient calendar bound after multiplication by $|u|$
    and comparison with the calendar margin. Positivity and butterfly
    admissibility near $T=0$ require their own uniform estimates; tapering
    alone does not supply them.
    \item[(iii)] If the butterfly constraint is active, $g[w]=0$ on a set
    $A$, then admissibility of the one-sided bump $u\downarrow 0$ implies
    $\mathcal L_w s\ge 0$ on $A$, where $\mathcal L_w$ is the linearization
    of $g$ at $w$. This first-order condition is necessary but not
    sufficient when equality holds: second-order terms can violate the
    constraint, and a sufficient condition for a finite bump must control
    the Taylor remainder throughout the window.
\end{itemize}
\end{theorem}
\begin{proof}
(i) Positivity of $w+us$ follows from $|u|\|s\|_\infty\le m/2$. On that
neighborhood the rational functional $g$ is continuously differentiable in
$u$, so $g[w+us]\ge\delta-|u|C_g\ge\delta/2$ for $|u|\le\delta/(2C_g)$.
Separately, $\partial_T(\widetilde w+u\widetilde s)\ge\delta_c-|u|C_c\ge\delta_c/2$
for $|u|\le\delta_c/(2C_c)$. (ii) Positivity for both $u=\pm u_0$ requires
$w\ge u_0|s|$. The taper's maturity derivative is $\widetilde s(K,T_1)/T_1$
below $T_1$ and $\partial_T\widetilde s$ above it; the taper is continuous
at $T_1$, and a jump in its derivative is not a singular mass in the first
derivative, so the almost-everywhere bound stated gives the calendar
estimate. (iii) At an active point, $g[w+us]=u\,\mathcal L_w s+O(u^2)$;
divide by $u>0$ and let $u\downarrow 0$ to obtain necessity. Equality leaves
the sign of the quadratic term undetermined. For a concrete failure of
sufficiency take, on a small strike window, $w(k)=1-k^2$, for which
$g[w](k)=k^2\bigl[(1-k^2)^{-2}-\tfrac14\bigr]>0$ has active point $k=0$; for
$s(k)=k$ one computes $\mathcal L_w s(0)=0$ but
$g[w+us](0)=-5u^2/16<0$ for every $u\neq 0$. A tangent direction need not
define an admissible straight bump.
\end{proof}

\begin{remark}
\label{rem:admissible-source-combined}
Combining the two theorems: a decomposition $(5.3c)$ whose velocity satisfies
(V1)--(V3) and (W) and whose source satisfies (i) generates an admissible
surface for $|u|$ below the smaller of the two radii, on the window where
both are established. Both parts are verifiable on a given surface, which was
the design goal here. In the empirical implementation the source is tapered
linearly below the first quoted tenor as in (ii) and the butterfly margin is
checked at the working bump; those checks verify (i) on the quoted window and
do not by themselves establish the short-end statements of (ii).
\end{remark}

\begin{figure}[!htbp]
\centering
\includegraphics[width=\textwidth]{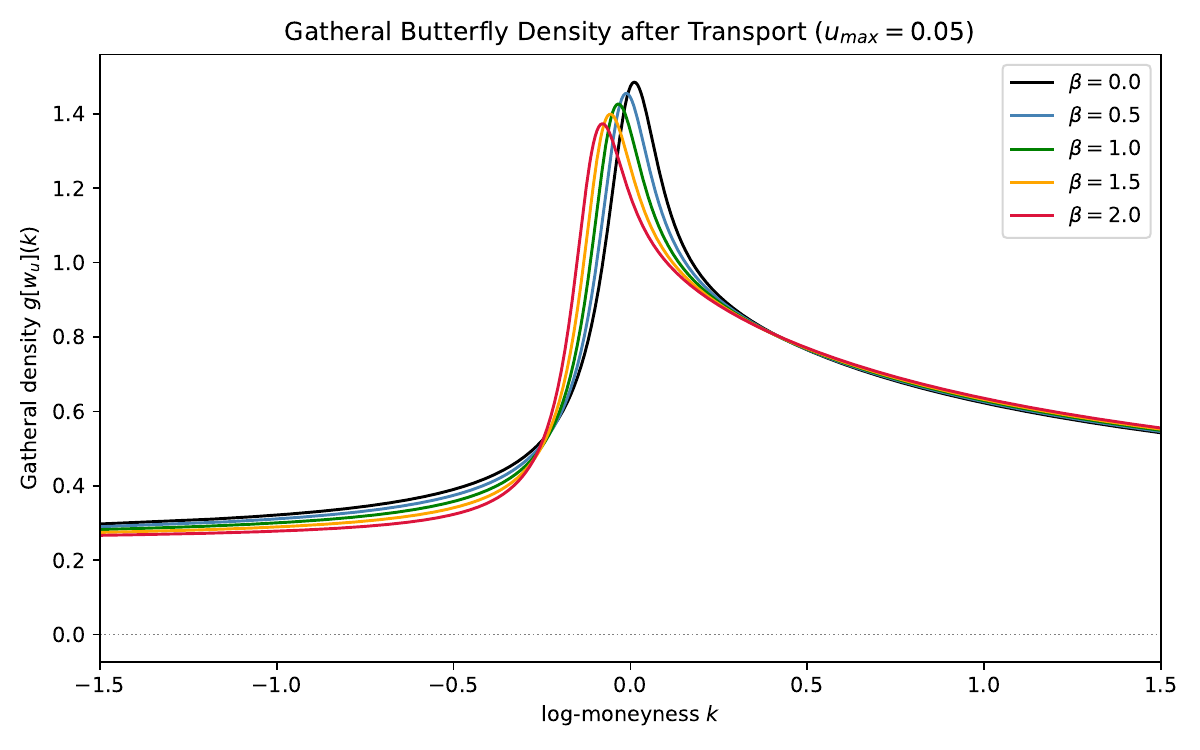}
\caption{Gatheral butterfly density $g[w_u](k)$ evaluated on the transported
SVI surface after a log-forward move of $u=0.05$ (a forward move of
approximately $5.1\%$), for $\beta\in\{0,0.5,1,1.5,2\}$.
The density remains strictly positive for all $\beta$ across the full grid
$k\in[-1.5,1.5]$, confirming local butterfly-arbitrage preservation
(Theorem~\ref{thm:local-arb-preserving} and
Corollary~\ref{cor:ssr-preserves}). The minimum over the grid decreases
monotonically in $\beta$, from $0.297$ at $\beta=0$ to $0.267$ at
$\beta=2$, so it stays well above zero at every tested speed.  Note that $g[w_u](k)\neq g[w_0](k+\beta u)$
(the curves are not simply translates of the initial density), empirically
confirming Remark~\ref{rem:global-ssr}: the Gatheral functional is not
translation-equivariant in $k$.}
\label{fig:gatheral-density}
\end{figure}

\begin{figure}[!htbp]
\centering
\includegraphics[width=0.65\textwidth]{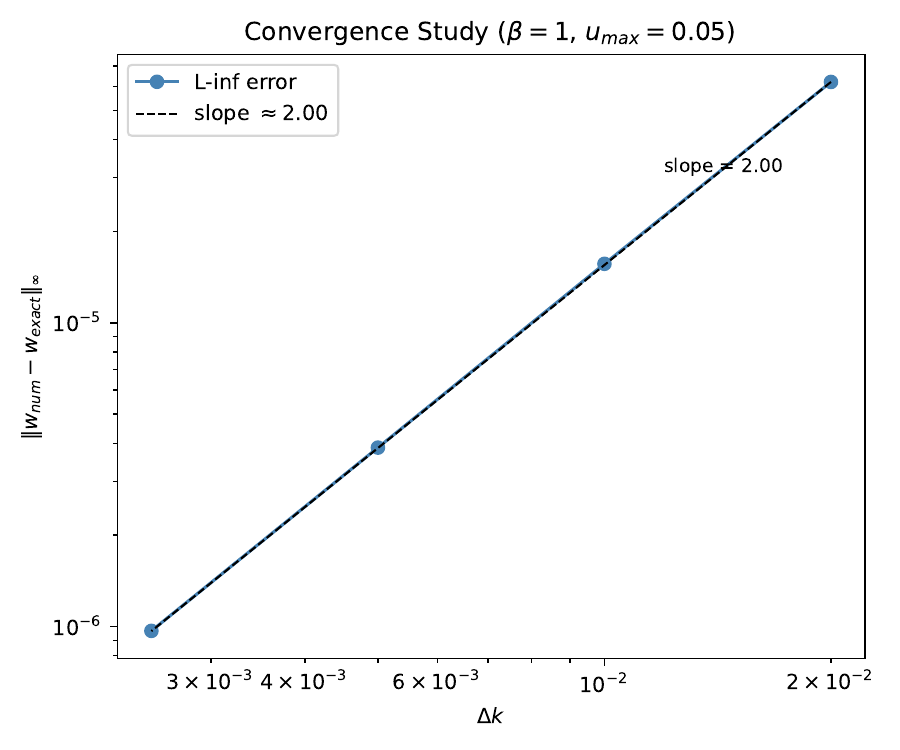}
\caption{Convergence of the second-order Beam-Warming upwind solver for the
SSR transport equation $\partial_u w=\beta\,\partial_k w$.  Log-log plot of
the $L^\infty$ error against grid spacing $\Delta k$
($N_k\in\{151,301,601,1201\}$).  The fitted slope is $2.002$, confirming
second-order convergence as expected from the Beam-Warming spatial
discretization combined with RK2 time-stepping.}
\label{fig:convergence}
\end{figure}

\FloatBarrier
\section{Classical Smile Dynamics as Special Cases}
\label{sec:classical-special-cases}
The classical smile-dynamics models can each be identified with a particular
choice of transport velocity field $v$. Sticky delta corresponds to
$v^{\rm SD}\equiv 0$ (the floating smile is frozen) and sticky strike to
$v^{\rm SS}\equiv 1$ (the fixed-strike surface is frozen); these were the
earliest empirical descriptions~\cite{Derman1999}. The SSR
model~\cite{Bergomi2004,Bergomi2005,Bergomi2009,BergomiBook} takes $v^{\rm SSR}\equiv\beta$ constant
and yields the global translation of Corollary~\ref{cor:ssr-preserves};
the transport equation $\partial_u w=\beta\,\partial_k w$ is solved
numerically using a second-order Beam--Warming scheme whose
$L^\infty$ convergence rate is confirmed in Figure~\ref{fig:convergence}.

Local volatility enters differently, and we are careful not to present it as
a further choice of velocity field. The Dupire
formula~\cite{Dupire1994} does induce a spot-dependent $v^{\rm LV}$, whose
short-maturity at-the-money limit is $2$: by the classical half-slope
relation~\cite{BerestyckiBuscaFlorent2002,Gatheral2006,BayerFrizEtAl2019},
$\partial_k\sigma^{\rm imp}(0,T)\to\tfrac{1}{2}\partial_k\sigma^{\rm LV}(S)$
as $T\downarrow 0$, so the ATM implied skew is half the local-vol skew, and
since $\beta=\Delta\sigma^*/S_\sigma^{\rm imp}$ the smaller denominator gives
$\beta_{\rm LV}\to 2$, consistent with the $\beta=2$ panel of
Figure~\ref{fig:transported-smiles}. That limit is the extent of the
correspondence. The induced $v^{\rm LV}$ is nonlocal in $w$, so it falls
outside Theorem~\ref{thm:local-arb-preserving}
(Remark~\ref{rem:nonlocal-v}), and a local-volatility model has no reason to
satisfy the constraint that pure transport imposes at a smile extremum, so it
generically carries a source rather than lying in $\mathcal{T}_w$;
Remark~\ref{rem:lv-asymptotic} sets this out.
Durrleman~\cite{Durrleman2010} derives the PDE structure of implied-volatility
surface dynamics explicitly from local volatility, establishing results
closely related to the transport-equation representation developed here;
the distinction is that Durrleman's velocity is a model-specific output
determined by Dupire inversion of the local-vol surface, whereas here $v$ is a
free organizing object, and the price of that freedom is that $v$ must be
estimated jointly with a source rather than read off by inversion
(Proposition~\ref{prop:identification-without-inversion}).
The quadratic ATM jet $v(k,T)=\beta+\eta k+\tfrac{\zeta}{2}k^2$ generates,
via Theorem~\ref{thm:jet-transport}, the coupled level--skew--curvature
dynamics~\eqref{eq:jet-three}. Finally, in rough Bergomi
models~\cite{BergomiGuyon2012,ElEuchFukasawaRosenbaum2018,Fukasawa2021} the velocity is
stochastic with short-maturity expectation $\mathbb E[v(0,T)]\to H+\tfrac{3}{2}$,
recovering the classical SSR short-maturity limit. Table~\ref{tab:taxonomy}
collects these correspondences.

\begin{table}[ht]
\centering
\begin{tabular}{lll}
\toprule
Dynamics & Velocity field $v(k,T;w)$ & Notes \\
\midrule
\multicolumn{3}{l}{\emph{(a) Pure transports: the response is a multiple of the skew by construction}}\\
Sticky delta          & $0$             & floating smile frozen \\
Sticky strike         & $1$             & fixed-$K$ surface frozen \\
SSR (Bergomi)         & $\beta$ const.  & global translation; locally admissible (Corollary~\ref{cor:ssr-preserves}) \\
Jet expansion         & $\beta+\eta k+\tfrac{\zeta}{2}k^2$ & skew/curvature transport \\
\midrule
\multicolumn{3}{l}{\emph{(b) Second-order extension: first-order part plus an It\^o correction}}\\
Stochastic-$F$ It\^o  & $v$ & class~\eqref{eq:transport-with-diffusion}, $D^2{=}(\sigma^F_t)^2 v^2{\ge}0$; Prop.~\ref{prop:stochastic-ito} \\
\midrule
\multicolumn{3}{l}{\emph{(c) Generative models: only the short-maturity ATM velocity is reproduced}}\\
Local vol (Dupire)    & $v(0,T)\to 2$ as $T\downarrow0$ & nonlocal in $w$; generically sourced (Rem.~\ref{rem:lv-asymptotic}) \\
Rough vol (rBergomi)  & $\mathbb E[v(0,T)]\to H+\tfrac{3}{2}$ & $T\downarrow0$; see~\cite{ElEuchFukasawaRosenbaum2018,BergomiGuyon2012,Fukasawa2021} \\
\bottomrule
\end{tabular}
\caption{Smile dynamics and the transport velocity field. The rules in
block~(a) \emph{are} transport laws: each prescribes $\partial_u w$ as a
multiple of $\partial_k w$, so they lie in $\mathcal{T}_w$ with $s\equiv 0$ by
construction. Block~(b) is the second-order class that a stochastic forward
induces, which contains the first-order transports as the $D^2\to0$ limit.
The models in block~(c) are not special cases of the framework: they induce a
velocity as a model-specific, nonlocal output, and what the framework
reproduces is only its short-maturity at-the-money limit; see
Remark~\ref{rem:lv-asymptotic}.}
\label{tab:taxonomy}
\end{table}

For block~(a) the decomposition of
Section~\ref{subsec:transport-and-source} lets the list be sharpened into an
equivalence.

\begin{proposition}[The exact classical rules are the constant-velocity pure transports]
\label{prop:classical-pure-transport}
Let the surface response admit the decomposition $(5.3c)$ on
$\mathcal{A}^\circ$. Then the dynamics is sticky delta, sticky strike, or the
Bergomi SSR if and only if
$s\equiv 0$ and $v$ is constant in $k$, with $v\equiv 0,1,\beta$
respectively. Equivalently: these three regimes are
exactly the elements of $\mathcal{T}_w$ with constant velocity, and every
departure from them is a strike-dependent velocity, a source, or both.
\end{proposition}
\begin{proof}
Sufficiency is the content of Table~\ref{tab:taxonomy} together with
Corollary~\ref{cor:ssr-preserves}. For necessity, each classical regime
prescribes the response pointwise for every admissible surface: sticky
delta sets $\partial_u w|_k\equiv 0$, giving $v\,\partial_k w+s\equiv 0$ for
all $w\in\mathcal{A}^\circ$, hence $s\equiv 0$ and $v\equiv 0$ since
$\partial_k w$ is not identically zero on any open set of surfaces; sticky
strike sets $\partial_u w|_K\equiv 0$, i.e.\ $\psi\equiv 0$, giving
$v\equiv 1$, $s\equiv 0$ by $(5.3)$; and the SSR rule prescribes
$\partial_u w=\beta\,\partial_k w$ by definition. In each case the
prescription holds for every surface in $\mathcal{A}^\circ$, which forces
the source to vanish identically rather than merely on one surface.
\end{proof}

\begin{remark}[Local volatility is not a special case of the transport framework]
\label{rem:lv-asymptotic}
Table~\ref{tab:taxonomy} should not be read as exhibiting local volatility as
one more choice of velocity field, and the distinction matters for the logic
of the paper.

Sticky delta, sticky strike and the SSR are pure transports \emph{by
definition}: each prescribes $\partial_u w$ as a multiple of the skew ($0$, $\partial_k w$ and $\beta\,\partial_k w$ respectively), so membership of
$\mathcal{T}_w$ with $s\equiv0$ is immediate rather than a discovery, which is
why Proposition~\ref{prop:classical-pure-transport} can say ``exactly''.

Local volatility differs on three counts. First, $v^{\rm LV}$ is not a free
field but a model output, strike- and maturity-dependent in general. Second,
that output is \emph{nonlocal} in $w$, since Dupire's formula involves the
whole smile, so Theorem~\ref{thm:local-arb-preserving} does not cover it
(Remark~\ref{rem:nonlocal-v}): the framework's central admissibility result
does not apply to local volatility at all. Third, and most consequentially
for $(5.3c)$, pure transport with bounded velocity requires
$\partial_u w(k^\ast)=0$ whenever $\partial_k w(k^\ast)=0$. In a
local-volatility model the response and the implied skew are distinct
functions of $k$, generated by the same local-volatility function but with no
structural reason to share zeros, so such a model has no reason to satisfy
that constraint and generically carries a source. We state this as a
structural observation, not a theorem: for a given surface it is a finite
computation to check, and we prove no genericity statement.

What the framework reproduces is the short-maturity at-the-money limit of the
induced velocity, $\lim_{T\downarrow0}v^{\rm LV}(0,T)=2$ by the half-slope
relation~\cite{BerestyckiBuscaFlorent2002,Gatheral2006,BayerFrizEtAl2019},
and $\mathbb E[v]\to H+\tfrac32$ for rough Bergomi; these are consistency checks on
the velocity scale, not claims that either model is a pure transport. This
also makes the result of Section~\ref{subsec:exp5} less surprising than it
may appear: if generative models of this type generically carry a source,
what the test contributes is where the source sits and how large it is, not
that it is nonzero.
\end{remark}

\FloatBarrier
\section{Empirical Evidence from SPX}
\label{sec:spx-empirical}
The theoretical framework of Sections~\ref{sec:general-transport}--\ref{sec:classical-special-cases}
makes testable predictions about observed implied volatility dynamics. We
now present an empirical study using SPX data spanning July 2021
to July 2026. The estimand $\hat v_0(k_0)$ is related in spirit to the
path-dependent regression of Guyon and Lekeufack~\cite{GuyonLekeufack2023},
who document that SPX implied-volatility dynamics are not captured by a
scalar SSR; the present framework provides a parameterization of the
full $k$-dependent velocity field $v(k)$ of which the scalar SSR is the
zeroth-order truncation.

\subsection{Data and Methodology}
\label{subsec:spx-data}

The dataset consists of daily SPX implied-volatility surface snapshots
spanning July 12, 2021 to July 9, 2026 ($n=1{,}118$ usable observations
per tenor after filtering $|\Delta u|\ge 0.001$ and polynomial
$R^2\ge 0.95$), sampled at seven tenors
$T\in\{1\mathrm{M},2\mathrm{M},3\mathrm{M},6\mathrm{M},9\mathrm{M},
12\mathrm{M},24\mathrm{M}\}$. The return filter conditions on the magnitude
of the regressor, so it changes the composition of the regression population
and not merely its size; the reported coefficients should be read as
conditional on it. All rolling statistics reported below use a
60-trading-day window. Moneyness is measured relative to the tenor
forward $F(T)=S_0\,e^{(r-y)T}$, where $r$ is the risk-free
rate and $y$ is the dividend yield, consistent with the forward-based
martingale structure of the theory; the strike grid contains 11 moneyness
levels $K/F\in\{0.85, 0.90, 0.925, 0.95, 0.975, 1.0, 1.025, 1.05, 1.075,
1.10, 1.15\}$, spanning a log-moneyness range of approximately
$[-0.16, +0.14]$.

The ATM variance jet is extracted by fitting a cubic Taylor polynomial
to the cross-section of total variance $w_i = \sigma_i^2 T$ in
forward log-moneyness $k_i = \log(K_i/F(T))$ at each date and tenor,
requiring $R^2\ge 0.95$ per fit:
\begin{equation}
    w(k) \approx a_0 + a_1 k + \frac{a_2}{2}\,k^2 + \frac{a_3}{6}\,k^3,
    \tag{8.1}\label{eq:jet-poly-fit}
\end{equation}
with factorial scaling $a_0=c_0$, $a_1=c_1$, $a_2=2c_2$, $a_3=6c_3$, so
that the fitted coefficients $a_n$ equal the ATM derivatives
$\partial_k^n w(0)$ directly, matching the definition of the ATM variance
jet in $(5.8)$. The degree-3 polynomial identifies
$a_0,a_1,a_2,a_3$, so all four jets required for the sequential
$(\beta,\eta,\zeta)$ estimator are available. For the local-jet experiments, expansion centers are placed at
$K/F\in\{0.905, 0.928, 0.951, 0.975, 1.000, 1.025, 1.051\}$ across multiple tenors.

The sequential estimation protocol from Section~\ref{sec:sequential_estimation}
is applied directly to the extracted smile jets, with $\Delta u = \Delta\log F $
denoting the daily log-forward increment.

\paragraph{Inference.}
Four variance estimators appear below, and the choice in each case is driven by
what the estimand depends on.
\emph{(i) HC3} is used for the single-equation regressions of Experiments~1 and~3, the scalar SSR and the local level velocity at a fixed $k_0$. These
are OLS slopes on a single regressor, where the leading concern is
heteroskedasticity in $\Delta u$ and HC3 is the standard small-sample
correction.
\emph{(ii) The pairs bootstrap} ($B=500$) is primary for the sequential
forward-substitution estimates of Experiment~2. Here $\hat\eta$ and
$\hat\zeta$ are functions of previously estimated coefficients \emph{and} of
the realized jet design $(a_1,a_2,a_3)$, so their sampling variation is
dominated by variation in that design, most acutely at 1M, where the
identifying regressor $a_1\Delta u$ is near zero. Resampling $(y_t,x_t)$ pairs
propagates design variation through the whole substitution chain, which a
residual-based scheme conditioning on the observed design does not;
Section~\ref{subsec:mc-design} quantifies the gap, a factor of $12.9$ for
$\hat\zeta$ at 1M.
\emph{(iii) The moving-block bootstrap} ($b=5$ days) is used wherever the
statistic depends on a covariance matrix across strikes or across equations (the strike-constancy Wald of Proposition~\ref{prop:chi-square-test} and the DGP of Section~\ref{sec:mc-validation}), because there the serial
dependence of overlapping surface fits is precisely what must be preserved.
\emph{(iv) Newey--West} (5 lags) is used in Experiment~5, where
\eqref{eq:identification-vs} is fitted at a single node and a parametric lag
correction suffices.
These choices are in tension at (ii): pairs resampling does not preserve
serial dependence, so the Experiment~2 standard errors could be understated.
Table~\ref{tab:jet-sequential} therefore reports the moving-block alternative as a check (on average $16\%$ larger, with no coefficient changing significance at the $5\%$ level), and we retain pairs as primary because
design variation, not serial correlation, is the binding constraint on these
particular estimates.

\paragraph{Generated regressors.}
The jets $a_0,\ldots,a_3$ and their increments are not observed: they are
coefficients of the per-date fit \eqref{eq:jet-poly-fit}, so every regression
below has regressors measured with error from the same fit that produced the
dependent variable. Three consequences should be carried through the results.
First, the measurement errors in $\Delta a_n$ and in $a_{n+1}$ are correlated
by construction, since both derive from the same quotes on the same date;
the classical errors-in-variables attenuation argument therefore does not
apply, and the sequential estimates
\eqref{eq:ssr_regression}--\eqref{eq:curv_regression} may be biased in either
direction. Second, resampling does not remove this bias: every pairs or block
resample inherits the same fitted jets, so the bootstrap is centered on the
pseudo-true value the generated regressors imply rather than on
$(\beta,\eta,\zeta)$, and the intervals we report should be read as sampling
intervals around that pseudo-true value. Third, the problem is worst where the
identifying regressor is smallest, which is $\hat\zeta$ at 1M
($a_1\approx-0.019$); this is a second reason, independent of the variance
argument in Section~\ref{subsec:mc-design}, to read the 1M curvature estimate
cautiously. The tracked-minimum test of Section~\ref{subsec:exp5} is the one
result in the paper free of this channel, since it uses no estimated skew as a regressor (only the location of the zero crossing, which enters at second order), and that is part of why we treat it as the sharpest evidence.

\subsection{Experiment 1: SSR Term Structure}
\label{subsec:exp1-ssr}

The classical regression $\Delta a_0 = \beta(T)\cdot a_1\cdot\Delta u
+ \varepsilon$ yields the following estimates ($n=1{,}118$, HC3 standard
errors):

\begin{table}[H]
\centering
\begin{tabular}{lccccccc}
\toprule
Tenor & $\hat\beta$ & HC3 SE & $t$-stat & $p$-value & $R^2$ & Rolling mean & Rolling std \\
\midrule
1M  & 1.4375 & 0.1240 & 11.593 & $<0.001$ & 0.796 & 1.4205 & 0.3983 \\
2M  & 1.3757 & 0.0567 & 24.244 & $<0.001$ & 0.800 & 1.3161 & 0.3681 \\
3M  & 1.3357 & 0.0480 & 27.827 & $<0.001$ & 0.803 & 1.2649 & 0.3456 \\
6M  & 1.2291 & 0.0365 & 33.703 & $<0.001$ & 0.810 & 1.1726 & 0.2890 \\
9M  & 1.1590 & 0.0313 & 37.024 & $<0.001$ & 0.816 & 1.1151 & 0.2532 \\
12M & 1.1145 & 0.0280 & 39.849 & $<0.001$ & 0.818 & 1.0765 & 0.2227 \\
24M & 1.0114 & 0.0242 & 41.829 & $<0.001$ & 0.792 & 0.9972 & 0.1627 \\
\bottomrule
\end{tabular}
\caption{SSR estimates from $\Delta a_0 = \beta\cdot a_1\cdot\Delta u$,
HC3 standard errors, $n=1{,}118$. Rolling statistics use a 60-trading-day
window.}
\label{tab:ssr-term-structure}
\end{table}

\begin{figure}[!htbp]
\centering
\includegraphics[width=\textwidth]{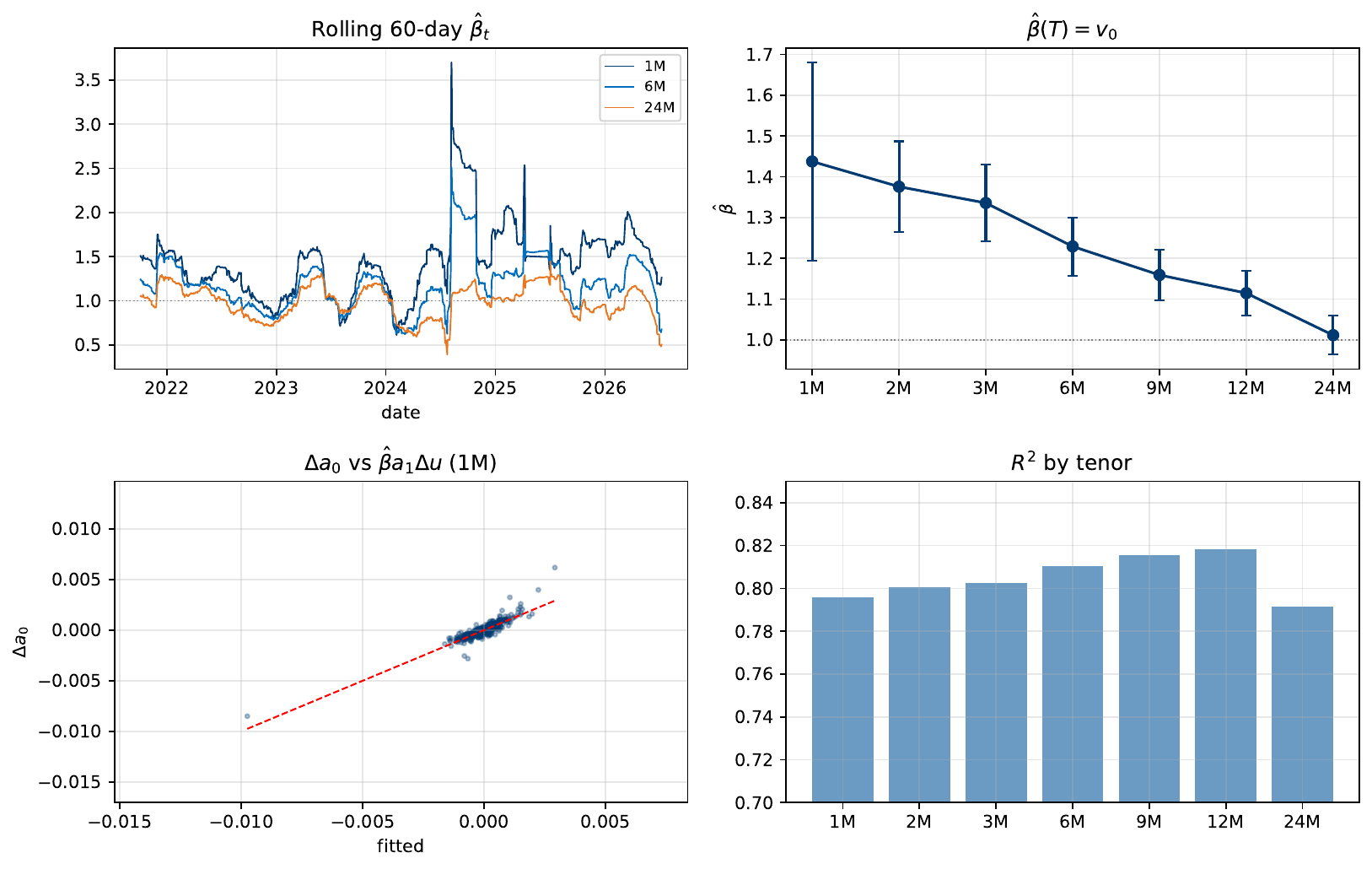}
\caption{SPX SSR (Experiment 1): rolling 60-day estimates (top left), term
structure $\hat\beta(T)$ (top right), regression scatter plot (bottom
left), and $R^2$ by tenor (bottom right). Data: July 2021--July 2026,
$n=1{,}118$. $\hat\beta>1$ at all tenors through 12M and $\hat\beta\approx
1$ at 24M; the term structure is monotonically decreasing from $1.4375$
at 1M to $1.0114$ at 24M.}
\label{fig:spx-exp1}
\end{figure}

$\hat\beta>1$ at all tenors through 12M, and $\hat\beta\approx 1$ at 24M.
The term structure is monotonically decreasing from $1.4375$ at 1M to
$1.0114$ at 24M, and $R^2$ is stable at $0.79$--$0.82$ across tenors;
Figure~\ref{fig:spx-exp1} shows the rolling estimates, the term structure,
the regression scatter and the $R^2$ profile.

\subsection{Experiment 2: Jet Transport Hierarchy}
\label{subsec:exp2-jet}

The primary estimator is sequential forward substitution (FS) with pairs
bootstrap standard errors ($N_{\mathrm{boot}}=500$). Because the cubic
polynomial fit identifies $a_3$, the curvature-transport coefficient
$\zeta$ is estimated from the third-order equation and reported alongside
$\beta$ and $\eta$.

\begin{table}[H]
\centering
\small
\begin{tabular}{lccccccccc}
\toprule
Tenor & $\hat\beta$ & SE$(\hat\beta)$ & $p(\beta)$ & $\hat\eta$ & SE$(\hat\eta)$ & $p(\eta)$ & $\hat\zeta$ & SE$(\hat\zeta)$ & $p(\zeta)$ \\
\midrule
1M  & 1.4375 & 0.0969 & $<0.001$ &  1.9308 & 1.1782 & 0.101    & 61.632 & 30.935 & 0.046    \\
2M  & 1.3757 & 0.0542 & $<0.001$ &  1.2766 & 0.3590 & $<0.001$ & 49.284 & 12.383 & $<0.001$ \\
3M  & 1.3357 & 0.0465 & $<0.001$ &  0.9008 & 0.2210 & $<0.001$ & 37.260 &  8.090 & $<0.001$ \\
6M  & 1.2291 & 0.0367 & $<0.001$ &  0.3168 & 0.0681 & $<0.001$ & 17.547 &  1.633 & $<0.001$ \\
9M  & 1.1590 & 0.0312 & $<0.001$ & $-0.0309$ & 0.0596 & 0.604  &  8.861 &  0.805 & $<0.001$ \\
12M & 1.1145 & 0.0271 & $<0.001$ & $-0.1554$ & 0.0584 & 0.008  & 3.915 & 0.750 & $<0.001$ \\
24M & 1.0114 & 0.0249 & $<0.001$ & $-0.3623$ & 0.0447 & $<0.001$ & 0.207 & 0.604 & 0.732 \\
\bottomrule
\end{tabular}

\vspace{0.5em}

\begin{tabular}{lcc}
\toprule
Tenor & Wald$(\eta=\zeta=0)$ & $p$(Wald) \\
\midrule
1M  &   7.33 & 0.026    \\
2M  &  15.87 & $<0.001$ \\
3M  &  21.26 & $<0.001$ \\
6M  & 130.7  & $<0.001$ \\
9M  & 121.2  & $<0.001$ \\
12M & 38.17  & $<0.001$ \\
24M &  73.73 & $<0.001$ \\
\bottomrule
\end{tabular}
\caption{Sequential forward-substitution estimates with pairs-bootstrap SEs
($n=1{,}118$, $B=500$). $\zeta$ is the
curvature-transport coefficient estimated from the third-order equation.
The Wald statistic tests the joint null $H_0:\eta=\zeta=0$.
Moving-block bootstrap SEs ($b=5$, same block length as the Wald test)
are on average $16\%$ larger than the pairs-bootstrap SEs reported;
no coefficient changes significance at the $5\%$ level.}
\label{tab:jet-sequential}
\end{table}

\paragraph{Self-similarity test.}
Self-similar transport requires $\eta=0$ and $\zeta=0$. The joint Wald
test of $H_0:\eta=\zeta=0$ rejects at all tenors ($p=0.026$ at 1M,
$p<0.001$ at all other tenors). The two coefficients drive the
rejection in different maturity regimes: $\hat\zeta$ (curvature
transport) is individually significant from 1M through 12M, while
$\hat\eta$ (skew transport) is individually significant at 2M--6M and
again at 12M--24M. Each coefficient dips below individual significance
exactly where it passes through zero: $\hat\eta$ at 9M ($p=0.60$,
near its sign change between 6M and 9M) and $\hat\zeta$ at 24M
($p=0.73$, where it has decayed toward zero). The joint test nevertheless rejects at every tenor. Self-similar transport is therefore rejected
across the full term structure.

At 12M the FS estimate $\hat\eta=-0.155$ is roughly half the slope of
$\hat v_0(k_0)$ fitted across the 12M profile ($\hat\beta_1\approx-0.31$
by OLS on the seven centers of Table~\ref{tab:strike-dep-ssr}); the
discrepancy is consistent with finite-range curvature contamination from
$\hat\zeta=3.92$, which bends the $v_0(k)$ profile over the $\pm10\%$
moneyness range and inflates the apparent linear slope beyond the ATM
derivative $\hat\eta$.

\subsection{Experiment 3: Strike-Dependent Velocity Profile}
\label{subsec:exp3-strike}

The local identification formula~\eqref{eq:identification-v} lets us ask
whether the velocity profile $v(k)$ is actually constant across strikes,
which is precisely the self-similar transport assumption $v(k)\equiv\beta$.
Under self-similarity, the $M$ estimates $\hat{v}_0(k_{0,1}),\ldots,
\hat{v}_0(k_{0,M})$ obtained at different expansion centers should all be
estimating the same number $\beta$. The following Wald test formalizes this.

\begin{proposition}[Strike-constancy test]
\label{prop:chi-square-test}
Self-similar transport ($v(k)\equiv\beta$) is equivalent to
$v_0(k_0)=\beta$ for all $k_0$. Let $\{k_{0,i}\}_{i=1}^M$ be a grid of
expansion centers and let $\hat{v}_0(k_{0,i})$ be the OLS estimator from
Proposition~\ref{prop:identification-v} at each center. Because all $M$
estimators are computed from the same time series of spot returns $\Delta u$
and the same daily polynomial surface fits, they are strongly positively
correlated. Let $\hat\Sigma\in\mathbb R^{M\times M}$ denote the
block-bootstrap covariance matrix of $(\hat{v}_0(k_{0,1}),\ldots,
\hat{v}_0(k_{0,M}))$, estimated by resampling blocks of $b$ consecutive
days. Form the inverse-variance-weighted mean
$\bar\beta = \mathbf{1}^\top\hat\Sigma^{-1}\hat{v}\big/
\mathbf{1}^\top\hat\Sigma^{-1}\mathbf{1}$. Then the Wald statistic
\begin{equation}
    \chi^2
    =
    (\hat{v}-\bar\beta\,\mathbf{1})^\top
    \hat\Sigma^{-1}
    (\hat{v}-\bar\beta\,\mathbf{1})
    \overset{H_0}{\sim}
    \chi^2_{M-1}
    \tag{8.2}\label{eq:chi-square-stat}
\end{equation}
tests $H_0: v(k)\equiv\mathrm{const}$ against the two-sided alternative.
The $\chi^2_{M-1}$ null distribution follows from the joint asymptotic
normality of $\hat{v}$.

The full covariance is essential here: the diagonal version
$\chi^2_{\mathrm{naive}}=\sum_i(\hat v_i-\bar v)^2/\widehat{\mathrm{SE}}_i^2$
is severely conservative under strong positive correlation, since the
marginal SEs carry the common-factor variance while the deviations
$\hat v_i-\bar v$ are small. On SPX 1M the mean pairwise correlation is
$0.984$ and $\chi^2/\chi^2_{\mathrm{naive}}\approx100$.
\end{proposition}

Running the local level-velocity regression at each expansion center
$k_0=\log(K/F)$ independently across the term structure yields the
following. Table~\ref{tab:strike-dep-ssr} reports point estimates at
three representative tenors (1M, 3M, 12M); the complete Wald-test term
structure is in Table~\ref{tab:strike-constancy-ts}.

\begin{table}[H]
\centering
\begin{tabular}{lcc@{\hskip 1.2em}cc@{\hskip 1.2em}cc}
\toprule
 & \multicolumn{2}{c}{1M} & \multicolumn{2}{c}{3M} & \multicolumn{2}{c}{12M} \\
\cmidrule(lr){2-3}\cmidrule(lr){4-5}\cmidrule(lr){6-7}
$K/F$ & $\hat{v}_0$ & SE & $\hat{v}_0$ & SE & $\hat{v}_0$ & SE \\
\midrule
$90.5\%$        & 1.4248 & 0.074 & 1.3903 & 0.055 & 1.1620 & 0.032 \\
$92.8\%$        & 1.4082 & 0.074 & 1.3527 & 0.051 & 1.1401 & 0.030 \\
$95.1\%$        & 1.4019 & 0.078 & 1.3304 & 0.049 & 1.1263 & 0.029 \\
$97.5\%$        & 1.4096 & 0.091 & 1.3237 & 0.048 & 1.1183 & 0.028 \\
$100.0\%$ (ATM) & 1.4375 & 0.124 & 1.3357 & 0.048 & 1.1145 & 0.028 \\
$102.5\%$       & 1.4891 & 0.199 & 1.3689 & 0.052 & 1.1134 & 0.028 \\
$105.1\%$       & 1.5555 & 0.366 & 1.4258 & 0.065 & 1.1126 & 0.029 \\
\midrule
Range (\%)      & \multicolumn{2}{c}{$+11.0\%$} &
                  \multicolumn{2}{c}{$+7.7\%$} &
                  \multicolumn{2}{c}{$-4.4\%$} \\
Profile         & \multicolumn{2}{c}{U-shaped} &
                  \multicolumn{2}{c}{U-shaped} &
                  \multicolumn{2}{c}{decreasing} \\
\bottomrule
\end{tabular}
\caption{Local level velocity $\hat{v}_0(k_0)$ at three representative
tenors, computed by the ratio estimator of
Proposition~\ref{prop:identification-v}. HC3 standard errors; see
Table~\ref{tab:strike-constancy-ts} for Wald test statistics. The rise on the
call side, and the standard errors that grow with it, locate the pole of the
ratio at $k^\ast$; Remark~\ref{subsec:reading-velocity} works this out.}
\label{tab:strike-dep-ssr}
\end{table}

The strike-constancy hypothesis $H_0:v(k,T)\equiv v(T)$ is tested at
each tenor using the block-bootstrap Wald statistic of
Proposition~\ref{prop:chi-square-test} ($B=500$ resamples, block
length $b=5$ days).

Because the seven $k_0$-estimates share the same daily return series,
$\hat\Sigma$ is highly correlated (mean pairwise $\rho=0.908$ at 1M rising to
$0.998$ at 24M) and severely ill-conditioned ($\mathrm{cond}(\hat\Sigma)$
from $2.7\times10^9$ to $3.0\times10^{15}$). Inverting it directly inflates
the full-rank $\chi^2(6)$, so its magnitude is unreliable and the
non-monotone pattern across tenors is conditioning noise rather than signal.
As the primary test we therefore report a \emph{smooth-trend Wald}
projecting the seven-vector onto linear and quadratic contrasts in
$k_0=\log(K/F)$, which gives a well-conditioned $2\times2$ sub-covariance and
a $\chi^2(2)$ insensitive to the near-singular directions.

The price is a narrower alternative, and it changes what a non-rejection
means. Writing
$\hat v=\bar\beta\mathbf 1+c_1\mathbf k+c_2\mathbf k^{(2)}+\mathrm{resid}$
with $\mathbf k,\mathbf k^{(2)}$ collecting $k_{0,i}$ and $k_{0,i}^2$, the
statistic tests $H_0':c_1=c_2=0$. It has power against smooth linear and
quadratic departures and \emph{none} against variation orthogonal to those
contrasts, so an oscillating or locally kinked $v_0(k_0)$ would not register.
Where strike-constancy is reported below as ``not rejected'' at 1M or 24M,
the claim is accordingly that no smooth linear or quadratic trend is
detected, not that $v_0$ is constant in $k$. The statistic that tests
arbitrary nonconstant alternatives is the full-rank $\chi^2(6)$ of \eqref{eq:chi-square-stat}, precisely the one the conditioning renders
unreliable.

\begin{table}[H]
\centering
\small
\begin{tabular}{lcrrrrc}
\toprule
Tenor & Profile & Mean $\rho$
  & $\mathrm{cond}(\hat\Sigma)$
  & \makecell{Full\\$\chi^2(6)$}
  & \makecell{Smooth\\$\chi^2(2)$}
  & $p$ \\
\midrule
1M  & U-shaped    & $0.908$ & $2.7\times10^{9}$  & $109$ & $3.57$ & $0.168$ \\
2M  & U-shaped    & $0.973$ & $9.9\times10^{9}$  & $354$ & $23.1$ & $<0.001$ \\
3M  & U-shaped    & $0.982$ & $2.2\times10^{10}$ & $345$ & $32.7$ & $<0.001$ \\
6M  & U-shaped    & $0.988$ & $2.9\times10^{11}$ & $250$ & $27.5$ & $<0.001$ \\
9M  & U-shaped    & $0.989$ & $1.9\times10^{12}$ & $208$ & $16.5$ & $<0.001$ \\
12M & decreasing  & $0.994$ & $1.3\times10^{13}$ & $200$ & $10.1$ & $0.006$ \\
24M & decreasing  & $0.998$ & $3.0\times10^{15}$ & $91$  & $5.46$ & $0.065$ \\
\bottomrule
\end{tabular}
\caption{Strike-constancy Wald test for $\beta=v_0$,
$H_0:v_0(k,T)\equiv v_0(T)$ ($B=500$ block-bootstrap resamples,
$b=5$ days, $M=7$ strike centers).
\emph{Full $\chi^2(6)$}: standard Wald inverting the full $\hat\Sigma$;
magnitudes are unreliable due to severe ill-conditioning
($\mathrm{cond}(\hat\Sigma)$ rising from $10^9$ to $10^{16}$ across the
term structure).
\emph{Smooth $\chi^2(2)$}: smooth-trend Wald projecting onto
linear and quadratic contrasts in $k_0=\log(K/F)$; this is the primary
reported statistic. It tests the reduced null $c_1=c_2=0$ in
$\hat v=\bar\beta\mathbf 1+c_1\mathbf k+c_2\mathbf k^{(2)}+\mathrm{resid}$ and
so has power only against smooth linear and quadratic departures from
constancy; a non-rejection is not evidence that $v_0$ is constant in $k$.
The ``profile'' column labels the shape of the raw point estimates
$\hat v_0(k_0)$ from OTM puts to OTM calls; its turn from U-shaped to
decreasing with maturity is explained in
Remark~\ref{subsec:reading-velocity}.}
\label{tab:strike-constancy-ts}
\end{table}

The smooth-trend Wald rejects $\beta$ strike-constancy at the
intermediate tenors $2$M--$12$M ($\chi^2(2)$ from $10.1$ to $32.7$,
all $p\le0.006$) but not at $1$M ($p=0.168$) or, marginally, at $24$M
($p=0.065$). The full-rank values ($91$--$354$) are inflated by
conditioning noise and should be read as order-of-magnitude diagnostics
only. Block-length sensitivity ($b=5,10,21$) confirms the full-rank
statistics vary substantially across block lengths at every tenor,
consistent with instability from near-singular inversion. For the
higher-order components, the smooth-trend Wald likewise rejects $\eta$
strike-constancy at nearly all tenors and $\zeta$ strike-constancy at the
intermediate tenors.

\begin{remark}[Why the profile shape changes with maturity]
\label{subsec:reading-velocity}
The ``profile'' column of Table~\ref{tab:strike-constancy-ts} turns from
U-shaped at short tenors to monotonically decreasing at 12M--24M, and
$(5.3d)$ accounts for the switch. Where the source is negative, the ratio
$v+s/\partial_k w$ rises without bound as $k\uparrow k^\ast$ and
$\partial_k w\uparrow0^-$, so a velocity that declines across strikes
acquires an upturn on the call side and presents as U-shaped. The pole sits
at $k^\ast$, which migrates outward with maturity
(Remark~\ref{rem:inversion-fails}), so the upturn lies inside the sampled
band at short tenors and beyond it at long ones, which is the observed pattern. At 1M the band ends at $k_0=\log(1.051)=+5.0\%$ while the skew
vanishes at $+5.2\%$ on average, placing the last expansion center against
the pole; the standard errors in Table~\ref{tab:strike-dep-ssr} grow from
$0.07$ at $90.5\%$ to $0.37$ at $105.1\%$ accordingly. A fixed declining
velocity and a migrating pole therefore reproduce the whole term structure
of profile shapes, which Experiment~5 confirms by estimating $v$ and $s$
together.
\end{remark}

Applying Proposition~\ref{prop:identification-v} at the same seven centers
traces the full local velocity jets at 1M:

\begin{table}[H]
\centering
\small
\begin{tabular}{lcccccc}
\toprule
$K/F$ & $\hat{v}_0(k_0)$ & SE$(v_0)$ & $\hat{v}_1(k_0)$ & SE$(v_1)$ & $\hat{v}_2(k_0)$ & SE$(v_2)$ \\
\midrule
$90.5\%$        & 1.4248 & 0.074 & $-0.8900$ & 0.352 &  17.820 &   2.634 \\
$92.8\%$        & 1.4082 & 0.074 & $-0.4815$ & 0.376 &  21.940 &   3.684 \\
$95.1\%$        & 1.4019 & 0.078 &   0.0572  & 0.447 &  28.556 &   6.657 \\
$97.5\%$        & 1.4096 & 0.091 &   0.8170  & 0.670 &  40.221 &  13.699 \\
$100.0\%$ (ATM) & 1.4375 & 0.124 &   1.9308  & 1.169 &  61.632 &  30.019 \\
$102.5\%$       & 1.4891 & 0.199 &   3.4919  & 2.437 &  98.560 &  76.618 \\
$105.1\%$       & 1.5555 & 0.366 &   5.2949  & 5.094 & 147.532 & 198.251 \\
\bottomrule
\end{tabular}
\caption{Local velocity jet profiles $\hat{v}_0(k_0)$, $\hat{v}_1(k_0)$,
$\hat{v}_2(k_0)$ at 1M. The $\hat v_0$ column is the same regression as the
1M column of Table~\ref{tab:strike-dep-ssr} and carries the same HC3
standard errors; those for $\hat v_1$ and $\hat v_2$ are pairs-bootstrap
($B=500$), as in Table~\ref{tab:jet-sequential}. $\hat{v}_1(k_0)$ is
negative in the OTM put region and increasing toward ATM and calls.
Standard errors for $\hat{v}_1$ and $\hat{v}_2$ grow rapidly off-ATM;
estimates beyond $\pm 5\%$ moneyness should be interpreted with caution.}
\label{tab:local-vel-profile}
\end{table}

\begin{figure}[!htbp]
\centering
\includegraphics[width=\textwidth]{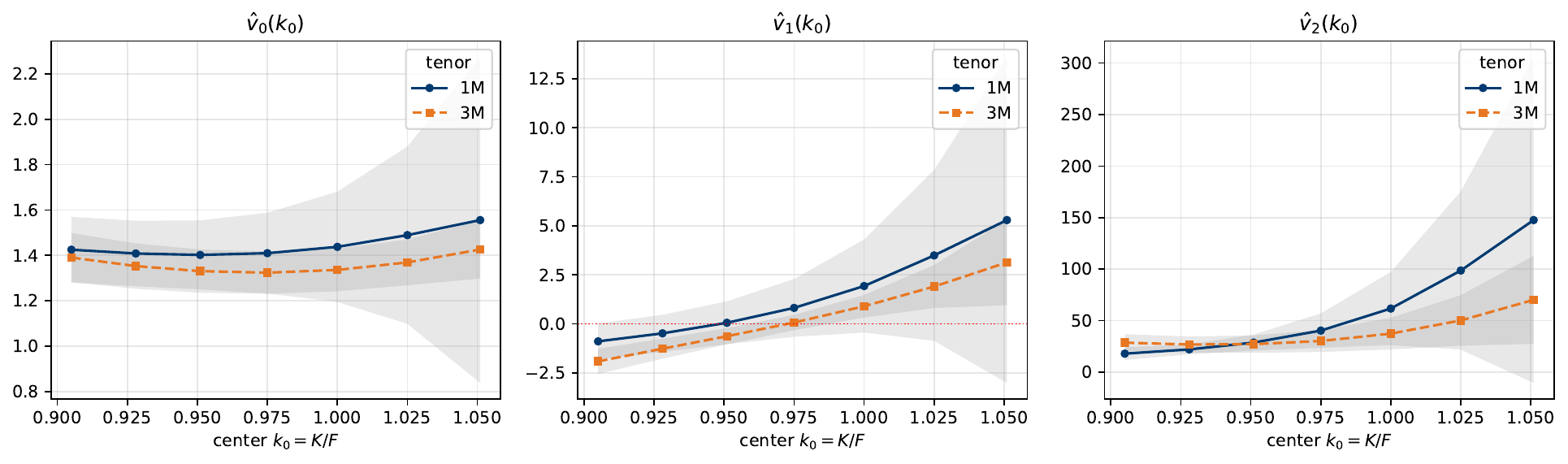}
\caption{Local velocity-jet profiles $\hat{v}_0(k_0)$, $\hat{v}_1(k_0)$,
and $\hat{v}_2(k_0)$ versus expansion center $k_0$ at 1M and 3M.}
\label{fig:spx-exp4}
\end{figure}

Figure~\ref{fig:spx-exp4} plots all three velocity jets against expansion
center $k_0$ at 1M and 3M; the estimates quoted in this paragraph are the
1M panels, and Table~\ref{tab:local-vel-profile} tabulates that tenor.
The smooth-trend Wald test does not formally reject
strike-constancy at 1M ($p=0.168$) due to near-singular $\hat\Sigma$ at
this tenor, though $\hat v_0(k_0)$ is visually non-constant, with the shape
of that variation as Remark~\ref{subsec:reading-velocity} describes.
$\hat{v}_1(k_0)$ shows a clear pattern: negative at the
OTM-put centers ($K/F\le92.8\%$), crossing zero near $94.9\%$, and
positive at ATM and OTM calls. At ATM, $\hat{v}_1(0)=\hat\eta=1.931$,
consistent with Experiment~2. Standard errors for $\hat\zeta$ and
off-ATM $\hat v_1$, $\hat v_2$ are large relative to the estimates
(Table~\ref{tab:local-vel-profile}); beyond $\pm5\%$ moneyness,
higher-order jet coefficients should be interpreted with caution.

\subsection{Experiment 4: Out-of-Sample Predictive Performance}
\label{subsec:exp4-oos}

To test whether $\eta$ and $\zeta$ earn their keep predictively, we run
a 5-fold blocked time-series cross-validation on the aligned daily panel
$(a_0,a_1,a_2,a_3,\Delta u)$ and the full smile grid $w(k_i,T)$ at 11
moneyness points $K/F\in[0.85, 1.15]$. On each training fold we
estimate four nested models by forward substitution; on the held-out
fold we score them on four metrics:
$\mathrm{RMSE}(\Delta a_n)$ for $n=0,1,2$ and the full-smile RMSE
$\sqrt{\frac{1}{11}\sum_j(\widehat{\Delta w}(k_j)-\Delta w(k_j))^2}$
where $\widehat{\Delta w}(k)=\widehat{\Delta a_0}
+\widehat{\Delta a_1}\,k+\tfrac{1}{2}\widehat{\Delta a_2}\,k^2$.
Table~\ref{tab:oos-smile} reports mean OOS RMSE ($\times10^{-4}$)
across the five folds at four tenors.

\begin{table}[H]
\centering
\small
\begin{tabular}{llrrrr}
\toprule
Tenor & Model
  & \makecell{RMSE\\$\Delta a_0$}
  & \makecell{RMSE\\$\Delta a_1$}
  & \makecell{RMSE\\$\Delta a_2$}
  & \makecell{RMSE\\Smile} \\
\midrule
\multirow{4}{*}{1M}
 & Random walk              & 4.690 & 30.59 & 188.5 &  5.393 \\
 & SSR only                 & 2.254 & 22.24 & 220.4 &  3.289 \\
 & SSR + skew               & 2.254 & 21.88 & 281.6 &  3.436 \\
 & Full ($\beta,\eta,\zeta$) & 2.254 & 21.88 & 219.4 &  3.236 \\
\midrule
\multirow{4}{*}{3M}
 & Random walk              & 9.094 & 39.92 & 175.7 &  9.673 \\
 & SSR only                 & 4.360 & 26.07 & 272.8 &  5.274 \\
 & SSR + skew               & 4.360 & 25.12 & 318.7 &  5.332 \\
 & Full ($\beta,\eta,\zeta$) & 4.360 & 25.12 & 215.7 &  5.191 \\
\midrule
\multirow{4}{*}{6M}
 & Random walk              & 13.01 & 39.80 & 215.9 & 13.452 \\
 & SSR only                 &  6.028 & 24.61 & 234.8 &  6.759 \\
 & SSR + skew               &  6.028 & 24.28 & 251.6 &  6.783 \\
 & Full ($\beta,\eta,\zeta$) &  6.028 & 24.28 & 195.5 &  6.761 \\
\midrule
\multirow{4}{*}{12M}
 & Random walk              & 17.88 & 36.97 & 284.1 & 18.148 \\
 & SSR only                 &  8.013 & 24.67 & 199.3 &  8.599 \\
 & SSR + skew               &  8.013 & 24.77 & 195.0 &  8.591 \\
 & Full ($\beta,\eta,\zeta$) &  8.013 & 24.77 & 191.2 &  8.615 \\
\bottomrule
\end{tabular}
\caption{Out-of-sample RMSE ($\times10^{-4}$) from 5-fold blocked
time-series cross-validation on the real SPX panel ($n=1{,}118$ per
tenor). \emph{RMSE(Smile)}: root mean squared error of
$\widehat{\Delta w}(k)$ evaluated at 11 moneyness points
$K/F\in[0.85,1.15]$.}
\label{tab:oos-smile}
\end{table}

By construction $\mathrm{RMSE}(\Delta a_0)$ is identical for the three
transport models: $\beta$ is estimated from the $\Delta a_0$ equation
alone and $\eta$, $\zeta$ do not enter it.
The non-tautological comparisons are on $\Delta a_1$, $\Delta a_2$,
and the full smile.
Adding $\eta$ (SSR+skew and Full) reduces $\Delta a_1$ RMSE by $1$--$4\%$
relative to SSR only at 1M--6M (essentially flat at 12M).
The largest gain is on $\Delta a_2$: the full model reduces RMSE by
$0.5\%$ at 1M (where the small ATM skew $a_1\approx-0.019$ makes $\zeta$ near-degenerate to identify at 1M) and by $17$--$21\%$ at
$3$M--$6$M, narrowing to $\sim4\%$ at $12$M.
SSR+skew worsens $\Delta a_2$ RMSE at short and medium tenors
($+28\%$ at 1M, $+17\%$ at 3M): without $\hat\zeta$ to account for the
dominant $\zeta\,a_1\,\Delta u$ signal, finite-sample noise in $\hat\eta$
propagates through the $2\hat\eta\,a_2\,\Delta u$ term and degrades
curvature predictions.
The $17$--$21\%$ improvement in $\Delta a_2$ RMSE does not translate into
comparable smile RMSE gains, and the reason is simply how the two are
weighted. The smile RMSE decomposes approximately as
\[
    \mathrm{RMSE(Smile)}^2
    \;\approx\;
    \mathrm{RMSE}(\Delta a_0)^2
    \;+\;
    \mathrm{RMSE}(\Delta a_1)^2\,\overline{k^2}
    \;+\;
    \tfrac{1}{4}\,\mathrm{RMSE}(\Delta a_2)^2\,\overline{k^4},
\]
where $\overline{k^2}\approx0.007$ and $\overline{k^4}\approx0.00015$
over the $\pm16\%$ moneyness grid.
The weights give the hierarchy: $\Delta a_0$ contributes $\sim75\%$ of total
smile variance, $\Delta a_1$ about $17\%$, and $\Delta a_2$ only $\sim6\%$,
since curvature errors carry a factor $\tfrac12 k^2\le0.013$ even at the
widest strikes. A $\sim20\%$ reduction in $\Delta a_2$ RMSE therefore buys
only $\sim1.6\%$ on total smile RMSE, exactly as observed. What this amounts
to is a split of the risk: $\beta$ governs \emph{level} risk, which dominates
smile RMSE and which the SSR already captures, while $\eta$ and $\zeta$
govern \emph{shape} risk, which enters with small $k$ and $k^2/2$ weights and
so contributes little to aggregate RMSE yet is material for positions with
vanna and volga. For scale, the random walk is on average $89\%$ worse than
the SSR on the full smile.

Two caveats belong with these numbers. The full-smile gains over SSR are small in absolute terms (at most about $1.6\%$) and are not
uniformly positive: at 12M the full model is marginally \emph{worse} than SSR
alone on the smile metric ($8.615$ against $8.599$ in
Table~\ref{tab:oos-smile}), and the simulated comparison of
Table~\ref{tab:mc-rmse} separates the models by under $0.5\%$ on the same
metric. We therefore make no claim that the extended model dominates the
scalar SSR as a forecaster of the smile. The claim is the narrower one the
decomposition above supports: $\beta$ already captures level dynamics, which
dominate aggregate smile RMSE, while $\eta$ and $\zeta$ improve the \emph{shape} dynamics, materially so for curvature at $3$M--$6$M, where $\mathrm{RMSE}(\Delta a_2)$ falls by $17$--$21\%$. Shape is what matters for positions carrying vanna and volga rather than for aggregate
smile error.

\subsection{Experiment 5: Transport or Source?}
\label{subsec:exp5}
Proposition~\ref{prop:identification-without-inversion} makes the
pure-transport hypothesis testable without dividing by anything. This
experiment reuses the SPX implied-volatility surface data of
Experiments 1--4, restricted to the $797$ dates admitting a stable
local-cubic skew fit at each node. The regressions are run on first
differences, so they use those increments whose two endpoints are both
available and consecutive in the filtered panel; the realized count is
reported per tenor in Table~\ref{tab:pure-transport-fixed} and is smaller
than $797$ for that reason. We use the tenor forward as the regressor
$\Delta u=\Delta\log F(T)$, the daily skew $b_1(k_0,t)$ from a local
cubic fit (half-width 5 nodes) re-gridded onto 41 log-moneyness nodes,
and Newey--West standard errors with 5 lags. The quoted strike range
covers roughly $85$--$115\%$ of spot; at 1M and 3M this comfortably
contains the interior smile minimum on most trading days, but by 6M and
1Y the minimum has migrated outside it (Remark~\ref{rem:inversion-fails}),
so the minimum-based tests below are run at 1M and 3M only. The
at-the-money test does not require the minimum to be in range and is
reported at all four tenors.

\textbf{The fixed-node test.} We fit \eqref{eq:identification-vs} node by
node and test $H_0:s=0$ at two places: at the money, where $|b_1|$ is large
and the diagnostics of Experiments 1--3 are computed, and, at 1M and 3M,
at the node nearest the minimum of the sample-mean smile, where
$\mathcal T_w$ is nearly trivial. Table~\ref{tab:pure-transport-fixed}
reports the result.

\begin{table}[ht]
\centering
\begin{tabular}{lccccccc}
\toprule
 & & \multicolumn{2}{c}{at the money} & \multicolumn{4}{c}{node nearest the smile minimum (1M, 3M only)} \\
\cmidrule(lr){3-4}\cmidrule(lr){5-8}
Tenor & $n$ & $\hat s$ & $t$ & $k_0$ & $\hat s$ & $t$ & $\hat v$ \\
\midrule
1M & 673 & $-0.0044$ & $-0.8$ & $+5.3\%$ & $-0.0139$ & $-4.9$ & $1.09$ \\
3M & 673 & $+0.0142$ & $+1.1$ & $+9.3\%$ & $-0.0246$ & $-5.4$ & $1.16$ \\
6M & 671 & $+0.0506$ & $+3.2$ & \multicolumn{4}{c}{smile minimum outside quoted range} \\
1Y & 671 & $+0.1105$ & $+4.3$ & \multicolumn{4}{c}{smile minimum outside quoted range} \\
\bottomrule
\end{tabular}
\caption{Test of pure transport, $H_0:s=0$ in \eqref{eq:identification-vs},
at fixed nodes ($n$ daily increments; Newey--West, 5 lags). At the money
the source is indistinguishable from zero at 1M and 3M and increasingly
positive at 6M and 1Y. At the node nearest the smile minimum -- run only
at 1M and 3M, where the minimum lies within the quoted strike range for
this panel -- the same regression rejects at both tenors; the tracked
test below confirms this without an estimated skew as a regressor. At 6M
and 1Y the smile minimum lies outside the quoted range, so the
minimum-based test is not estimable there with the present data.}
\label{tab:pure-transport-fixed}
\end{table}

\textbf{The tracked-minimum test.} The fixed node is the minimum of the
\emph{average} smile, not of each day's, so on a given day $b_1$ there is
small rather than zero and Remark~\ref{rem:rejection-meaning} applies.
We therefore repeat the test at each day's actual minimum, at the two
tenors where it lies within the quoted range: on day $t$ we locate the
interior zero crossing $k_t^\ast$ of the daily skew, hold that
log-moneyness fixed over the day, and regress the floating increment
$\Delta w(k_t^\ast,t)$ on $\Delta u_t$ (with the drift control; the
transport regressor $b_1(k_t^\ast,t)\Delta u_t$ is identically zero by
construction and is inert when included). An interior crossing exists on
$785/797$ days at 1M and $739/797$ at 3M, at mean log-moneyness $+5.2\%$
and $+8.8\%$ respectively. At 6M and 1Y no interior crossing exists on
any day, since the minimum lies outside the quoted range, so the tracked
test is not estimable at those tenors. The estimated response at the
minimum is
\[
    \hat s^{\,\mathrm{tracked}}
    =
    \begin{cases}
    -0.0145\ (t=-6.5) & \text{1M,} \\
    -0.0203\ (t=-7.6) & \text{3M.}
    \end{cases}
\]
At a point where the skew is zero on that day, a bounded velocity produces
no response, so this is the source and nothing else. The test involves no
estimated skew as a regressor, so the errors-in-variables channel of
Remark~\ref{rem:rejection-meaning} is closed; the only measurement input
is the location of the minimum, and an error $\delta k$ in it contributes
$v\,\partial_{kk}w\,\delta k\,\Delta u$, second order in the grid spacing.

\textbf{Robustness.} The estimated source at the fixed node is stable to
about $0.001$--$0.002$ across five skew estimators (local cubic with
half-widths 3, 5, 8 nodes, local quadratic, and the raw central difference
with no smoothing) and across grids of 25, 41 and 61 nodes: $\hat s\in
[-0.0157,-0.0138]$ at 1M and $\hat s\in[-0.0258,-0.0243]$ at 3M. This is
sensitivity evidence: it shows that the result does not depend on how
the skew is smoothed, and the tracked test shows that it does not depend on
the skew regressor at all. It does not, and cannot, exclude a common error
in the fitted surfaces themselves, which would have to be correlated with
the forward return to produce the estimate.

\begin{figure}[!htbp]
\centering
\includegraphics[width=\textwidth,trim={0 0 0 23bp},clip]{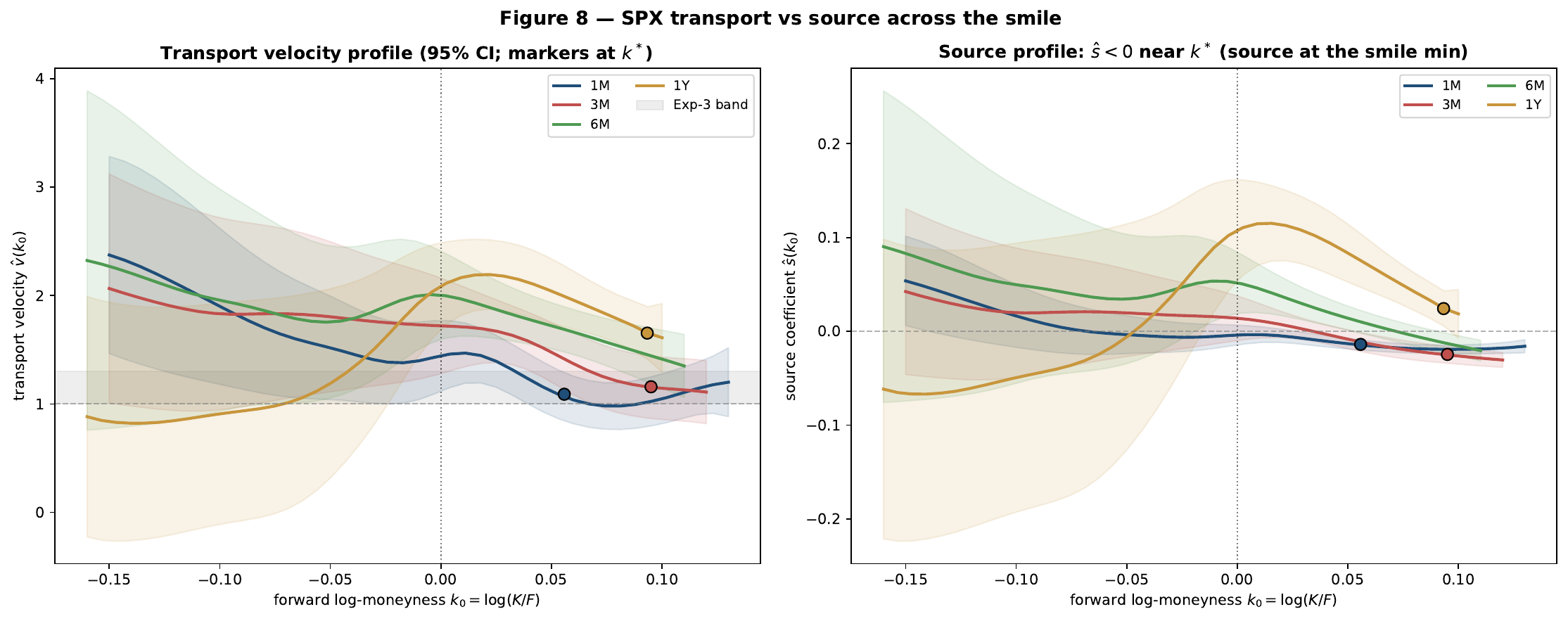}
\caption{Transport velocity and source across the smile (Experiment 5),
all four tenors, with 95\% confidence bands. \emph{Left}: joint-fit
velocity $\hat v(k_0)$; the dashed horizontal line marks sticky-strike
($v=1$) and the shaded region is the range sampled in Experiment 3.
\emph{Right}: the source $\hat s(k_0)$, with a dashed horizontal
reference at $s=0$. Markers show the node nearest the smile minimum at
1M and 3M, the only tenors at which that node lies within the quoted
strike range; no marker is plotted for 6M or 1Y, consistent with
Table~\ref{tab:pure-transport-fixed}'s scope note.}
\label{fig:spx-exp5}
\end{figure}

Figure~\ref{fig:spx-exp5} shows the fitted velocity and source across the
smile at all four tenors. Two conclusions follow, scoped to where the data
support them.

At 1M and 3M, where the smile minimum lies within the quoted range, the joint
fit does not merely sharpen Experiments 2--3: it \emph{reverses} the apparent
sign of the strike dependence. The sequential estimates of
Table~\ref{tab:jet-sequential} give $\hat\eta=+1.93$ at 1M and $+0.90$ at 3M,
and the ratio profiles of Table~\ref{tab:strike-dep-ssr} rise toward the call wing at both tenors, both pointing to a velocity that \emph{increases} with
strike. Estimated jointly with the source, the velocity instead \emph{falls}
from the put wing to the call wing at these tenors, so that puts are
transported faster than calls. The two readings are reconciled by $(5.3d)$ and
by the discussion of Remark~\ref{subsec:reading-velocity}: a negative source
near the smile minimum drives the ratio up on the call side, and coefficients
estimated under the restriction $s\equiv0$ absorb that source into the
velocity, producing an apparent upward strike dependence even where the
underlying velocity is falling. We therefore read the short-tenor
$\hat\eta>0$ of Experiment 2 as source contamination rather than as evidence
on $v$, and the sign from the joint fit as the one to quote. The decline is not monotone: Figure~\ref{fig:spx-exp5} shows the 1M profile high on the put wing, dipping, rising into a shallow hump near the money and falling again. We therefore claim only the net direction, not monotonicity. Bounded pure
transport is rejected away from the money at both tenors, with the rejection
concentrated exactly where the ratio estimator degenerates. At 6M
and 1Y the at-the-money test still shows a source growing increasingly
positive with maturity, but the minimum-based test is not estimable with the
present quoted range, so we claim no rejection at the smile minimum beyond
3M. The summary of SPX smile dynamics at short tenors is therefore a
decomposition rather than a velocity alone: a velocity constant over the
window and declining from puts to calls, plus a source vanishing at the money and
negative at the minimum.

\subsection{Statistical Summary}
\label{subsec:spx-summary}

Table~\ref{tab:stat-summary} consolidates the hypothesis tests.

\begin{table}[H]
\centering
\begin{tabular}{lllcc}
\toprule
Test & $H_0$ & Statistic & $p$-value & Result \\
\midrule
SSR significance (1M)   & $\beta=0$          & $t=11.59$          & $<0.001$ & Reject \\
SSR significance (3M)   & $\beta=0$          & $t=27.83$          & $<0.001$ & Reject \\
SSR significance (24M)  & $\beta=0$          & $t=41.83$          & $<0.001$ & Reject \\
Self-similarity (1M)    & $\eta=\zeta=0$      & Wald$=7.33$        & $0.026$  & Reject \\
Self-similarity (3M)    & $\eta=\zeta=0$      & Wald$=21.26$       & $<0.001$ & Reject \\
Self-similarity (24M)   & $\eta=\zeta=0$      & Wald$=73.73$       & $<0.001$ & Reject \\
Strike-constancy (1M)$^\dagger$   & $v_0(k)=\text{const}$ & $\chi^2(2)=3.57$  & $0.168$  & Not rejected \\
Strike-constancy (3M)$^\dagger$   & $v_0(k)=\text{const}$ & $\chi^2(2)=32.7$  & $<0.001$ & Reject \\
Strike-constancy (12M)$^\dagger$  & $v_0(k)=\text{const}$ & $\chi^2(2)=10.1$  & $0.006$  & Reject \\
Strike-constancy (24M)$^\dagger$  & $v_0(k)=\text{const}$ & $\chi^2(2)=5.46$  & $0.065$  & Not rejected \\
\midrule
Pure transport, ATM (1M)$^\ddagger$        & $s(0)=0$      & $t=-0.8$  & $0.45$   & Not rejected \\
Pure transport, ATM (3M)$^\ddagger$        & $s(0)=0$      & $t=+1.1$  & $0.26$   & Not rejected \\
Pure transport, minimum (1M)$^\ddagger$    & $s(k^\ast)=0$ & $t=-6.5$  & $<0.001$ & Reject \\
Pure transport, minimum (3M)$^\ddagger$    & $s(k^\ast)=0$ & $t=-7.6$  & $<0.001$ & Reject \\
\bottomrule
\end{tabular}
\caption{Summary of hypothesis tests ($n=1{,}118$ for Experiments 1--3).
Self-similarity is tested via the joint Wald test of $H_0:\eta=\zeta=0$.
$^\dagger$Strike-constancy uses the smooth-trend Wald $\chi^2(2)$
(linear and quadratic contrasts in $k_0=\log(K/F)$; see
Table~\ref{tab:strike-constancy-ts} for full details including
$\mathrm{cond}(\hat\Sigma)$ diagnostics). ``Not rejected'' therefore means no
smooth linear or quadratic trend in $k_0$ is detected, not that $v_0$ is
constant in $k$. Tenors 2M, 6M, 9M also reject
at $p<0.001$ (smooth $\chi^2(2)\in[16.5,27.5]$); 24M is marginal
($p=0.065$). $^\ddagger$Experiment 5 (Newey--West $t$-statistics): the ATM
rows are the fixed-node test of Table~\ref{tab:pure-transport-fixed},
reported at all four tenors there; the minimum rows are the
tracked-minimum test of Section~\ref{subsec:exp5}, run only at 1M and 3M,
where the smile minimum lies within the quoted strike range.}
\label{tab:stat-summary}
\end{table}

\begin{figure}[!htbp]
\centering
\includegraphics[width=\textwidth]{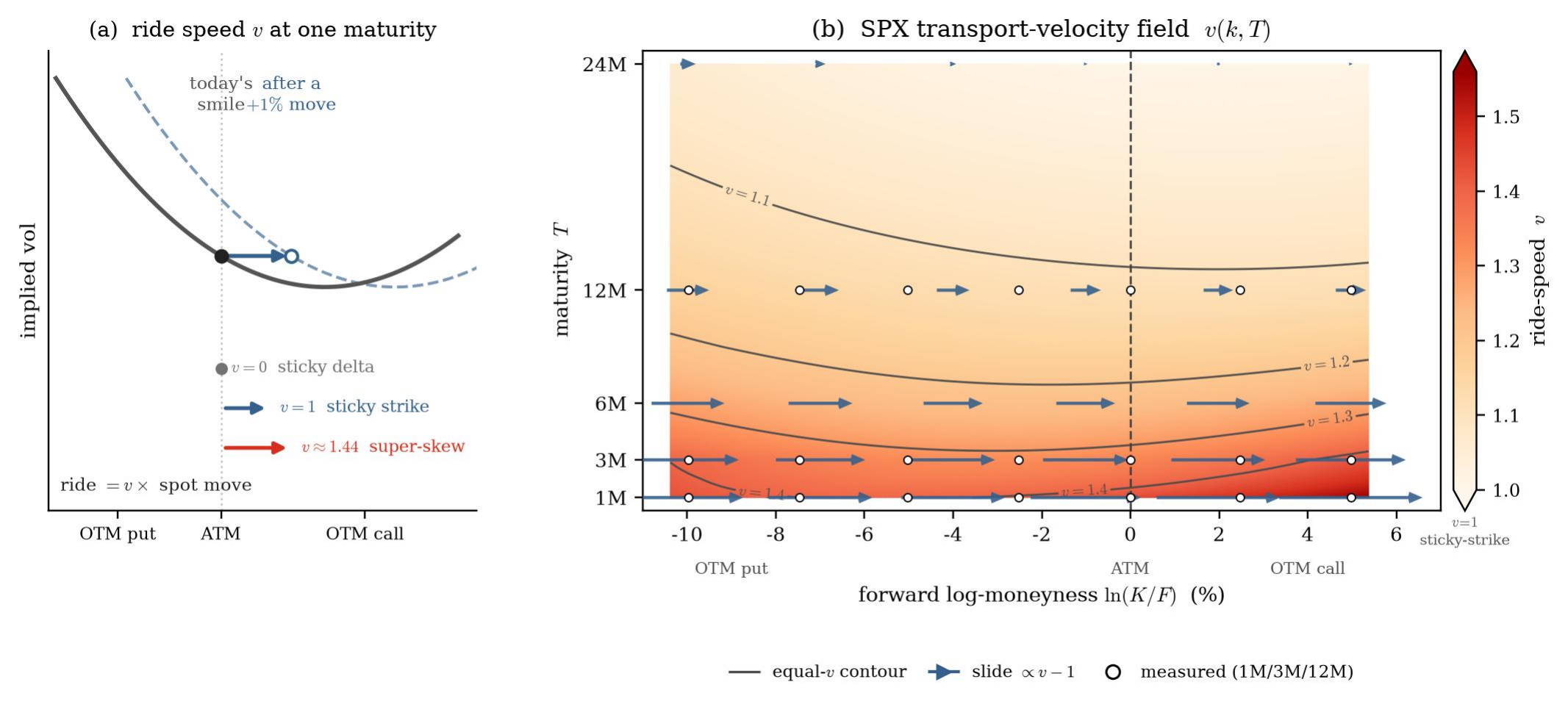}
\caption{The SPX transport velocity field $v(k,T)$.
\emph{(a)}~Ride mechanics: a $1\%$ spot move shifts the smile by $v\%$ in forward log-moneyness, with $v{=}0$ sticky-delta, $v{=}1$
sticky-strike, $v{>}1$ super-skew (SPX 1M: $v{\approx}1.44$).
\emph{(b)}~Estimated field over forward log-moneyness $k{=}\log(K/F)$
and maturity $T$: color encodes ride-speed, arrows mark the
excess over sticky-strike ($v{-}1$), contours are constant-$v$ levels,
dashed line is forward-ATM.
The ATM speed $\beta(T)$ decays toward $1$ at long maturities. The panel
shows the raw ratio estimate of Proposition~\ref{prop:identification-v};
its across-strike shape is diagnosed, not read literally, in
Remark~\ref{subsec:reading-velocity}.
Estimated from Tables~\ref{tab:strike-dep-ssr}
and~\ref{tab:ssr-term-structure}; field interpolated between measured tenors.}
\label{fig:velocity-flow}
\end{figure}

Figure~\ref{fig:velocity-flow} also shows that the raw
ratio $\hat v_0(k_0)$ is neither separable as $v_0(T)\cdot\phi(k)$ nor
monotonic in either dimension over this grid, so a single SSR parameter
$\beta(T)$ captures only its at-the-money value. By
Remark~\ref{subsec:reading-velocity} that non-separability is a property of
the ratio rather than of $v$.

\FloatBarrier
\section{Numerical Validation}
\label{sec:mc-validation}

This section validates finite-sample properties of the estimator under the
assumed data-generating process(DGP); it does not constitute evidence that SPX dynamics follow the
jet-transport model.

\subsection{Estimator Finite-Sample Properties}
\label{subsec:mc-design}
We estimate $(\beta,\eta,\zeta)$ by joint OLS on the lower-triangular system~\eqref{eq:lower-tri-system}; the empirical methodology shares the goal of extracting low-dimensional surface dynamics with Cont and da Fonseca~\cite{ContdaFonseca2002}. The validation is run at
three tenors (1M, 3M, 6M). At each tenor the true parameters
$(\beta,\eta,\zeta)$ are set to the empirical forward-substitution estimates
from Section~\ref{sec:spx-empirical}, and the signal is built on the
\emph{observed} covariate path $(a_1,a_2,a_3,\Delta u)$.
In each of $S=500$ trials the residual triple
$(\varepsilon_0,\varepsilon_1,\varepsilon_2)$ is replaced by a
moving-block resample of the empirical regression residuals
(block length $b=5$, matching the Wald-test block length),
drawn jointly across the three equations to preserve both serial
autocorrelation and cross-equation correlation.
The cubic coefficient $a_3$ enters the $\Delta a_2$ equation as an observed
regressor and is held fixed; its realized heavy-tailed variability propagates
through $\varepsilon_2$.
The full empirical sample ($n=1{,}118$ days) is used, with an inner
pairs-bootstrap ($B=150$) to obtain reported standard errors within each trial.

Table~\ref{tab:mc-recovery} reports parameter recovery for all three tenors.

\begin{table}[H]
\centering
\small
\begin{tabular}{llrrrrrc}
\toprule
Tenor & Param & True & Bias
  & \makecell{Realized\\SE}
  & \makecell{Reported\\SE}
  & \makecell{Emp\\SE}
  & Cov.\ (\%) \\
\midrule
\multirow{3}{*}{1M}
 & $\beta$ & $1.4375$ & $-0.003$ & $0.021$ & $0.020$ & $0.097$ & $97.4$ \\
 & $\eta$  & $1.9308$ & $+0.006$ & $0.139$ & $0.142$ & $1.178$ & $96.0$ \\
 & $\zeta$  & $61.63$  & $-0.199$ & $2.39$  & $2.71$  & $30.94$ & $98.0$ \\
\midrule
\multirow{3}{*}{3M}
 & $\beta$ & $1.3357$ & $-0.002$ & $0.020$ & $0.018$ & $0.047$ & $94.4$ \\
 & $\eta$  & $0.9008$ & $-0.001$ & $0.072$ & $0.067$ & $0.221$ & $94.8$ \\
 & $\zeta$  & $37.26$  & $-0.027$ & $1.160$ & $1.128$ & $8.090$ & $97.0$ \\
\midrule
\multirow{3}{*}{6M}
 & $\beta$ & $1.2291$ & $-0.002$ & $0.018$ & $0.017$ & $0.037$ & $95.2$ \\
 & $\eta$  & $0.3168$ & $-0.005$ & $0.045$ & $0.043$ & $0.068$ & $95.4$ \\
 & $\zeta$  & $17.55$  & $+0.030$ & $0.608$ & $0.603$ & $1.633$ & $94.4$ \\
\bottomrule
\end{tabular}
\caption{Parameter recovery under the block-residual bootstrap DGP
($S=500$ trials, $n=1{,}118$ days, block length $b=5$, inner pairs-bootstrap
$B=150$). True values are the empirical forward-substitution estimates at each
tenor. \emph{Realized SE}: standard deviation of estimates across trials,
conditioning on the observed design. \emph{Reported SE}: median of the
estimator's own pairs-bootstrap SE within each trial. \emph{Emp SE}:
pairs-bootstrap SE from the actual SPX fit. Coverage uses the reported SE.}
\label{tab:mc-recovery}
\end{table}

The estimator is well calibrated at all tenors: realized and reported SE
agree closely and coverage is at or above the nominal $95\%$ throughout. The
empirical SE exceeds the realized SE for $\zeta$ by $12.9\times$ at 1M,
$7.0\times$ at 3M and $2.7\times$ at 6M, which is not a calibration failure
but a difference in what each measures: the block-residual bootstrap
conditions on the observed design and reports
$\mathrm{SE}(\hat\zeta\mid X)$, whereas the pairs bootstrap resamples design
rows and so picks up the variance of the near-zero identifying regressor
$a_1\Delta u$ ($a_1\approx-0.019$ at 1M). The gap compresses as $|a_1|$ grows
with tenor. Empirical $t$-statistics for $\hat\zeta$ are $2.0$, $4.6$ and
$10.7$ at 1M, 3M and 6M: the 3M and 6M rejections are sharp, the 1M one
near-degenerate.

\begin{table}[H]
\centering
\small
\begin{tabular}{lrrrr}
\toprule
Model
  & \makecell{RMSE\\$\Delta a_0$}
  & \makecell{RMSE\\$\Delta a_1$}
  & \makecell{RMSE\\$\Delta a_2$}
  & \makecell{RMSE\\Smile} \\
\midrule
Random walk              & 3.834 & 28.978 & 322.805 & 4.198 \\
SSR only                 & 2.456 & 22.338 & 219.399 & 2.789 \\
SSR + skew               & 2.456 & 22.007 & 240.453 & 2.800 \\
Full ($\beta,\eta,\zeta$) & 2.456 & 22.007 & 205.384 & 2.785 \\
\bottomrule
\end{tabular}
\caption{Out-of-sample RMSE ($\times10^{-4}$) from 5-fold time-series
cross-validation ($S=500$ replications, $n=1{,}118$ days,
DGP calibrated to 1M SPX with $\beta=1.4375$, $\eta=1.9308$,
$\zeta=61.63$).
\emph{RMSE($\Delta a_n$)}: root mean squared error of the predicted
jet coefficient change.
\emph{RMSE(Smile)}: full-smile RMSE evaluated at the seven
$k_0$ expansion centers ($K/F\in\{0.905,\ldots,1.051\}$),
$\Delta\hat w(k)=\widehat{\Delta a_0}+\widehat{\Delta a_1}k
+\tfrac{1}{2}\widehat{\Delta a_2}k^2$.}
\label{tab:mc-rmse}
\end{table}

All three transport models tie on $\Delta a_0$ RMSE (2.456):
$\eta$ and $\zeta$ enter the $\Delta a_1$ and $\Delta a_2$ equations
and leave $\Delta a_0$ unaffected by construction.
Adding $\eta$ reduces $\Delta a_1$ RMSE by $1.5\%$ (SSR+skew and full
vs.\ SSR only).
The full model achieves the lowest $\Delta a_2$ RMSE (205.4),
a $6.4\%$ improvement over SSR only (219.4); notably,
SSR+skew worsens $\Delta a_2$ RMSE to 240.5
because the finite-sample estimation error in $\hat\eta$
propagates through the $2\hat\eta\,a_2\,\Delta u$ term without
$\hat\zeta$ to correct for the dominant $\zeta\,a_1\,\Delta u$ signal.
Full-smile RMSE differences across models are small ($<0.5\%$)
within the $\pm10\%$ moneyness window used here, since
$k^2\le0.01$ makes the $\tfrac{1}{2}\Delta a_2\,k^2$ contribution
negligible relative to $\Delta a_0$ at these strikes.

Table~\ref{tab:mc-power} reports the rejection rate of the Wald test
$H_0:\eta=0$ (nominal 5\%) at four sample sizes.

\begin{table}[H]
\centering
\begin{tabular}{cc}
\toprule
$N$ & Rejection rate \\
\midrule
50  & 0.140 \\
100 & 0.280 \\
200 & 0.460 \\
500 & 0.810 \\
\bottomrule
\end{tabular}
\caption{Power of the Wald test $H_0:\eta=0$ (nominal 5\%) at $\eta=1.931$
(1M SPX calibration), $S=200$ replications.}
\label{tab:mc-power}
\end{table}

Power grows steadily with sample size, reaching $0.810$ at $N=500$, about two years of daily data. The low power at short $N$ reflects the weak
identification of $\eta$ at 1M discussed above: with $|a_1|\approx0.019$ the
identifying regressor $a_1\Delta u$ is small, so the pairs-bootstrap SE is
large relative to $\hat\eta$.

The distribution of the recovered parameters across trials is shown in
Figure~\ref{fig:param-recovery} of
Appendix~\ref{app:simulation-diagnostics}.

\FloatBarrier
\section{Extensions}
\label{sec:extensions}

We record the extensions the results of this paper directly motivate.

\subsection{What the Source Depends On}
Experiment~5 establishes that the response has a component outside the
transport subspace and locates it; it does not say what that component
depends on, and this is the most immediate open question. The SSR is known to
depend on the forward-variance curve~\cite{FrizGatheral2025}, so the natural
conjecture is that $s$ carries the state dependence while $v$ is comparatively stable, so that the decomposition separates a slow component
from a fast state-driven one. Two features of
Table~\ref{tab:pure-transport-fixed} bear on this: the at-the-money source
grows monotonically in magnitude from 1M to 1Y, while at the two tenors where
the smile minimum lies in the quoted range it is decisively negative there
with the velocity staying near $1.1$--$1.2$. A state-dependent model of $s$
would also sharpen the hedging consequence, since by
Remark~\ref{subsec:fixed-strike-response} the spot hedge depends on
$\psi=(v-1)\partial_k w+s$ paired with the book's vega, so the source enters
the delta of any position holding vega where the smile is flat, which for SPX is the call wing at short tenors. Taking $\psi$ rather than $v$ as the primitive
and building the hedging theory on the measured response is the natural next
step, and one we leave to future work.

\subsection{Term-Structure Transport}
The estimated coefficients have a non-trivial term structure: $\hat\beta$
decays monotonically, $\hat\eta$ crosses zero between 6M and 9M, and
$\hat\zeta$ decays toward zero (Table~\ref{tab:jet-sequential}). Here that
maturity dependence is accommodated by letting $v=v(k,T)$ and estimating at
each tenor independently; the dynamics in the maturity direction are not
modeled. The transport equation moves the surface in $k$ only, with $T$ held
fixed along each characteristic. Allowing $T$ to evolve replaces the velocity
by a pair $(v,\lambda)$,
\[
    \dot k_\tau = -v(k_\tau,T_\tau),
    \qquad
    \dot T_\tau = -\lambda(k_\tau,T_\tau),
    \qquad
    \partial_u w = v\,\partial_k w + \lambda\,\partial_T w,
\]
for which a $T$-affine ansatz $v=\gamma_0+\gamma_1 T$ would capture the
observed maturity dependence; this is the implied-variance counterpart of the
term-structure program of Schweizer and Wissel~\cite{SchweizerWissel2008} and
the tangent-model formulation of Carmona and
Nadtochiy~\cite{CarmonaNadtochiy2011}. Establishing admissibility for the
coupled field is harder than the problem solved in
Section~\ref{sec:arb-preservation}, because $\lambda$ enters the
calendar-density equation~\eqref{eq:q-transport} directly, so calendar and
butterfly preservation must be established jointly under the two-dimensional
flow. Compatibility of local jet representations centered at different
strikes is likewise not addressed here.

\subsection{Stochastic Transport}
For $dF_t/F_t=\sigma^F_t\,dW_t$, combining
Proposition~\ref{prop:stochastic-ito} with the transport equation gives
\begin{equation}
    dw(k,T)
    =
    v\,\partial_k w\!\left(\sigma^F_t\,dW_t-\tfrac{1}{2}(\sigma^F_t)^2\,dt\right)
    +\tfrac{1}{2}(\sigma^F_t)^2\bigl(v_u\,\partial_k w
    +v\,v_w\,(\partial_k w)^2
    +v\,\partial_k v\,\partial_k w
    +v^2\,\partial_{kk}w\bigr)\,dt,
    \tag{10.1}\label{eq:stochastic-flow}
\end{equation}
with $v_u,v_w,\partial_k v$ as in Proposition~\ref{prop:stochastic-ito}.
When $v_u\equiv v_w\equiv 0$ this reduces to the admissible second-order
form~\eqref{eq:transport-with-diffusion} with $D^2=(\sigma^F_t)^2v^2\ge0$.
For $w$-dependent $v$ the term
$\tfrac{1}{2}(\sigma^F_t)^2 v\,v_w(\partial_k w)^2$ falls outside that class,
and whether admissibility survives is open
(Remark~\ref{rem:admissible-class}). This places the stochastic-forward
theory within the dynamic-local-volatility framework of Carmona and
Nadtochiy~\cite{CarmonaNadtochiy2009}, and in a simpler setting than the
path-functional framework of Cont and Fourni\'e~\cite{ContFournie2013}, to
which it reduces when $v_w\equiv0$.

\subsection{Admissibility-Preserving Operators and the Geometry of
$\mathcal{A}^\circ$}
\label{subsec:rigorous-geometry}
Two foundational questions remain open. The first is classification: the
operators studied here are of order at most two with non-negative
second-order coefficient (class~\eqref{eq:transport-with-diffusion}), and
while $D^2<0$ is excluded by
Proposition~\ref{prop:backward-parabolic-excluded}, the $N\ge3$ case rests on
the heuristic calculation of Remark~\ref{rem:higher-order-obstruction} and
would require controlling its $O(u^2)$ remainder. A characterization of all
operators $L$ for which $\partial_u w=Lw$ preserves $\mathcal{A}^\circ$ is
not available.

The second is the functional-analytic foundation of the geometric language.
On the compact window of Section~\ref{subsec:geometric-viewpoint} the
situation is already settled, and trivially so: the three strict inequalities
cut out an open subset of $\mathcal{X}=C^{2,1}(Q)$, so the admissible
directions at every surface are all of $\mathcal{X}$ and no manifold
structure is required. Extending to the noncompact domain
$\mathbb R\times\mathbb R_+$ is the substantive step, and it needs a weighted topology (a weighted Sobolev space $H^s_w$, $s>5/2$, for instance) in which the uniform lower bounds are built into the norm; in the Fr\'echet
topology of $C^\infty(\mathbb R\times\mathbb R_+)$ compact-open convergence
does not prevent tail violations of the strict inequalities, so openness
fails. One would then want $\mathcal{D}_w=v\,\partial_k w$ to be a smooth
vector field on $\mathcal{A}^\circ$ and its flow a smooth one-parameter
family of diffeomorphisms, globalizing
Theorem~\ref{thm:local-arb-preserving}. The Filipovi\'c--Teichmann geometry
of yield-curve spaces~\cite{FilipovicTeichmann2004} is the template, with
invariant manifolds and finite-dimensional realizations, with the SVI and SSVI families the natural candidates here, as the objects to look for.

\subsection{Connection to Martingale Optimal Transport}
Each sufficiently regular $w\in\mathcal{A}$ induces, via Kellerer's theorem,
a peacock of risk-neutral marginals and hence a non-empty family of
martingale couplings, so the flows of
Theorem~\ref{thm:local-arb-preserving} are admissibility-preserving
deformations within that family, a dynamic counterpart to the static
robust-pricing framework of Beiglb\"ock, Henry-Labord\`ere and
Penkner~\cite{BeiglbockHenryLaborderePenkner2013}. Making the correspondence
precise is left for future work.

\FloatBarrier
\appendix
\section{Supplementary: Simulation Diagnostics}
\label{app:simulation-diagnostics}
The figure below is the graphical counterpart to
Table~\ref{tab:mc-recovery}. It describes the finite-sample behavior of the
estimator under the assumed data-generating process and carries no
implication about SPX dynamics.

\begin{figure}[!htbp]
\centering
\includegraphics[width=\textwidth]{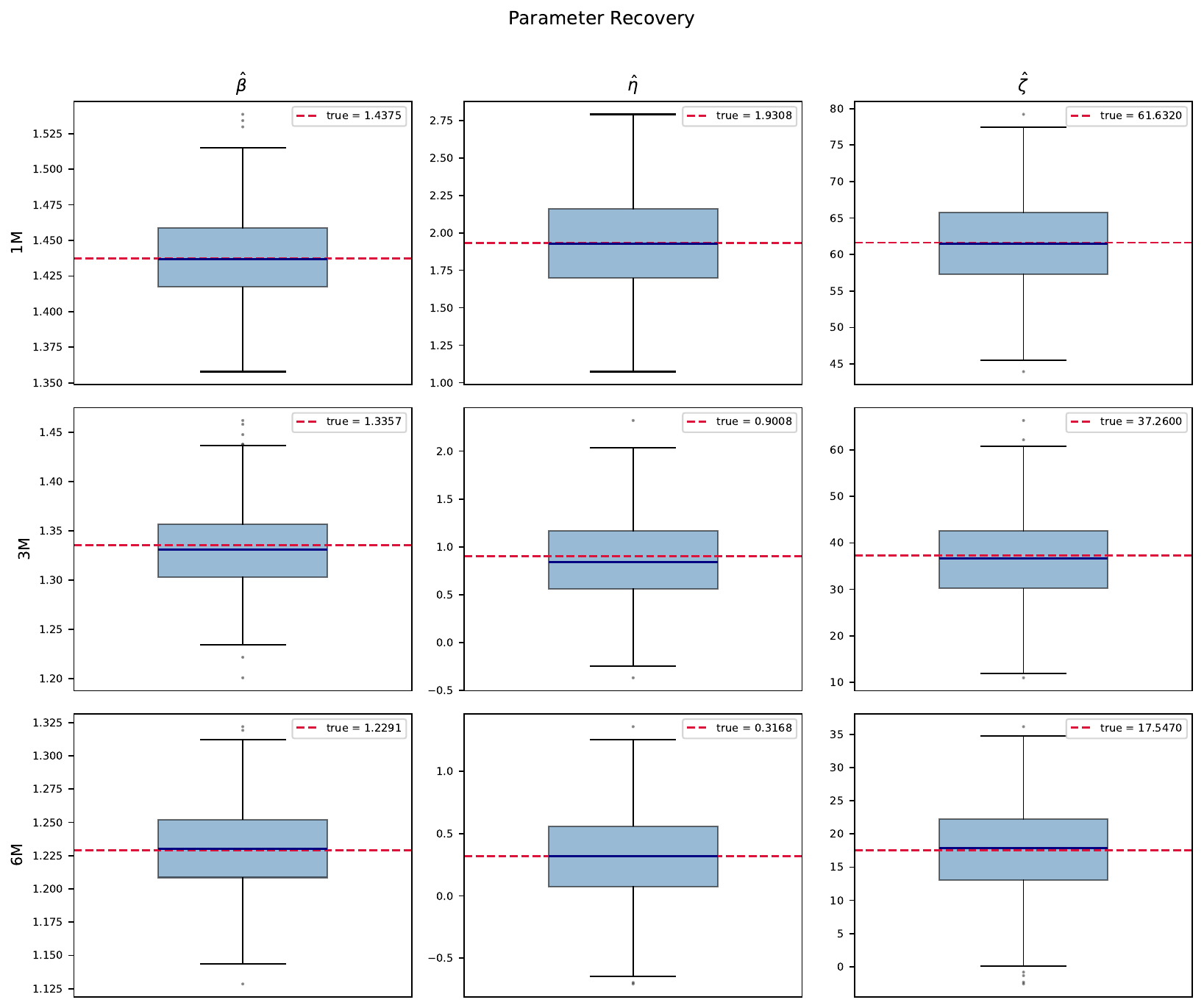}
\caption{Parameter recovery under the block-residual bootstrap DGP
($S=500$ trials, $n=1{,}118$ days, block length $b=5$). Each row
corresponds to a tenor (1M, 3M, 6M); each column to a parameter
($\hat\beta$, $\hat\eta$, $\hat\zeta$). True values (red dashed) are
the empirical forward-substitution estimates at each tenor
($\beta=1.44,\,1.34,\,1.23$; $\eta=1.93,\,0.90,\,0.32$;
$\zeta=61.6,\,37.3,\,17.5$). Box plots show the distribution of
estimates across trials; all parameters are recovered with
negligible bias. Tabulated in Table~\ref{tab:mc-recovery}.}
\label{fig:param-recovery}
\end{figure}

\FloatBarrier

\section*{Disclaimer}
This paper was prepared for informational purposes with contributions from the Quantitative Trading \& Research team of JPMorgan Chase \& Co. This paper is not a product of the Research Department of JPMorgan Chase \& Co. or its affiliates. Neither JPMorgan Chase \& Co. nor any of its affiliates makes any explicit or implied representation or warranty and none of them accept any liability in connection with this paper, including, without limitation, with respect to the completeness, accuracy, or reliability of the information contained herein and the potential legal, compliance, tax, or accounting effects thereof. This document is not intended as investment research or investment advice, or as a recommendation, offer, or solicitation for the purchase or sale of any security, financial instrument, financial product or service, or to be used in any way for evaluating the merits of participating in any transaction.

\end{document}